\documentclass[onecolumn]{IEEEtran}

\usepackage{macro}

\title{Zero-Error Communication over Adversarial MACs}

\author{
\IEEEauthorblockN{
Yihan Zhang\IEEEauthorrefmark{1}\IEEEauthorrefmark{2}
}\\
\IEEEauthorblockA{
\IEEEauthorrefmark{1}Faculty of Computer Science, Technion Israel Institute of Technology \\
\IEEEauthorrefmark{2}Institute of Theoretical Computer Science and Communications, The Chinese University of Hong Kong \\
\href{mailto:yihanzhang@cuhk.edu.hk}{yihanzhang@cuhk.edu.hk}, \href{mailto:zephyr.z798@gmail.com}{zephyr.z798@gmail.com}
}
}

\begin{document}
\maketitle

\begin{abstract}
We consider zero-error communication over a two-transmitter deterministic adversarial multiple access channel (MAC) governed by an adversary who has access to the transmissions of both senders (hence called \emph{omniscient}) and aims to maliciously corrupt the communication. 
None of the encoders, jammer and decoder is allowed to randomize using private or public randomness. 
This enforces a combinatorial nature of the problem. 
Our model covers a large family of channels studied in the literature, including all deterministic discrete memoryless noisy or noiseless MACs. 
In this work, given an arbitrary two-transmitter deterministic omniscient adversarial MAC, we characterize when the capacity region 
\begin{enumerate}
	\item has nonempty interior (in particular, is two-dimensional);
	\item consists of two line segments (in particular, has empty interior);
	\item consists of one line segment (in particular, is one-dimensional);
	\item or only contains $ (0,0) $ (in particular, is zero-dimensional). 
\end{enumerate}
This extends a recent result by Wang, Budkuley, Bogdanov and Jaggi (2019) from the point-to-point setting to the multiple access setting. 
Indeed, our converse arguments build upon their generalized Plotkin bound and involve delicate case analysis. 
One of the technical challenges is to take care of both ``joint confusability'' and ``marginal confusability''. 
In particular, the treatment of marginal confusability does \emph{not} follow from the point-to-point results by Wang et al. 
Our achievability results follow from random coding with expurgation. 
\end{abstract}



\section{Introduction}
\label{sec:intro}
The multiple access channel (MAC) model was first (implicitly) considered by Shannon \cite{shannon-1961-twoway}. 
This model is arguably one of the simplest communication models beyond the point-to-point setting. 
The problem concerns information transmission over a three-node network. 
Two\footnote{In this paper, we only consider MACs with two transmitters. Generalizations to more transmitters are left as an open question (see \Cref{itm:open-multiuser-mac} in \Cref{sec:concl-rmk-open-prob}).} independent senders simultaneously send signals to the channel; a single receiver aims to recover both senders' transmitted messages given the channel-distorted signal. 
The goal for the parties in such a communication scenario is to reliably deliver as much information from the senders to the receiver. 
The fundamental limits (i.e., \emph{capacity region}, see \Cref{def:ach-rate-cap-region}) of discrete memoryless MACs under the average error criterion was derived independently by Ahlswede \cite{ahlswede-1973-mac-cap,ahlswede-1974-mac} and Liao \cite{liao-1972-mac-cap-thesis}\footnote{The capacity region given by Ahlswede \cite{ahlswede-1973-mac-cap,ahlswede-1974-mac} and Liao \cite{liao-1972-mac-cap-thesis} is written in terms of the convex hull of the union of multiple regions. An alternative form involving an auxiliary time-sharing variable was given by Slepian and Wolf \cite{slepian-wolf-1973-mac}. A cardinality bound on the alphabet of the auxiliary variable was given in \cite{csiszar-korner-book2011}.}. 
The Gaussian counterpart\footnote{This paper only concerns MACs with finite-sized alphabets and will {not} deal with the Euclidean case.} was solved by Cover \cite{cover-1975-gaussian-mac} and Wyner \cite{wyner-1974-gaussian-mac}. 
MACs are so far the essentially only multiuser channel whose fundamental limits are well-understood in full generality.

In the classical Shannon's setup of the MAC problem, it is assumed that the channel is given by a \emph{fixed} (i.e., time-invariant) law\footnote{We use lowercase boldface letters to denote (scalar) random variables.} $ W_{\bfy|\bfxa,\bfxb} $ that maps a given pair of input symbols\footnote{Throughout this paper, we use superscripts to denote the indices of the transmitter. E.g., $ \xa$ (resp. $\xb$) denotes a symbol transmitted by the first (resp. second) transmitter.} $ (\xa,\xb)\in\cXa\times\cXb $ to an output symbol $ y\in\cY $ with probability $ W_{\bfy|\bfxa,\bfxb}\paren{y\condon\xa,\xb} $. 
Such a channel well models white noise between the senders and the receiver, while it fails to model \emph{adversarial} noise that is potentially injected by a malicious adversary. 
In this paper, we take a coding-theoretic perspective on multiple access. 
A general \emph{omniscient adversarial MAC} model is introduced and studied. 
We assume that the channel is governed by an adversary who has full access to the transmitted signals from both senders (hence called \emph{omniscient}). 
The adversary aims to prevent communication from happening by transmitting a carefully designed noise sequence to the channel. 
We therefore at times also call the adversary the \emph{jammer}. 
None of the encoders, the jammer and the decoder is allowed to randomize. 
To enforce a combinatorial nature of the problem, it is further assumed that the channel obeys a zero-one law, i.e., the distribution $ W_{\bfy|\bfxa,\bfxb,\bfs} $ (where $ \bfs $ denotes the symbol sent by the jammer) only takes values in $ \zo $ and can be realized by a deterministic function $ y = W(\xa,\xb,s) $ (with a slight abuse of notation). 
The main contribution of this paper is a \emph{zero-th} order (see the next paragraph) characterization of the capacity region of an arbitrary omniscient adversarial MAC with \emph{maximum} error probability. 
In fact, since nothing in the system is stochastic, it is not hard to see that maximum error criterion is equivalent to zero error criterion. 
Our results can be appreciated through different lenses, e.g., arbitrarily varying channels, zero-error information theory, coding theory, etc. 
Elaboration on various connections is deferred to \Cref{sec:related-work}. 

Classical Shannon theory and combinatorial coding theory provide systematic ways of studying the \emph{first-order} asymptotics, i.e., capacity, of (stochastic and adversarial respectively) communication channels. 
By first-order we mean the number of bits that can be reliably transmitted through the channel. 
The first-order asymptotics of discrete memoryless channels (DMCs) are well-established in the seminal paper by Shannon \cite{shannon-1948} which laid the foundation of information theory. 
The first-order asymptotics of most multiuser channels remain open, except for MAC as mentioned before and a handful of other special cases. 
On the other hand, in the theory of error-correcting codes which deals with worst-case errors, essentially no capacity is characterized for any nontrivial channel. 
Indeed, even the capacity of adversarial bitflip channels -- one of the simplest nontrivial channels remains a holy grail problem in coding theory. 
This problem is well known to be equivalent to the sphere packing problem in binary Hamming space.
Our work can be viewed as a first step towards pushing the existing wisdom of classical coding theory to the general multiuser setting. 
For one thing, we consider very general channel models, not just the bitflip channel which is the most studied one in coding theory. 
For another thing, we go beyond the point-to-point setting and consider MACs. 
Due to the lack of techniques for characterizing the capacity, this work only aims to characterize the ``shape'' of the capacity region of any given adversarial MAC. 
More specifically, we determine the dimension of the capacity region -- when it has nonempty interior; when it only consists of (one or two) line segment(s); and when it only contains $ {(0,0)} $. 
We call such positivity conditions a characterization of the \emph{zero-th} order asymptotics of the channel. 
See \Cref{sec:our_results} for the formal statements of our results. 
Finally, we remark that there has been a stream of work on high-order (second-/third-/fourth-order) asymptotics of channels \cite{polyanskiy-poor-verdu-2010-finite-bl,tomamichel-tan-2013-third-dmc,tan-tomamichel-2015-third-awgn,scarlett-2014-second-mac,yavas-2020-second-gaussian-mac,kosut-2020-second-mac}.

\begin{remark}
\label{rk:error-criteria}
The capacity region of a (non-adversarial) MAC under average error criterion can be achieved using deterministic encoding and the region is invariant even if stochastic encoding is allowed. 
However, unlike the point-to-point case, under \emph{maximum} error criterion and \emph{deterministic} encoding, the capacity region of a MAC is strictly smaller than that under average error criterion \cite{dueck-1978-max-vs-avg-mac}. 
To the best of our knowledge, the exact capacity region in this case is still open. 
Furthermore, under maximum error criterion, stochastic encoding can achieve the capacity region with average probability of error. 
This shows that randomization at the encoders can boost the capacity under maximum error criterion -- a phenomenon absent in the point-to-point setting. 
\end{remark}

\section{Related work}
\label{sec:related-work}
Our model and results are connected to various facets of information theory and adjacent fields. 
We list non-exhaustively several connections below and {compare}, when proper, our results with existing ones.

\subsection{Arbitrarily varying channels}
\label{sec:related-avc}
Our model of general omniscient adversarial MAC is intimately related to a classical model studied in the literature known as the \emph{arbitrarily varying channel (AVC)}. 
An AVC is a channel with a state $ \bfs $ that does not follow any fixed distribution, i.e., is arbitrarily varying. 
A noticeable difference between the classical AVC model and our model is that the bulk of the literature on AVC deals with channels with an \emph{oblivious} adversary who does not know anything about the transmitted sequence. 
Under average error criterion, this problem is significantly easier (though not trivial) than the omniscient counterpart. 
Indeed, the fundamental limits of point-to-point AVCs \cite{csiszar-narayan-it1988-obliviousavc,csiszar-narayan-1991-gavc} and arbitrarily varying MACs (AVMACs) \cite{ahlswedecai-1999-obli-avmac-no-constr,pereg-steinberg-2019-avmac} (and several other channels which we do not spell out here) are well-understood. 

In fact, an oblivious AVMAC with maximum probability of error is equivalent to our model of omniscient adversarial MAC. 
However, the maximum error criterion is much less studied in the AVC literature. 
Obtaining a tight first-order characterization of the capacity remains an formidable challenge even for very simple channels. 
The main focus of this work is a zero-th order characterization of the capacity region of general omniscient adversarial MACs.
Though we do present nontrivial inner and {outer bounds}, there is no reason to expect any of them to be optimal. 
\Cref{itm:open-error-criterion} in \Cref{sec:concl-rmk-open-prob} contains more discussions and open problems regarding error criterion. 
See also \Cref{sec:comparison-our-peregsteinberg} for an in-depth comparison between our work and \cite{pereg-steinberg-2019-avmac} on AVMACs.

\subsection{Zero-error information theory}
\label{sec:related-zero-error-it}
Since randomization in the encoding/jamming/decoding strategies are ruled out from our model and only deterministic channels are considered, there is no probability anywhere in the system and maximum error criterion is equivalent to zero error criterion. 
For this reason, it is worth mentioning the connections between our work and zero-error information theory -- a combinatorial facet of information theory. 
The basic deviation of zero-error information theory from ordinary Shannon theory is to insist on \emph{zero error} criterion which changes the nature of the problem in a fundamental way.
Despite of years of research, there is essentially no capacity result for any general channel model except for sporadic special channels \cite{lovasz-1979-shannon-cap}.
Usually channels studied in zero-error information theory do not consist of an adversarial noise (a.k.a. an arbitrarily varying state in AVC jargon). 
It turns out that if the adversarial noise in our model is \emph{unconstrained} (i.e., the state vector\footnote{We use underlines to denote vectors of length $n$ -- the number of channel uses. See \Cref{sec:notation} for notational conventions of this paper.} $ \vs $ can take any value in $ \cS^n $), then the channel is equivalent to a non-adversarial channel under zero error criterion. 
On the other hand, the presence of state constraints brings significant effect on the behaviour of the channel. 
Such a phenomenon already shows up in the point-to-point setting \cite{csiszar-narayan-it1988-obliviousavc}. 
Classical zero-error information theory approaches the problem of zero-error communication via the notion of \emph{Shannon capacity} of graphs \cite{shannon-1956-zero-error} -- getting rid of channel probabilities.\footnote{Unfortunately, Shannon capacity is not computable since it is defined as a limit as $n$, the blocklength, goes to infinity. See \Cref{sec:related-nonstoc-it} and \Cref{itm:open-tensorization} in \Cref{sec:concl-rmk-open-prob} for remarks on $n$-letter capacity expressions.}
Recently, the positivity of zero-error capacity of MACs (and several other multiuser channels) was characterized by Devroye \cite{devroye-2016-zero-positive}. 
However, she only dealt with non-adversarial channels, or equivalently, adversarial channels without state constraints. 
Several other general multiuser channels with zero error such as two-way channels \cite{gu-shayevitz-2019-twoway-zeroerror} and relay channels \cite{chen-2014-relay-zero-error-1,chen-2015-relay-zero-error-conf,chen-2017-relay-zero-error-jrnl,asadi-2018-relay-zero-error} were also studied in the literature. 
Many other works on zero-error multiuser channels concentrate around specific channels such as binary adder MAC \cite{austrin-2017-binary-adder-mac}, $ \mathsf{AND}\text{-}\mathsf{OR} $ interference channel \cite{nair-2020-andor}, etc. 
See \Cref{sec:related-specific} for more related work on special MACs. 

\subsection{Kolmogorov complexity}
\label{sec:related-kolmogorov-cplx}
Besides Shannon's notion of graph capacity, Kolmogorov \cite{kolmogorov-1956,tikhomirov-kolmogorov-1993-eps-cap} introduced the $ \eps $-entropy and $ \eps $-capacity (which are the normalized covering and packing number (using balls of radius $ \eps $) of a space) as another non-stochastic approach to zero-error source and channel coding, respectively. 
However, there was no coding theorems companying these notions. 
The results in \cite{wbbj-2019-omni-avc} which we build upon can be cast as packing \emph{general} shapes (not necessarily balls) without overlap in a general space. 
For MACs, the geometric interpretation of packing and covering does not seem to be as obvious/clean as in the point-to-point case.

\subsection{Non-stochastic information theory}
\label{sec:related-nonstoc-it}
Recently, Nair \cite{nair-2011-nonstoc-conf,nair-2013-nonstoc-jrnl} proposed yet another alternative framework towards understanding zero-error communication known as \emph{non-stochastic information theory}. 
He introduced non-stochastic analogs of information measures and proved coding theorems for worst-case error models. 
Extensions to MACs (see \cite{zafzouf-nair-evans-2019-nonstoc-mac} for the two-transmitter case and \cite{zafzouf-nair-2020-nonstoc-mac-multiuser} for the multi-transmitter case), channels with feedback \cite{nair-2012-nonstoc-feedback,saberi-2018-zeroerror-erasure,saberi-2020-zeroerror-feedback-additive}, channels with memory \cite{saberi-2020-zeroerror-state,saberi-2019-zeroerror-memory-eras-symm} and function evaluation \cite{farokhi-nair-2020-nonstoc-fn-eval} are presented in followup works by Nair and his coauthors. 
In most cases, Nair's framework only gives $n$-letter expressions for capacity, similar to the graph-theoretic approach mentioned in \Cref{sec:related-zero-error-it}. 
More recently, Lim--Franceschetti \cite{lim-franceschetti-2017-non-stoc-it} and Rangi--Franceschetti \cite{rangi-franceschetti-2019-non-stoc-it} refined Nair's framework by introducing new non-stochastic information measures to incorporate decoding errors while retaining the worst-case nature of the error model. 
The latter work \cite{rangi-franceschetti-2019-non-stoc-it} also studied the possibility of obtaining single-letter expressions for the capacity of a certain family of channels.

As a comparison, our approach does not even yield $n$-letter capacity expressions. 
However, we can handle general adversarial channels with potentially constrained adversarial noise. 
In \cite{rangi-franceschetti-2019-non-stoc-it}, following Nair's framework, such channels are treated as \emph{nonstationary} channels with \emph{memory} for which no $n$-letter capacity expression was obtained. 
More words on $n$-letter expressions can be found in \Cref{itm:open-tensorization} of \Cref{sec:concl-rmk-open-prob}.

\subsection{Coding theory and generalized Plotkin bound}
\label{sec:related-wbbj}
Since our problem inherently exhibits a combinatorial nature, one can view our contributions as Shannon-theoretic results for a coding-theoretic model. 
We borrow insights and techniques from both information theory and coding theory and try to build a bridge between them in the particular MAC setting. 
At a technical level, the principal tool that we use is inspired by a recent Plotkin-type bound for general point-to-point omniscient adversarial channels \cite{wbbj-2019-omni-avc}. 
Our contribution is to generalize it to the MAC setting and use it, along with delicate case analysis, to characterize the ``dimension'' of the capacity region. 
The results in both \cite{wbbj-2019-omni-avc} and this paper are in turn generalizations of the Plotkin bound in classical coding theory. 
This bound (together with a standard probabilistic construction) pins down the exact threshold of the noise level of a bitflip channel\footnote{A bitflip channel takes a binary sequence as input and arbitrarily flips a fixed fraction of bits.} such that positive rates are achievable (see \Cref{def:ach-rate-cap-region} for the formal definition of achievable rates). 


\subsection{Specific channels}
\label{sec:related-specific}

Our model covers a large family of channels studied in the literature, including the $\OR$ MAC, the collision MAC, the adder MAC \cite{gu-2018-zero-error-mac,austrin-2017-binary-adder-mac}, the disjunctive MAC \cite{d-2019-separable-list-dec-mac}, the multiple access hyperchannel \cite{shchukin-2016-list-dec-mac}, etc. 
Indeed, our model incorporates all deterministic channel models. 
Interested readers are encouraged to refer to the lecture notes \cite{lec-notes-mac} and \cite[Chapter 29, 30]{polyanskiy-wu-2014-lec-notes-it}. 




\section{Overview of our results}
\label{sec:overview-results}

This work initiates a systematic study of memoryless MACs in the presence of an omniscient adversary (who may \emph{not} behave memorylessly) under the maximum probability of error criterion. 
In particular, the main attention of this paper is focused on the capacity threshold. 
In what follows, we summarize the contributions of this paper.

\begin{enumerate}
	\item We introduce in \Cref{sec:basic_def} the model of \emph{omniscient adversarial MACs} which covers a large family of channels of interests. 
	In particular, all component-wise deterministic memoryless channels with finite alphabets fall into our framework. 
	In this work we focus on the maximum probability of error criterion. 
	For technical reasons, we make additional assumptions that are listed in \Cref{sec:additional-technical-assump}.

	\item We introduce in \Cref{sec:conf-set} the notion of \emph{confusability}, both the operational version (\Cref{claim:operational-nonconf}) and the distributional version (\Cref{def:conf-set}) which turn out to be equivalent (\Cref{claim:distributional-nonconf}, \Cref{rk:equiv-operational-distributional}). 
	Specifically, we define the \emph{joint confusability set} and the (first and second) \emph{marginal confusability sets} (for both transmitters separately) to capture the disability to reliably transmit both (for the joint case) or exactly one (for the marginal cases) of the sequences. 
	One can think of the confusability sets as the sets of ``bad'' distributions that (the types\footnote{The type of a (collection of) vector(s) is the empirical distribution/histogram. See \Cref{def:type} for a formal definition.} of) any good code should avoid.
	The significance of the notion of confusability is that it precisely captures all information one needs for understanding the capacity region of any adversarial MAC. 
	In fact, adversarial MACs with the same confusability sets share a common capacity region (\Cref{lem:cap-determined-by-conf-set}), though they may appear different at the first glance.
	Various properties of the confusability sets are presented in \Cref{prop:prop-conf-set}.

	\item Towards understanding capacity thresholds, we find a class of distributions that we call \emph{good} (\Cref{def:good_distr}). 
	Again, they are separately tailored for the joint case and two marginal cases. 
	While being of independent interest on their own, the sets of good distributions are particularly useful in our context of determining the capacity threshold. 
	One should think of these classes of distributions as the \emph{only} types of distributions that one needs to consider for the purpose of achieving positive rates (though in this way one may not be able to achieve the capacity which is anyway unknown given the current techniques). 
	We also define a cone of tensors referred to as \emph{co-good} tensors (\Cref{def:cogood-distr}) and show that the cones of good and co-good tensors are dual to each other (\Cref{thm:duality}), which will be critical to the proofs in the proceeding sections.
	Various properties of good distributions and co-good tensors are presented. 
	We expect these distributions/tensors and the associated duality to be useful elsewhere. 

	\item 
	We completely characterize, for any given omniscient adversarial MAC, the ``shape'' of the capacity region,
	that is, when the capacity region
	\begin{enumerate}
		\item has nonempty interior (in particular, is two-dimensional);
		\item consists of two line segments (in particular, has empty interior);
		\item consists of one line segment (in particular, is one-dimensional);
		\item or only contains $ (0,0) $ (in particular, is zero-dimensional). 
	\end{enumerate}
	The proof comprises of the direct part and the converse part.
	The technically most challenging case is to handle the (non-)achievability of rate pairs both components of which are strictly positive. 
	For the marginal cases, we emphasize that they do \emph{not} follow from the point-to-point results in \cite{wbbj-2019-omni-avc} in a black-box manner.

\end{enumerate}

We then briefly discuss separately our achievability and converse results and the techniques for proving them.
For a more detailed discussion on the proof techniques, see \Cref{sec:overview-techniques}.
\begin{enumerate}
	\item For the achievability part, one could use good non-confusable distributions (whenever they exist) to sample good codes of positive rates (\Cref{thm:achievability}). 
	This follows from the standard random coding argument which in turn is proved using Chernoff-union bounds. 
	We also strengthen the above positivity results by giving \emph{inner bounds} on the capacity region (\Cref{lem:inner_bound_prod_distr}). 
	This follows by carefully expurgating the codes and analyzing the large deviation exponents of the error events using the Sanov's theorem (\Cref{lem:sanov}). 
	The most challenging case is where both transmitters are able to achieve positive rates.

	\item 
	On the other hand, for the converse part, if one cannot construct positive rate good codes using good distributions, then she/he cannot construct them using any other types of distributions (\Cref{thm:converse}). 
	This part is much less obvious and forms the bulk of the technically most challenging portion of this work. 
	As alluded to above, 
	the crux of the proof is to leverage the duality between the cone of good distributions and the cone of co-good tensors defined before and to apply a double counting trick that is reminiscent of the one used in the classical Plotkin bound in coding theory. 
	Technically, to make the trick actually work, we have to preprocess the code by applying a standard constant composition reduction and an equicoupled subcode extraction (using Ramsey's theorems \Cref{lem:hypergraph_ramsey,thm:regular-ramsey}). 
	The hardest case is to show that two transmitters cannot simultaneously achieve positive rates as long as there does not exist a distribution that is \emph{simultaneously} jointly good and (first and second) marginally good. 
\end{enumerate}

\section{Organization of this paper}
\label{sec:org}
The rest of the paper is organized as follows.
Notational conventions of this paper are listed in \Cref{sec:notation}, followed by preliminaries in \Cref{sec:prelim}. 
We formally introduce the omniscient adversarial MAC model in \Cref{sec:basic_def}.
Before proceeding, we first study the special case of binary noisy $ \XOR $ MACs in \Cref{sec:warmup_examples} with proofs deferred to \Cref{app:pf_plotkin_binary_noisy_xor_mac}. 
Then in \Cref{sec:conf-set,sec:good_distr} respectively, we introduce two important notions of (sets of) distributions, viz.: the confusability sets and the sets of good distributions, and prove properties of them. 
Building on the machinery we have developed in the previous sections, the main result  (\Cref{thm:our-results}) of this paper, i.e., a characterization of the ``shape'' of capacity region, is formally stated in \Cref{sec:our_results}. 
Before presenting the detailed proofs, we outline a roadmap with underlying ideas of the proofs in \Cref{sec:overview-techniques}. 
\Cref{sec:achievability} contains a full proof of the achievability part of our main theorem. 
\Cref{sec:conv-pos-pos,sec:conv-pos-zero} prove the ``joint'' case and the ``marginal'' cases of the converse part, respectively. 
We conclude the paper with a list of remarks and open questions in \Cref{sec:concl-rmk-open-prob}. 
A table of frequently used notation can be found in \Cref{app:table-of-notation}.

\section{Notation}
\label{sec:notation}
Sets are denoted by capital letters in calligraphic typeface, e.g., $ \cX,\cS,\cY $, etc. 
All alphabets in this paper are finite sized. 
For a positive integer $ M $, we use $ [M] $ to denote $ \curbrkt{1,\cdots, M} $. 
Let $\cX$ be a finite set. 
For an integer $ 0\le k\le\cardX $, we use $ \binom{\cX}{k} $ to denote $ \curbrkt{\cX'\subseteq\cX\colon \card{\cX'}= k} $. 

Random variables are denoted by lowercase letters in boldface, e.g., $\bfx,\bfs,\bfy $, etc. 
Their realizations are denoted by corresponding lowercase letters in plain typeface, e.g., $x, s, y$, etc. Vectors (random or fixed) of length $n$, where $n$ is the blocklength of the code without further specification, are denoted by lowercase letters with  underlines, e.g., $\vbfx,\vbfs,\vbfy,\vx,\vs,\vy$, etc. 
The $i$-th entry of a vector $\vx\in\cX^n$ (resp. $\vbfx\in\cX^n$) is denoted by $\vx(i)$ (resp. $\vbfx(i) $). 

For vectors and random variables/vectors, we use superscripts to denote the indices of the transmitters, e.g., $ \vxa,\bfxa,\vbfxa $ (resp. $\vxb,\bfxb,\vbfxb$) correspond to the first (resp. second) transmitter.

We use the standard Bachmann--Landau (Big-Oh) notation. 
For two real-valued functions $f(n),g(n)$ of positive integers, we say that $f(n)$ \emph{asymptotically equals} $g(n)$, denoted by $f(n)\asymp g(n)$, if 
$\lim_{n\to\infty}{f(n)}/{g(n)} = 1$.
We write $f(n)\doteq g(n)$ (read $f(n)$ \emph{dot equals} $g(n)$) if
$\lim_{n\to\infty}\paren{\log f(n)}/\paren{\log g(n)} = 1$.
Note that $f(n)\asymp g(n)$ implies $f(n)\doteq g(n)$, but the converse is not true.
For any $\cA\subseteq\cX$, the indicator function of $\cA$ is defined as, for any $x\in\cX$, 
\[\one_{\cA}(x)\coloneqq\begin{cases}1,&x\in \cA\\0,&x\notin \cA\end{cases}.\]
At times, we will slightly abuse notation by saying that $ \indicator{\sfA} $ is $1$ when event $\sfA$ happens and $0$ otherwise. 
Note that $\one_{\cA}(\cdot)=\indicator{\cdot\in\cA}$.
In this paper, all logarithms are to the base 2. 

We use $ \Delta(\cX) $ to denote the probability simplex on $\cX$.
Related notations such as $ \Delta(\cX\times\cY) $ and $ \Delta(\cY|\cX) $ are similarly defined. 
For a distribution $ P_{\bfx,\bfy|\bfu}\in\Delta(\cX\times\cY|\cU) $, we use $ \sqrbrkt{P_{\bfx,\bfy|\bfu}}_{\bfx|\bfu}\in\Delta(\cX|\cU) $ to denote the marginal distribution onto $ \bfx $ given $ \bfu $, i.e., for every $ x\in\cX,u\in\cU $, $ \sqrbrkt{P_{\bfx,\bfy|\bfu}}_{\bfx|\bfu}(x|u) = \sum_{y\in\cY}P_{\bfx,\bfy|\bfu}(x,y|u) $. 
We use $ \Delta^{(n)}(\cX) $ to denote the set of types (i.e., empirical distributions/histograms, see \Cref{def:type} for formal definitions) of length-$n$ vectors over alphabet $ \cX $.
That is, $ \Delta^{(n)}(\cX) $ consists of all distributions $ P_\bfx\in\Delta(\cX) $ that are induced by $ \cX^n $-valued vectors. 
Other notations such as $ \Delta^{(n)}(\cX\times\cY) $ and $ \Delta^{(n)}(\cY|\cX) $ are similarly defined. 
The notation $ \bfx\sim P_\bfx $ (resp. $ \vbfx\sim P_{\vbfx} $) means that the p.m.f. of a random variable (resp. vector) $ \bfx $ (resp. $ \vbfx $) is $ P_\bfx $ (resp. $ P_\vbfx $). 
If $ \bfx $ is uniformly distributed in $ \cX $, then we write $ \bfx\sim\cX $. 
Throughout this paper, we use $ \distinf{\cdot}{\cdot} $ and $ \distone{\cdot}{\cdot} $ to respectively denote the $ \ell^\infty $ and $\ell^1 $ distances between two distributions which are defined as follows
\begin{align}
\distinf{P}{Q} \coloneqq \sum_{x\in\cX}\abs{P(x) - Q(x)}, \quad 
\distone{P}{Q} \coloneqq \max_{x\in\cX}\abs{P(x) - Q(x)}, \notag 
\end{align} 
for any $ P,Q\in\Delta(\cX) $. 
For a distribution $ P\in\Delta(\cX) $ and a subset $ \cA\subseteq\Delta(\cX) $, the distance (w.r.t. some metric $ \dist(\cdot,\cdot) $) between $ P $ and $\cA$ is defined as $ \dist(P,\cA)\coloneqq\inf_{Q\in\cA}\dist(P,Q) $. 
For $ \cB\subseteq\Delta(\cX) $, the distance between $ \cA $ and $ \cB $ is defined as 
$ \dist(\cA,\cB)\coloneqq\inf_{(P,Q)\in\cA\times\cB}\dist(P,Q) $. 
The inner product between $ P $ and $ Q $ is defined as $ \inprod{P}{Q}\coloneqq\sum_{x\in\cX}P(x)Q(x) $. 
The $ \ell^p $-norm of a vector is denoted by $ \norm{p}{\cdot} $. 
Note that $ \distinf{\cdot}{\cdot} = \norminf{\cdot - \cdot} $ and $ \distone{\cdot}{\cdot} = \normone{\cdot - \cdot} $.

\section{Preliminaries}
\label{sec:prelim}
Let $ P_\bfx\in\Delta(\cX) $. 
We always assume $ \supp(P_\bfx) = \cX $. 
Otherwise, we can properly reduce $ \cX $ to $ \cX' $ and again assume $ P_\bfx\in\Delta(\cX'), \supp(P_\bfx) = \cX' $. 
Define the polynomial $\nu(P_\bfx,n) $ as
\begin{align}
\nu(P_\bfx,n) \coloneqq& 
\sqrt{ (2\pi n)^{\card{\cX}} \prod_{x\in\cX} P_\bfx(x) }. \label{eqn:def-poly}
\end{align}
Note that $ \nu(P_\bfx,n) \ne0 $. 

\begin{lemma}
\label{lem:type-class-conc}
If $ \vbfx\sim P_\bfx^\tn $, then 
for any $ \vx $ of type $ P_\bfx $, we have
$\prob{\vbfx = \vx} = 2^{-H(P_\bfx)}$. 
Moreover, 
$ \prob{\tau_\vbfx = P_\bfx}\asymp1/\nu(P_\bfx, n) $. 
\end{lemma}

\begin{lemma}[Chernoff bound]
\label{lem:chernoff}
Let $ \bfx_1,\cdots,\bfx_N $ be independent $ \curbrkt{0,1} $-valued random variables. 
Let $ \bfx \coloneqq \sum_{i = 1}^N\bfx_i $. 
Then for any $ \sigma\in[0,1] $, 
\begin{align}
\prob{\bfx\ge(1+\delta)\expt{\bfx}} \le& \exp\paren{-\frac{\delta^2}{3}\expt{\bfx}}, \notag \\
\prob{\bfx\le(1 - \delta)\expt{\bfx}}\le& \exp\paren{-\frac{\delta^2}{2}\expt{\bfx}}, \notag \\
\prob{\bfx\notin(1\pm\delta)\expt{\bfx}} \le& 2\exp\paren{-\frac{\delta^2}{3}\expt{\bfx}}. \notag 
\end{align}
\end{lemma}

\begin{lemma}[Sanov's theorem]
\label{lem:sanov}
Let $ \cQ\subseteq\Delta(\cX) $ be a subset of distributions which equals the closure of its interior. 
Let $ \vbfx\sim P_\bfx^\tn $ for some $ P_\bfx\in\Delta(\cX) $. 
Then 
\begin{align}
\lim_{n\to\infty}\frac{1}{n}\log\prob{\tau_\vbfx\in\cA} =& -\inf_{Q_\bfx\in\cQ} \kl{Q_\bfx}{P_\bfx}, \notag 
\end{align}
where the Kullback--Leibler (KL) divergence $ \kl{\cdot}{\cdot} $ between two distributions is defined in \Cref{def:kl-div}. 
\end{lemma}

\begin{fact}
\label{lem:type-concac}
Let $ \vx = (\vx^{(1)}, \vx^{(2)})\in\cX^n $ where $ \vx^{(1)}\in\cX^{\alpha n} $ and $ \vx^{(2)}\in\cX^{(1-\alpha)n} $ for some $ \alpha\in[0,1] $.
Then we have $ \tau_{\vx} = \alpha\tau_{\vx^{(1)}} + (1-\alpha)\tau_{\vx^{(2)}} $. 
\end{fact}

\begin{definition}[Net]
\label{def:net}
Let $ (\cX,\dist) $ be a metric space and $ \eta>0 $ be a constant.
A subset $ \cN\subseteq\cX $ is an \emph{$ \eta $-net} if for all $ x\in\cX $, there exists $ x'\in\cN $ such that $ \dist(x,x')\le\eta $. 
\end{definition}

The following lemma can be proved by taking a simple coordinate quantization. 
A proof can be found in, e.g., \cite{zhang-2020-generalized-listdec-itcs}. 
\begin{lemma}[Bound on size of a net]
\label{lem:bound-net}
Let $ \cX $ be a finite alphabet. 
For any constant $ \eta>0 $, there exists an $ \eta $-net of $ (\Delta(\cX),d_{\infty}) $ of size at most $ \ceil{\frac{\cardX}{2\eta}}^{\cardX} \le \paren{\frac{\cardX}{2\eta}+1}^{\cardX} $. 
\end{lemma}

\begin{fact}
\label{fact:distinf-distone}
For any $ \vx,\vy\in\bR^k $, we have
$\distinf{\vx}{\vy} \le \distone{\vx}{\vy} \le k\cdot\distinf{\vx}{\vy}$. 
\end{fact}


\begin{definition}[Kullback--Leibler (KL) divergence]
\label{def:kl-div}
Let $ \cX $ be a finite set and let $ P,Q\in\Delta(\cX) $. 
Assume that $ P $ is absolutely continuous w.r.t. $ Q $ (i.e., $ \supp(P)\subseteq\supp(Q) $).
The \emph{Kullback--Leibler (KL) divergence} between $ P $ and $Q$ is defined as $ \kl{P}{Q} \coloneqq\sum_{x\in\cX}P(x)\log\frac{P(x)}{Q(x)} $. 
\end{definition}

\begin{definition}[Types]
\label{def:type}
Let $ \cX $ be a finite set and $ n\in\bZ_{\ge1} $. 
The \emph{type} of a vector $ \vx\in\cX^n $, denoted by $ \tau_\vx\in\Delta(\cX) $, is the empirical distribution/histogram of $\vx$ defined as: for every $ x\in\cX $, $ \tau_\vx(x) = \frac{1}{n}\card{\curbrkt{i\in[n]\colon \vx(i) = x}} $. 
The set of all types of $ \cX^n $-valued vectors is denoted by $ \Delta^{(n)}(\cX) $. 
Let $ \cY $ be another finite set and $ \vy\in\cY^n $. 
The \emph{joint type} $ \tau_{\vx,\vy} $ (and $ \Delta^{(n)}(\cX\times\cY) $ correspondingly) and the \emph{conditional type} $ \tau_{\vx|\vy} $ (and $ \Delta^{(n)}(\cX|\cY) $ correspondingly) are defined in a similar manner. 
Furthermore, these definitions can be extended to tuples of vectors in the canonical way.
The set of vectors of the same type is called a \emph{type class}. 
\end{definition}

\begin{fact}[Types are dense in distributions]
\label{fact:type-dense-in-distr}
Let $ \cX $ be a finite set. 
The set $ \bigcup_{n\in\bZ_{\ge1}}\Delta^{(n)}(\cX) $ of types induced by vectors of all possible lengths is dense in the corresponding set $ \Delta(\cX) $ of distributions. 
\end{fact}

The number of types of length-$n$ vectors is polynomial in $n$. 

\begin{lemma}[Number of types \cite{csiszar-1998-types}]
\label{lem:size-type-class}
The number of types corresponding to $ \cX^n $-valued vectors equals $ \binom{n - \card{\cX} - 1}{\card{\cX} - 1}\le(n+\card{\cX} - 1)^{\card{\cX} - 1} $. 
\end{lemma}

\begin{lemma}[Marginalization does not increase distance]
\label{lem:marg-doesnot-increase-dist}
Let $ P_{\bfa,\bfb},Q_{\bfa,\bfb}\in\Delta(\cA\times\cB) $.
Then
$\distone{\sqrbrkt{P_{\bfa,\bfb}}_{\bfa} }{ \sqrbrkt{Q_{\bfa,\bfb}}_{\bfa}} \le \distone{P_{\bfa,\bfb}}{ Q_{\bfa,\bfb}}$. 
\end{lemma}

\begin{proof}
The lemma follows from triangle inequality. 
\begin{align}
\distone{\sqrbrkt{P_{\bfa,\bfb}}_{\bfa} }{ \sqrbrkt{Q_{\bfa,\bfb}}_{\bfa}} 
\le& \sum_{a\in\cA} \abs{ \sum_{b\in\cB}P_{\bfa,\bfb}(a,b) - \sum_{b\in\cB}Q_{\bfa,\bfb}(a,b) } 
\le \sum_{(a,b)\in\cA\times\cB} \abs{P_{\bfa,\bfb}(a,b) - Q_{\bfa,\bfb}(a,b)} = \distone{P_{\bfa,\bfb}}{Q_{\bfa,\bfb}}. \qedhere 
\end{align}
\end{proof}

\section{Basic definitions}
\label{sec:basic_def}

\subsection{Channel and coding}
\label{sec:channel-coding}

\begin{definition}[Omniscient adversarial MACs]
\label{def:omni-adv-mac}
An \emph{omniscient adversarial two-user multiple access channel (MAC)} $ \mactwofull $ is comprised of 
\begin{enumerate}
	\item three alphabets $ \cXa,\cXb,\cS,\cY $ for the input sequence from the first user, the input sequence from the second user, the jamming sequence and the output sequence, respectively;
	\item input constraints $ \ipconstra\subseteq\Delta(\cXa) $ and $ \ipconstrb\subseteq\Delta(\cXb) $ for the first and second users, respectively;
	\item state constraints $ \stconstr\subseteq\Delta(\cS) $ for the jammer;
	\item and the adversarial channel transition law $ W_{\bfy|\bfxa,\bfxb,\bfs} $ 
	that is governed by the adversary.
\end{enumerate}
Suppose that the first (resp. second) transmitter wishes to send a message $ \ma\in[M_1] $ (resp. $\mb\in[M_2]$) to the receiver. 
They are allowed to encode\footnote{Importantly, the encoding process must be completed locally by two individual encoders without cooperation. } $ (\ma,\mb) $ into two sequences (called \emph{codewords}) $ \enca(\ma) = \vxa\in\cXa^n $ and $ \encb(\mb) = \vxb\in\cXb^n $ respectively such that $ \tau_{\vxa}\in\ipconstra,\tau_{\vxb}\in\ipconstrb $.
These two codewords are transmitted into the channel. 
Knowing the transmitted $ \vxa,\vxb $ and the codebooks $ (\cCa,\cCb)\in\cXa^{M_1\times n}\times\cXb^{M_2\times n} $ (i.e., the collection of codeword pairs that encode the messages in $ [M_1]\times[M_2] $; see \Cref{def:codes}), 
the adversary injects an adversarial noise (a.k.a. the \emph{state vector} or \emph{jamming vector}) $ \vs\in\cS^n $ such that $ \tau_{\vs}\in\stconstr $. 
The channel acts on the inputs $ \vxa,\vxb,\vs $ and generates an output $ \vbfy $ memorylessly, i.e., for any $ \vy\in\cY^n $, 
\begin{align}
W_{\vbfy|\vbfxa,\vbfxb,\vbfs}\paren{\vy\condon\vxa,\vxb,\vs} = W_{\bfy|\bfxa,\bfxb,\bfs}^\tn\paren{\vy\condon\vxa,\vxb,\vs} = \prod_{j = 1}^n W_{\bfy|\bfxa,\bfxb,\bfs}\paren{\vy(j)\condon\vxa(j),\vxb(j),\vs(j)}
. \notag
\end{align}
Receiving $ \vbfy $, the decoder is required to output an estimate $ \dec(\vbfy) = \paren{\wh\ma,\wh\mb} $ of the transmitted messages $(\ma,\mb)$.
See \Cref{fig:adv-mac} for a system diagram of $\mactwo$. 
\end{definition}

\begin{figure}[htbp]
	\centering
	\includegraphics[width=0.75\textwidth]{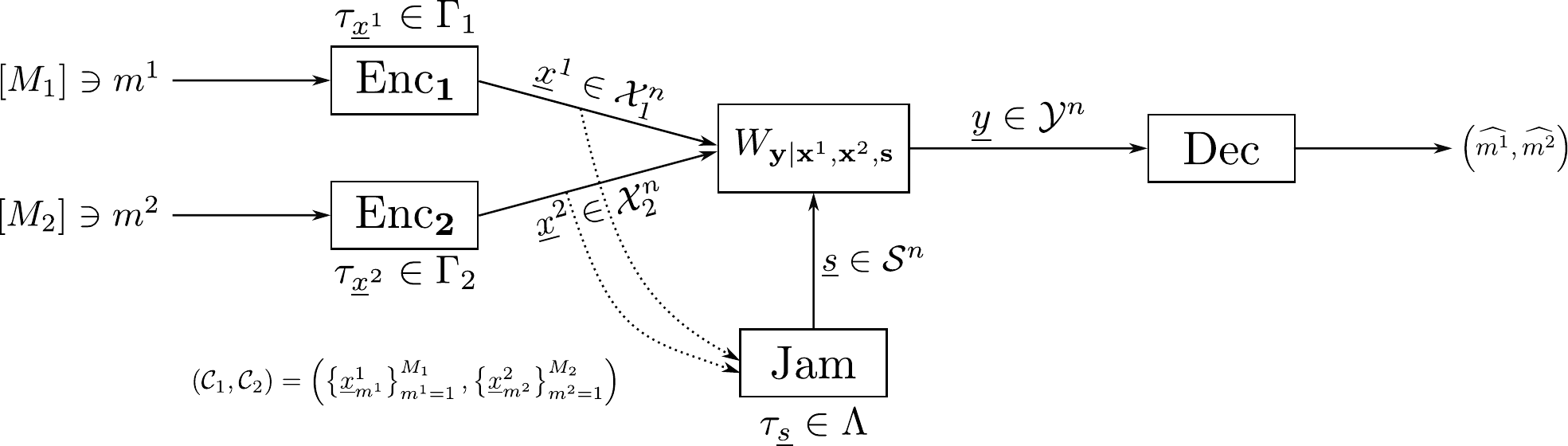}
	\caption{A system diagram of a general two-user omniscient adversarial MAC. }
	\label{fig:adv-mac}
\end{figure}

\begin{remark}
Though the channel from the transmitters to the receiver is memoryless, the state vector $ \vbfs $ is not necessarily generated memorylessly by the jammer given $ \vbfxa,\vbfxb $. 
That is, $ P_{\vbfs|\vbfxa,\vbfxb} $ may not factor. 
Indeed, the adversary can put probability mass one on a single sequence $ \vs $. 
\end{remark}

\begin{definition}[Codes]
\label{def:codes}
A code pair $ (\cCa,\cCb) $ for an omniscient adversarial MAC $ \mactwofull $ consists of 
\begin{enumerate}
	\item two encoders $ \enc_1\colon[M_1]\to\cXa^n $ and $ \enc_2\colon[M_2]\to\cXb^n $ for the first and the second users which map $ \ma\in[M_1] $ and $ \mb\in[M_2] $ to $ \enc_1(\ma) = \vxa_{\ma} $ and $ \enc_2(\mb) = \vxb_{\mb} $ respectively; and
	\item a decoder $ \dec\colon\cY^n\to[M_1]\times[M_2] $ that maps $ \vy $ to $ \dec(\vy) = \paren{\wh \ma,\wh \mb} $. 
\end{enumerate}
We call the images of $ \enc_1 $ and $ \enc_2 $ a \emph{codebook pair} (or simply a \emph{code pair}, overloading the terminology), denoted, with a slight abuse of notation, by $ (\cCa,\cCb)\in\cXa^{M_1\times n}\times\cXb^{M_2\times n} $. 
The length $n$ of each codeword is called the \emph{blocklength}. 
The \emph{rate pair} of $ (\cCa,\cCb) $ is defined as $ R_1 = R(\cCa) \coloneqq \frac{\log M_1}{n\log\card{\cXa}} $ and $ R_2 = R(\cCb) \coloneqq \frac{\log M_2}{n\log\card{\cXb}} $. 

We assume that the code pair $ (\cCa,\cCb) $ is known to $\enca,\encb,\jam$ (see \Cref{def:max-error} below) and is fixed before communication is instantiated.
\end{definition}

\begin{remark}
When we talk about ``a'' code (pair), we always mean an infinite sequence of codes of increasing blocklengths, i.e., $ \curbrkt{\paren{\cCa^{(i)}, \cCb^{(i)}}}_{i\ge1} $ each of blocklength $ n_i $ where $ n_1<n_2<\cdots\in\bZ_{\ge1} $. 
\end{remark}

\begin{definition}[Maximum probability of error]
\label{def:max-error}
A code pair $ (\cCa,\cCb)\in\cXa^{M_1\times n}\times\cXb^{M_2\times n} $ (equipped with encoders $ \enca,\encb $ and a decoder $ \dec $) is said to attain \emph{maximum probability of error $ \eps $} for an omniscient adversarial MAC $$ \mactwofull $$ if 
\begin{align}
& 
\max_{(\ma,\mb)\in[M_1]\times[M_2]}
\max_{\substack{\jam(\enca(\ma),\encb(\mb))\in\cS^n\\\tau_{\jam(\enca(\ma),\encb(\mb))}\in\stconstr}} 
{\probover{\vbfy\sim W_{\bfy|\bfxa,\bfxb,\bfs}^\tn\paren{\cdot\condon\enca\paren{\ma},\encb\paren{\mb},\jam(\enca(\ma),\encb(\mb))}}{\dec\paren{\vbfy} \ne \paren{\ma,\mb}}} \notag \\
=& 
\max_{(\ma,\mb)\in[M_1]\times[M_2]} 
\max_{\substack{\jam(\enca(\ma),\encb(\mb))\in\cS^n\\\tau_{\jam(\enca(\ma),\encb(\mb))}\in\stconstr}} 
\sum_{\vy\in\cY^n:\dec(\vy) \ne (\ma,\mb)} W_{\bfy|\bfxa,\bfxb,\bfs}^\tn\paren{\vy\condon\enca\paren{\ma},\encb\paren{\mb},\jam(\enca(\ma),\encb(\mb))} \notag \\
\le& \eps. \label{eqn:max-prob-error} 
\end{align}
The second maximization is over all legitimate jamming functions $ \jam\colon\cXa^n\times\cXb^n\to\cS^n $ such that $ \tau_{\jam(\enca(\ma),\encb(\mb))} \in\stconstr $.

\end{definition}

\begin{remark}
We emphasize that this paper is focused on the maximum probability of error as defined in \Cref{def:max-error}. 
One can instead place different bounds on the constituent error probabilities \cite{tan-kosut-2013-dispersion-three} 
\begin{align}
& \max_{(\ma,\mb)\in[M_1]\times[M_2]} \max_{\vs:\tau_\vs\in\stconstr} \prob{\curbrkt{\wh\bfma\ne\ma}\cup\curbrkt{\wh\bfmb\ne\mb}}, \notag \\
& \max_{(\ma,\mb)\in[M_1]\times[M_2]} \max_{\vs:\tau_\vs\in\stconstr} \prob{{\wh\bfma\ne\ma}}, \notag \\
& \max_{(\ma,\mb)\in[M_1]\times[M_2]} \max_{\vs:\tau_\vs\in\stconstr} \prob{{\wh\bfmb\ne\mb}}. \notag 
\end{align}
This may create wacky behaviours of the capacity region \cite{zhang-2020-twoway} and is a more challenging question. 
\end{remark}

\begin{definition}[Achievable rate pairs and capacity region]
\label{def:ach-rate-cap-region}
A rate pair $ (R_1,R_2) $ is said to be \emph{achievable} for an omniscient adversarial MAC $ \mactwo $ under the maximum error criterion if there exists a code $ {(\cCa,\cCb)} $ for $\mactwo$ of rates $ R(\cCa)\ge R_1 $ and $ R(\cCb)\ge R_2 $ with $o(1)$ maximum probability of error. 
The closure of all achievable rate pairs is called the \emph{capacity region} of $ \mactwo $. 
\end{definition}

\begin{definition}[Constant composition codes]
\label{def:cc}
A code $ \cC\subseteq\cX^n $ is called $P$-constant composition for some distribution $ P\in\Delta(\cX) $ if all codewords in $\cC$ have type $P$. 
\end{definition}

A simple application of Markov's inequality and \Cref{lem:size-type-class} yields the following reduction from general codes to constant composition codes. 

\begin{lemma}[Constant composition reduction]
\label{lem:cc-reduction}
For any code $\cC\subseteq\cX^n $, there exists a constant composition subcode $ \cC'\subseteq\cC $ of size at least $ |\cC|/(n+\card{\cX} - 1)^{\card{\cX} - 1} $. 
In particular, $ R(\cC') $ is the same as $ R(\cC) $ (asymptotically in $n$). 
\end{lemma}

\Cref{lem:cc-reduction} shows that for the purpose of understanding the capacity (region), it suffices to study constant composition codes. 
Throughout this paper, we focus on constant composition code pairs by fixing two feasible input distributions $ (\ipdistra,\ipdistrb)\in\ipconstra\times\ipconstrb $.

\subsection{Additional technical assumptions}
\label{sec:additional-technical-assump}
For technical reasons, we make further assumptions on the model considered throughout this paper. 
\begin{enumerate}
	\item All alphabets $ \cXa,\cXb,\cS,\cY $ are finite. 
	In particular, our proof will heavily rely on the assumption of the finiteness of $ \cXa $ and $ \cXb $. 
	It is unclear how to extend our results to the large alphabet regime, e.g., the case where $ \card{\cXa},\card{\cXb} $ are increasing in $n$. 
	In fact, we believe that the behaviour of adversarial MACs is considerably different when the alphabet sizes are sufficiently large. 
	See \Cref{itm:open-large-alphabet} in \Cref{sec:concl-rmk-open-prob}. 
	\item In this work we only focus on \emph{state deterministic} channels, i.e., channels for which $ W_{\bfy|\bfxa,\bfxb,\bfs} $ is a zero-one law.
	Alternatively, the channel transition law can be written as a (deterministic) function $ W\colon\cXa\times\cXb\times\cS\to\cY$ such that $ y = W(\xa,\xb,s) $. 
	\item To avoid peculiar behaviours, we assume that $ \ipconstra,\ipconstrb,\stconstr $ are all convex sets. 
	\item We do not assume the availability of common randomness between the encoders and the decoder (while kept secret from the jammer).
	In the AVC literature, the capacity in the presence of shared randomness is known as the \emph{random code capacity} \cite{ahlswede-1978-avc-no-constr,csiszar-narayan-1988-obliavc-constr-random}. 
	\item No party in the system is allowed to use private randomness. 
	That is, the encoding/jamming/decoding functions are all deterministic. 
	In the case of point-to-point omniscient adversarial channels \cite{wbbj-2019-omni-avc}, there are reductions showing that the capacity remains the same under stochastic/deterministic encoding/jamming/decoding. 
	Furthermore, average error criterion is equivalent to maximum error criterion which is further equivalent to zero error criterion when the channel is deterministic. 
	Therefore, the omniscient point-to-point channel problem is combinatorial in nature. 
	However, for our model of omniscient MACs, as alluded to in \Cref{rk:error-criteria}, we expect neither the equivalence between stochastic and deterministic encoding nor the equivalence between average/maximum probability of error. 
	For simplicity, we choose to work with deterministic encoding/jamming/decoding and maximum/zero error criterion in this paper. 
	The average probability of error counterpart is left for future study (see \Cref{itm:open-error-criterion} in \Cref{sec:concl-rmk-open-prob}). 
\end{enumerate}

Under the above assumptions of deterministic encoding/jamming/decoding/channel law and maximum error criterion, the probability in \Cref{eqn:max-prob-error} is either zero or one. 
Therefore, vanishing maximum probability of error implies zero error. 
This enforces a combinatorial nature of the problem in hand. 
Our results serve as a first step towards understanding omniscient adversarial MACs.

\section{Warmup example: binary noisy $ \XOR $ MAC}
\label{sec:warmup_examples}



In this section, we study a warmup example of binary noisy $ \XOR $ MAC defined as follows. 

\begin{definition}[Binary noisy $ \XOR $ MAC]
\label{def:binary-noisy-xor-mac}
A two-user binary noisy $ \XOR $ MAC $ \macxor $ takes as input two binary transmissions $ (\vxa,\vxb)\in\paren{\zon}^2 $ and a binary noise sequence $ \vs\in\zon $ with (relative) Hamming weight at most $ p $ and outputs $ \vy = \vxa\oplus\vxa\oplus\vs $ where the addition is modulo two. 
\end{definition}

The following theorem generalizes the classical Plotkin bound in coding theory to the multiuser setting. 

\begin{theorem}
\label{thm:plotkin_binary_noisy_xor_mac}
If $ p > 1/4 $, then there exists no rate pairs $ (R_1, R_2) $ such that $ R_1>0, R_2>0 $. 
\end{theorem}
\begin{proof}
See \Cref{app:pf_plotkin_binary_noisy_xor_mac}. 
\end{proof}

\section{Confusability sets and their properties}
\label{sec:conf-set}

In this section, we introduce one of the core definitions of this paper: the \emph{confusability sets} associated to an adversarial MAC. 
They are the sets of \emph{bad} distributions that any good code should avoid. 
As the name suggests, they precisely characterize the ``confusability'' of a given channel. 
In fact, they determine the capacity region of the channel and therefore are arguably the most important statistics associated to the channel. 
Some properties of confusability sets are proved. 

We first present an obvious-looking claim which relates the the zero error criterion with \emph{operational non-confusability}.

\begin{claim}[Equivalence between zero error and operational non-confusability]
\label{claim:operational-nonconf}
Let $ \mactwofull $ be a two-user omniscient adversarial MAC. 
A code pair $ (\cCa,\cCb)\in\cXa^{M_1\times n}\times\cXb^{M_2\times n} $ attains zero error for $\mactwo$ if and only if all of the following conditions (which we call \emph{operational non-confusability} conditions) are satisfied:
\begin{enumerate}
	\item \label[cond]{cond:zero-error-joint}
	for all $ 1\le i_1\ne i_2\le M_1 $ and $ 1\le j_1\ne j_2\le M_2 $, there do not exist $ \vsa,\vsb\in\cS^n $ with $ \tau_{\vsa},\tau_{\vsb}\in\stconstr $ such that $ W(\vxa_{i_1}, \vxb_{j_1}, \vsa) = W(\vxa_{i_2}, \vxb_{j_2}, \vsb) $; 
	in this case we say that $ (\vxa_{i_1}, \vxb_{j_1}) $ and $ (\vxa_{i_2}, \vxb_{j_2}) $ are \emph{non-confusable};

	\item \label[cond]{cond:zero-error-marg1}
	for all $ 1\le i_1\ne i_2\le M_1 $ and $ 1\le j\le M_2 $, there do not exist $ \vsa,\vsb\in\cS^n $ with $ \tau_{\vsa},\tau_{\vsb}\in\stconstr $ such that $ W(\vxa_{i_1}, \vxb_j,\vsa) = W(\vxa_{i_2}, \vxb_j,\vsb) $; 
	in this case we say that $ (\vxa_{i_1}, \vxb_j) $ and $ (\vxa_{i_2},\vxb_j) $ are \emph{non-confusable};

	\item \label[cond]{cond:zero-error-marg2}
	for all $ 1\le i\le M_1 $ and $ 1\le j_1\ne j_2\le M_2 $, there do not exist $ \vsa,\vsb\in\cS^n $ with $ \tau_{\vsa},\tau_{\vsb}\in\stconstr $ such that $ W(\vxa_i,\vxb_{j_1},\vsa) = W(\vxa_i,\vxb_{j_2},\vsb) $; 
	in this case we say that $ (\vxa_i,\vxb_{j_1}) $ and $ (\vxa_i,\vxb_{j_2}) $ are \emph{non-confusable}. 
\end{enumerate}
\end{claim}

\begin{proof}
Intuitively, a violation of the zero error criterion must be the case where a received vector $\vy$ can be explained by (at least) two distinct pairs of codewords via admissible jamming vectors. 
In this case, the decoder is confused by (at least) two candidate pairs of codewords and is forced to make a decoding error with nonzero probability. 
Formally, the claim follows from the following simple arguments. 

We first prove the contrapositive of the direct part. 
If $ (\cCa,\cCb) $ has nonzero error, then there must exist a pair of codewords $ (\vxa,\vxb)\in(\cCa,\cCb) $ which leads to a decoding error. 
In particular, at least one of $ \vxa $ and $ \vxb $ cannot be correctly decoded. 
Then at least one of \Cref{cond:zero-error-joint,cond:zero-error-marg1,cond:zero-error-marg2} must be satisfied. 
Indeed, 
\begin{enumerate}
	\item \label[cond]{cond:nonzero-joint}
	\Cref{cond:zero-error-joint} corresponds to the case where neither $\vxa$ nor $\vxb$ can be correctly decoded. 
	More specifically, there must exist another pair of codewords $ \wt\vxa\ne\vxa $ and $ \wt\vxb\ne\vxb $ such that $ W(\vxa,\vxb,\vs) = W(\wt\vxa,\wt\vxb,\wt\vs) $ for some $ \vs,\wt\vs\in\cS^n $ with $ \tau_{\vs},\tau_{\wt\vs}\in\stconstr $. 
	In this case, the decoder could not decide to output $ (\vxa,\vxb) $ or $ (\wt\vxa,\wt\vxb) $. 
	
	\item \label[cond]{cond:nonzero-marg1}
	\Cref{cond:zero-error-marg1} corresponds to the case where $ \vxa $ is confusable with another codeword.
	More specifically, there must exist another codeword $ \wt\vxa\ne\vxa $ such that $ W(\vxa,\vxb,\vs) = W(\wt\vxa,\vxb,\wt\vs) $ for some $ \vs,\wt\vs\in\cS^n $ with $ \tau_{\vs},\tau_{\wt\vs}\in\stconstr $. 
	In this case, the decoder could not decide to output $ (\vxa,\vxb) $ or $ (\wt\vxa,\vxb) $. 

	\item \label[cond]{cond:nonzero-marg2}
	\Cref{cond:zero-error-marg2} corresponds to the case where $\vxb$ is confusable with another codeword.
	More specifically, there must exist another codeword $\wt\vxb\ne\vxb$ such that $ W(\vxa,\vxb,\vs) = W(\vxa,\wt\vxb,\wt\vs) $ for some $ \vs,\wt\vs\in\cS^n $ with $ \tau_{\vs},\tau_{\wt\vs}\in\stconstr $. 
	In this case, the decoder could not decide to output $ (\vxa,\vxb) $ or $ (\vxa,\wt\vxb) $. 
\end{enumerate}

The converse part is straightforward. 
If a code pair $ (\cCa,\cCb) $ attains zero error, then none of \Cref{cond:zero-error-joint,cond:zero-error-marg1,cond:zero-error-marg2} is satisfied. 
Otherwise, (at least) one of \Cref{cond:nonzero-joint,cond:nonzero-marg1,cond:nonzero-marg2} above holds which results in a decoding error, violating the zero-error assumption. 
\end{proof}

\begin{claim}[Permutation invariance of operational (non-)confusability]
\label{claim:per-inv}
If two pairs of codewords $ (\vxa,\vxb) $ and $ (\wt\vxa,\wt\vxb) $ (resp. $(\wt\vxa,\vxb)$ or $ (\vxa,\wt\vxb) $) are confusable/non-confusable (in the sense of \Cref{claim:operational-nonconf}), then any other pairs $ (\vxa_*,\vxb_*) $ and $ (\wt\vxa_*,\wt\vxb_*) $ (resp. $ (\wt\vxa_*,\vxb_*) $ or $(\vxa_*,\wt\vxb_*)$) of the same joint type $ \tau_{\vxa_*,\wt\vxa_*,\vxb_*,\wt\vxb_*} = \tau_{\vxa,\wt\vxa,\vxb,\wt\vxb} $ (resp. $\tau_{\vxa_*,\wt\vxa_*,\vxb_*} = \tau_{\vxa,\wt\vxa,\vxb} $ or $ \tau_{\vxa_*,\vxb_*,\wt\vxb_*} = \tau_{\vxa,\vxb,\wt\vxb} $) are also confusable/non-confusable. 
\end{claim}

\begin{proof}
Since the channel is component-wise and memoryless, the confusability conditions (\Cref{cond:zero-error-joint,cond:zero-error-marg1,cond:zero-error-marg2} in \Cref{claim:operational-nonconf}) are invariant under coordinate permutations. 
That is, $ (\vxa,\vxb) $ is confusable with $ (\wt\vxa, \wt\vxb) $ (resp. $ (\wt\vxa, \vxb) $ or $ (\vxa,\wt\vxb) $) if and only if $ (\pi(\vxa),\pi(\vxb)) $ is confusable with $ (\pi(\wt\vxa), \pi(\wt\vxb)) $ (resp. $ (\pi(\wt\vxa), \pi(\vxb)) $ or $ (\pi(\vxa),\pi(\wt\vxb)) $) for any $ \pi\in S_n $. 
Here for a vector $ \vv = (\vv(1),\cdots,\vv(n))\in\cV^n $, we use the notation $ \pi(\vv)\coloneqq(\vv(\pi(1)),\cdots,\vv(\pi(n))) $. 
Indeed, one simply takes $ \pi(\vs),\pi(\wt\vs) $ of type $ \tau_{\pi(\vs)} = \tau_{\vs}\in\stconstr $ and $ \tau_{\pi(\wt\vs)} = \tau_{\wt\vs}\in\stconstr $. 
Then for any $j\in[n] $, 
\begin{align}
W(\pi(\vxa), \pi(\vxb), \pi(\vs))(j) 
=& W(\pi(\vxa)(j), \pi(\vxb)(j), \pi(\vs)(j)) \label{eqn:component-wise} \\
=& W(\vxa(\pi(j)), \vxb(\pi(j)), \vs(\pi(j))) \notag \\
=& W(\vxa,\vxb,\vs)(\pi(j)) \notag \\
=& \pi(W(\vxa,\vxb,\vs))(j). \notag 
\end{align}
\Cref{eqn:component-wise} is because the channel acts on the inputs component-wise. 
That is, $ W(\pi(\vxa), \pi(\vxb), \pi(\vs)) = \pi(W(\vxa,\vxb,\vs)) $. 
Similarly, $ W(\pi(\wt\vxa), \pi(\wt\vxb), \pi(\wt\vs)) = \pi(W(\wt\vxa,\wt\vxb,\wt\vs)) $ (resp. $ W(\pi(\wt\vxa), \pi(\vxb), \pi(\wt\vs)) = \pi(W(\wt\vxa,\vxb,\wt\vs)) $ or $ W(\pi(\vxa), \pi(\wt\vxb), \pi(\wt\vs)) = \pi(W(\vxa,\wt\vxb,\wt\vs)) $). 
Since $ W(\vxa,\vxb,\vs) = W(\wt\vxa,\wt\vxb,\wt\vs) $ (resp. $W(\vxa,\vxb,\vs) = W(\wt\vxa,\vxb,\wt\vs)$ or $W(\vxa,\vxb,\vs) = W(\vxa,\wt\vxb,\wt\vs)$) and $\pi$ is bijective, we have $ W(\pi(\vxa), \pi(\vxb), \pi(\vs)) = W(\pi(\wt\vxa), \pi(\wt\vxb), \pi(\wt\vs)) $ (resp. $W(\pi(\vxa), \pi(\vxb), \pi(\vs)) = W(\pi(\wt\vxa), \pi(\vxb), \pi(\wt\vs))$ or $W(\pi(\vxa), \pi(\vxb), \pi(\vs)) = W(\pi(\vxa), \pi(\wt\vxb), \pi(\wt\vs))$). 

Finally, permutation invariance of confusability follows from the observation that all vectors of the same type can be obtained by properly permuting the coordinates. 
Since permutations are bijections, non-confusability is also invariant under coordinate permutation. 
\end{proof}

We are ready to give the definition of confusability sets. 
Before doing so, we first define \emph{self-couplings} as distributions with prescribed marginals in accordance with the use of constant composition code pairs. 

\begin{definition}[Self-couplings]
\label{def:self-coupling}
\begin{align}
\cJab \coloneqq& \curbrkt{
	\distraabb \in \Delta(\cXa^2\times\cXb^2) \colon
	\begin{array}{l}
	\sqrbrkt{\distraabb}_{\bfxa_1} = \sqrbrkt{\distraabb}_{\bfxa_2} = \ipdistra, \\
	\sqrbrkt{\distraabb}_{\bfxb_1} = \sqrbrkt{\distraabb}_{\bfxb_2} = \ipdistrb 
	\end{array}
}, \notag \\
\cJa \coloneqq& \curbrkt{
	\distraab \in \Delta(\cXa^2\times\cXb) \colon 
	\sqrbrkt{\distraab}_{\bfxa_1} = \sqrbrkt{\distraab}_{\bfxa_2} = \ipdistra, 
	\sqrbrkt{\distraab}_{\bfxb} = \ipdistrb 
}, \notag \\
\cJb \coloneqq& \curbrkt{
	\distrabb \in \Delta(\cXa\times\cXb^2) \colon 
	\sqrbrkt{\distrabb}_{\bfxa} = \ipdistra ,
	\sqrbrkt{\distrabb}_{\bfxb_1} = \sqrbrkt{\distrabb}_{\bfxb_2} = \ipdistrb 
}. \notag 
\end{align}
\end{definition}

The previous two claims (\Cref{claim:operational-nonconf}, \Cref{claim:per-inv}) motivate us to make the following definition of \emph{confusability sets}.
One should think of the conditions in the definition below as the distributional version of operational confusability in \Cref{claim:operational-nonconf}.

\begin{definition}[Confusability sets]
\label{def:conf-set}
Let $ \mactwofull $ be a 2-user adversarial MAC.
Let $ \ipdistra\in\Delta(\cXa) $ and $ \ipdistrb\in\Delta(\cXb) $. 
The \emph{joint confusability set} $ \cKab $, the \emph{first marginal confusability set} $ \cKa $ and the \emph{second marginal confusability set} $ \cKb $ of $ \mactwo $ w.r.t. input distributions $ \ipdistra $ and $ \ipdistrb $ are defined as follows:
\begin{align}
\cKab\coloneqq& \curbrkt{ \begin{array}{rl}
& \distraabb \in\cJab \colon \\
\exists&  P_{\bfxa_1,\bfxa_2,\bfxb_1,\bfxb_2, \bfs_1, \bfs_2, \bfy}\in\Delta\paren{\cXa^2\times \cXb^2\times \cS^2\times\cY}  \suchthat \\
& \sqrbrkt{ P_{\bfxa_1,\bfxa_2,\bfxb_1,\bfxb_2, \bfs_1, \bfs_2, \bfy} }_{\bfxa_1,\bfxa_2,\bfxb_1,\bfxb_2} = \distraabb ; \\
\forall&  \paren{\xa_1,\xa_2,\xb_1,\xb_2,s_1,s_2,y}\in\cXa^2\times\cXb^2\times\cS^2\times\cY, \\
& P_{\bfxa_1,\bfxa_2,\bfxb_1,\bfxb_2, \bfs_1, \bfs_2, \bfy} \paren{\xa_1,\xa_2,\xb_1,\xb_2,s_1,s_2,y} \\
=& \distraabb \paren{\xa_1,\xa_2,\xb_1,\xb_2} 
 P_{\bfs_1,\bfs_2|\bfxa_1, \bfxa_2, \bfxb_1, \bfxb_2} \paren{s_1,s_2\condon \xa_1, \xa_2, \xb_1, \xb_2} 
 W_{\bfy|\bfxa,\bfxb,\bfs}\paren{y\condon \xa_1,\xb_1,s_1} \\
=& \distraabb \paren{\xa_1,\xa_2,\xb_1,\xb_2} 
 P_{\bfs_1,\bfs_2|\bfxa_1, \bfxa_2, \bfxb_1, \bfxb_2} \paren{s_1,s_2\condon \xa_1, \xa_2, \xb_1, \xb_2} 
 W_{\bfy|\bfxa,\bfxb,\bfs}\paren{y\condon \xa_2,\xb_2,s_2} 
\end{array} }, \notag \\
\cKa \coloneqq& \curbrkt{ \begin{array}{rl}
& P_{\bfxa_1, \bfxa_2, \bfxb}\in\cJa\colon \\
\exists&  P_{\bfxa_1,\bfxa_2,\bfxb, \bfs_1,\bfs_2, \bfy}\in\Delta\paren{\cXa^2\times \cXb\times \cS^2 \times\cY} \suchthat \\
& \sqrbrkt{ P_{\bfxa_1,\bfxa_2,\bfxb, \bfs_1,\bfs_2, \bfy} }_{\bfxa_1,\bfxa_2,\bfxb} = P_{\bfxa_1, \bfxa_2, \bfxb}; \\
\forall&  \paren{\xa_1,\xa_2,\xb,s_1,s_2,y} \in \cXa^2\times\cXb\times\cS^2\times\cY, \\
& P_{\bfxa_1,\bfxa_2,\bfxb, \bfs_1, \bfs_2, \bfy} \paren{\xa_1,\xa_2,\xb,s_1,s_2,y} \\
=& P_{\bfxa_1, \bfxa_2, \bfxb}\paren{\xa_1,\xa_2,\xb} 
 P_{\bfs_1,\bfs_2|\bfxa_1, \bfxa_2, \bfxb} \paren{s_1,s_2\condon \xa_1, \xa_2, \xb} 
 W_{\bfy|\bfxa,\bfxb,\bfs}\paren{y\condon \xa_1,\xb,s_1} \\
=& P_{\bfxa_1, \bfxa_2, \bfxb}(\xa_1,\xa_2,\xb) 
 P_{\bfs_1,\bfs_2|\bfxa_1, \bfxa_2, \bfxb} \paren{s_1,s_2\condon \xa_1, \xa_2, \xb} 
 W_{\bfy|\bfxa,\bfxb,\bfs}\paren{y\condon \xa_2,\xb,s_2} 
\end{array} }, \notag \\
\cKb\coloneqq& \curbrkt{ \begin{array}{rl}
& P_{\bfxa, \bfxb_1, \bfxb_2}\in\cJb \colon \\
\exists&  P_{\bfxa,\bfxb_1,\bfxb_2, \bfs_1, \bfs_2, \bfy}\in\Delta\paren{\cXa\times \cXb^2\times \cS^2\times\cY} \suchthat \\
& \sqrbrkt{ P_{\bfxa,\bfxb_1,\bfxb_2, \bfs_1, \bfs_2, \bfy} }_{\bfxa,\bfxb_1,\bfxb_2} = P_{\bfxa, \bfxb_1, \bfxb_2} ;\\
\forall&  \paren{\xa, \xb_1,\xb_2,s_1,s_2,y} \in \cXa\times\cXb^2\times\cS^2\times\cY, \\
& P_{\bfxa,\bfxb_1,\bfxb_2, \bfs_1, \bfs_2, \bfy} \paren{\xa,\xb_1,\xb_2,s_1,s_2,y} \\
=& P_{\bfxa, \bfxb_1, \bfxb_2}\paren{\xa,\xb_1,\xb_2} 
 P_{\bfs_1,\bfs_2|\bfxa, \bfxb_1, \bfxb_2} \paren{s_1,s_2\condon \xa, \xb_1, \xb_2} 
 W_{\bfy|\bfxa,\bfxb,\bfs}\paren{y\condon \xa,\xb_1,s_1} \\
=& P_{\bfxa, \bfxb_1, \bfxb_2}\paren{\xa,\xb_1,\xb_2} 
 P_{\bfs_1,\bfs_2|\bfxa, \bfxb_1, \bfxb_2} \paren{s_1,s_2\condon \xa, \xb_1, \xb_2} 
 W_{\bfy|\bfxa,\bfxb,\bfs}\paren{y\condon \xa,\xb_2,s_2} 
\end{array} }. \notag 
\end{align}
\end{definition}

One should think of confusability sets as the sets of \emph{bad} distributions/types that any (sequence of) good codes should avoid. 
Indeed, one has the following claim. 

\begin{claim}
\label{claim:distributional-nonconf}
Let $ \mactwofull $ be a 2-user adversarial MAC and let $ \paren{\ipdistra, \ipdistrb}\in \ipconstra\times\ipconstrb $ be a pair of feasible input distributions. 
Let $ \curbrkt{\paren{ \cC_{1,i}, \cC_{2,i}}}_i\subseteq\cXa^{n_i}\times\cXb^{n_i} $ be a sequence of pairs of $ \ipdistra $- and $ \ipdistrb $-constant composition codes of increasing blocklengths $ n_i $'s. 
Then $ \curbrkt{\paren{ \cC_{1,i},\cC_{2,i}}}_i $ achieves zero error for $ \mactwo $ if an only if for every $ i $, there is no $ \paren{\vxa_1,\vxb_1}, \paren{\vxa_2,\vxb_2} \in \cC_{1,i}\times\cC_{2,i} $ and $ \vxa\in \cC_{1,i} $, $ \vxb\in\cC_{2,i} $,  such that at least one of the following happens:  $ \tau_{\vxa_1, \vxa_2, \vxb_1 \vxb_2}\in \cKab $, $ \tau_{\vxa_1,\vxa_2,\vxb} \in \cKa $, $ \tau_{\vxa,\vxb_1,\vxb_2} \in \cKb $. 
\end{claim}

\begin{proof}
\Cref{claim:per-inv} implies that the non-confusability properties (\Cref{cond:zero-error-joint,cond:zero-error-marg1,cond:zero-error-marg2} in \Cref{claim:operational-nonconf}) depend only on the \emph{type} of vectors rather than the order of coordinates. 
We can therefore quotient out type classes (\Cref{def:type}) and work with types instead of vectors.\footnote{Formally, let $\sim_{\mathrm{perm}} $ be a relation on vectors defined as $ \vv\sim_{\mathrm{perm}}\vv' $ iff there is $\pi\in S_n $ such that $ \vv' = \pi(\vv) $. It is easy to check that $ \sim_{\mathrm{perm}} $ is an equivalence relation. As \Cref{claim:per-inv} suggests, the confusability property is a \emph{class invariant} under $\sim_{\mathrm{perm}}$, i.e., it is invariant in each equivalence class by $\sim_{\mathrm{perm}}$. For the purpose of studying confusability, one can without loss of generality focus on equivalence classes (i.e., types) rather than vectors. } 
The above conditions are equivalent to 
\begin{enumerate}
	\item for all $ 1\le i_1\ne i_2\le |\cCa| $ and $ 1\le j_1\ne j_2\le |\cCb| $, there do not exist $ \vsa,\vsb\in\cS^n $ with $ \tau_{\vsa},\tau_{\vsb}\in\stconstr $ and $ \vy\in\cY^n $ such that 
	\begin{align}
	& \tau_{\vxa_{i_1}, \vxb_{j_1}, \vxa_{i_2}, \vxb_{j_2}, \vsa,\vsb,\vy}(\xa_1,\xb_1,\xa_2,\xb_2,s_1,s_2,y) \notag \\
	=& \tau_{\vxa_{i_1}, \vxb_{j_1}, \vxa_{i_2}, \vxb_{j_2}} (\xa_1,\xb_1,\xa_2,\xb_2) 
	\tau_{\vsa,\vsb|\vxa_{i_1}, \vxb_{j_1}, \vxa_{i_2}, \vxb_{j_2}}(s_1,s_2|\xa_1,\xb_1,\xa_2,\xb_2)
	W_{\bfy|\bfxa,\bfxb,\bfs}(y|\xa_1,\xb_1,s_1) \notag \\
	=& \tau_{\vxa_{i_1}, \vxb_{j_1}, \vxa_{i_2}, \vxb_{j_2}} (\xa_1,\xb_1,\xa_2,\xb_2) 
	\tau_{\vsa,\vsb|\vxa_{i_1}, \vxb_{j_1}, \vxa_{i_2}, \vxb_{j_2}}(s_1,s_2|\xa_1,\xb_1,\xa_2,\xb_2)
	W_{\bfy|\bfxa,\bfxb,\bfs}(y|\xa_2,\xb_2,s_2) \notag
	\end{align}
	for all $ (\xa_1,\xa_2,\xb_1,\xb_2,s_1,s_2,y)\in\cXa^2\times\cXb^2\times\cS^2\times\cY $;

	\item for all $ 1\le i_1\ne i_2\le |\cCa| $ and $ 1\le j\le |\cCb| $, there do not exist $ \vsa,\vsb\in\cS^n $ with $ \tau_{\vsa},\tau_{\vsb}\in\stconstr $ and $ \vy\in\cY^n $ such that 
	\begin{align}
	& \tau_{\vxa_{i_1}, \vxa_{i_2}, \vxb_{j}, \vsa,\vsb,\vy}(\xa_1,\xa_2,\xb,s_1,s_2,y) \notag \\
	=& \tau_{\vxa_{i_1}, \vxa_{i_2}, \vxb_{j}} (\xa_1,\xa_2,\xb) 
	\tau_{\vsa,\vsb|\vxa_{i_1},\vxa_{i_2}, \vxb_{j}}(s_1,s_2|\xa_1,\xa_2,\xb)
	W_{\bfy|\bfxa,\bfxb,\bfs}(y|\xa_1,\xb,s_1) \notag \\
	=& \tau_{\vxa_{i_1}, \vxa_{i_2}, \vxb_{j}} (\xa_1,\xa_2,\xb) 
	\tau_{\vsa,\vsb|\vxa_{i_1},\vxa_{i_2}, \vxb_{j}}(s_1,s_2|\xa_1,\xa_2,\xb)
	W_{\bfy|\bfxa,\bfxb,\bfs}(y|\xa_2,\xb,s_2) \notag
	\end{align}
	for all $ (\xa_1,\xa_2,\xb,s_1,s_2,y)\in\cXa^2\times\cXb\times\cS^2\times\cY $;

	\item for all $ 1\le i\le |\cCa| $ and $ 1\le j_1\ne j_2\le |\cCb| $, there do not exist $ \vsa,\vsb\in\cS^n $ with $ \tau_{\vsa},\tau_{\vsb}\in\stconstr $ and $ \vy\in\cY^n $ such that 
	\begin{align}
	& \tau_{\vxa_{i}, \vxb_{j_1}, \vxb_{j_2}, \vsa,\vsb,\vy}(\xa,\xb_1,\xb_2,s_1,s_2,y) \notag \\
	=& \tau_{\vxa_{i}, \vxb_{j_1}, \vxb_{j_2}} (\xa,\xb_1,\xb_2) 
	\tau_{\vsa,\vsb|\vxa_{i}, \vxb_{j_1}, \vxb_{j_2}}(s_1,s_2|\xa,\xb_1,\xb_2)
	W_{\bfy|\bfxa,\bfxb,\bfs}(y|\xa,\xb_1,s_1) \notag \\
	=& \tau_{\vxa_{i}, \vxb_{j_1}, \vxb_{j_2}} (\xa,\xb_1,\xb_2) 
	\tau_{\vsa,\vsb|\vxa_{i}, \vxb_{j_1}, \vxb_{j_2}}(s_1,s_2|\xa,\xb_1,\xb_2)
	W_{\bfy|\bfxa,\bfxb,\bfs}(y|\xa,\xb_2,s_2) \notag
	\end{align}
	for all $ (\xa,\xb_1,\xb_2,s_1,s_2,y)\in\cXa\times\cXb^2\times\cS^2\times\cY $. 
\end{enumerate}
We now get that $ (\cCa,\cCb)\in\cXa^n\times\cXb^n $ attains zero error for $\mactwo$ if and only if the above conditions hold.
Since these conditions should be satisfied for every $n$, by \Cref{fact:type-dense-in-distr}, we pass from types to distributions. 
According to \Cref{def:conf-set}, we finally get that an infinite sequence of codes $ \curbrkt{\paren{\cCa^{(n)}, \cCb^{(n)}}}_{n\ge1} $ attains zero error for $\mactwo$ if and only if for every $n$, 
\begin{enumerate}
	\item for all $ 1\le i_1\ne i_2\le |\cCa^{(n)}| $ and $ 1\le j_1\ne j_2\le |\cCb^{(n)}| $, $ \tau_{\vxa_{i_1}, \vxa_{i_2}, \vxb_{j_1}, \vxb_{j_2}}\notin\cKab $;
	\item for all $ 1\le i_1\ne i_2\le |\cCa^{(n)}| $ and $ 1\le j\le |\cCb^{(n)}| $, $ \tau_{\vxa_{i_1}, \vxa_{i_2}, \vxb_j}\notin\cKa $;
	\item for all $ 1\le i\le |\cCa^{(n)}| $ and $ 1\le j_1\ne j_2\le |\cCb^{(n)}| $, $ \tau_{\vxa_i,\vxb_{j_1}, \vxb_{j_2}}\notin\cKb $. 
\end{enumerate}
This finishes the proof. 
\end{proof}

\begin{remark}
\label{rk:equiv-operational-distributional}
\Cref{claim:operational-nonconf} and \Cref{claim:distributional-nonconf} actually imply that operational confusability and distributional confusability are equivalent, both of which are characterizations of zero error. 
\end{remark}

\begin{remark}
Using operational confusability, one can instead define the confusability sets in terms of types rather than distributions. 
\begin{align}
\cK_{1,2}^{(n)}(\ipdistra,\ipdistrb)\coloneqq&\curbrkt{\tau_{\vxa_1,\vxa_2,\vxb_1,\vxb_2}\in\cJab: 
\begin{array}{c}
(\vxa_1,\vxa_2,\vxb_1,\vxb_2)\in(\cXa^n)^2\times(\cXb^n)^2 \\
(\vxa_1,\vxb_1)\text{ and }(\vxa_2,\vxb_2)\text{ satisfy \Cref{cond:nonzero-joint} in the proof of \Cref{claim:operational-nonconf}}
\end{array}
}, \notag \\
\cK_{1}^{(n)}(\ipdistra,\ipdistrb)\coloneqq&\curbrkt{\tau_{\vxa_1,\vxa_2,\vxb}\in\cJa: 
\begin{array}{c}
(\vxa_1,\vxa_2,\vxb)\in(\cXa^n)^2\times\cXb^n \\
(\vxa_1,\vxb)\text{ and }(\vxa_2,\vxb)\text{ satisfy \Cref{cond:nonzero-marg1} in the proof of \Cref{claim:operational-nonconf}}
\end{array}
} \notag \\
\cK_{2}^{(n)}(\ipdistra,\ipdistrb)\coloneqq&\curbrkt{\tau_{\vxa,\vxb_1,\vxb_2}\in\cJb: 
\begin{array}{c}
(\vxa,\vxb_1,\vxb_2)\in\cXa^n\times(\cXb^n)^2 \\
(\vxa,\vxb_1)\text{ and }(\vxa,\vxb_2)\text{ satisfy \Cref{cond:nonzero-marg2} in the proof of \Cref{claim:operational-nonconf}}
\end{array}
}. \notag 
\end{align}
By \Cref{fact:type-dense-in-distr} and \Cref{rk:equiv-operational-distributional}, the above definition is (almost) the same as \Cref{def:conf-set}. 
Indeed, 
\begin{align}
\cKab =& \cl\paren{\bigcup_{n = 1}^\infty \cK_{1,2}^{(n)}(\ipdistra,\ipdistrb)}, \notag \\
\cKa =& \cl\paren{\bigcup_{n = 1}^\infty \cK_{1}^{(n)}(\ipdistra,\ipdistrb)}, \notag \\
\cKb =& \cl\paren{\bigcup_{n = 1}^\infty \cK_{2}^{(n)}(\ipdistra,\ipdistrb)}, \notag 
\end{align}
where $ \cl(\cdot) $ denotes the closure of a set. 
We stick with the distribution version of the definition rather than type version. 
\end{remark}

\begin{proposition}
\label{prop:prop-conf-set}
Fix any $ (\ipdistra,\ipdistrb)\in\ipconstra\times\ipconstrb $. 
The confusability sets enjoy the following properties.
\begin{enumerate}

	\item \label[prop]{itm:conf-set-prop-nontrivial}
	\emph{Nontriviality.}
	Any distributions $ P_{\bfxa,\bfxa,\bfxb,\bfxb}\in\cJab $, $ P_{\bfxa,\bfxa,\bfxb}\in\cJa $ and $ P_{\bfxa,\bfxb,\bfxb}\in\cJb $ are in $ \cKab,\cKa $ and $ \cKb $, respectively. 
	
	\item \label[prop]{itm:conf-set-prop-transp-inv}
	\emph{Transpositional invariance.} 
	If $ \distraabb $ is in $ \cKab $, then $ P_{\bfxa_2,\bfxa_1,\bfxb_2,\bfxb_1} $ is also in $ \cKa $;
	if $ \distraab $ is in $ \cKa $, then $ P_{\bfxa_2,\bfxa_1,\bfxb} $ is also in $ \cKa $;
	if $ \distrabb $ is in $ \cKb $, then $ P_{\bfxa,\bfxb_2,\bfxb_1} $ is also in $ \cKb $. 
	
	\item \label[prop]{itm:conf-set-prop-conv}
	\emph{Convexity.}
	All of $ \cKab,\cKa,\cKb $ are convex. 
\end{enumerate}
\end{proposition}

\begin{proof}
By \Cref{rk:equiv-operational-distributional}, it is convenient to prove the properties via operational confusability. 

To prove the first property, one simply observes that a pair of codewords $ (\vxa,\vxb) $ is apparently confusable with itself. 
In \Cref{cond:zero-error-joint} (of \Cref{claim:operational-nonconf}), one takes $ \vs = \wt\vs $. 

To prove the second property, one notes that if $ (\vxa,\vxb) $ is confusable with $ (\wt\vxa,\wt\vxb) $ (resp. $ (\wt\vxa,\vxb) $ or $ (\vxa,\wt\vxb) $), then $ (\wt\vxa,\wt\vxb) $ (resp. $ (\wt\vxa,\vxb) $ or $ (\vxa,\wt\vxb) $) is also confusable with $ (\vxa,\vxb) $. 
In the conditions of \Cref{claim:operational-nonconf}, one interchanges the corresponding $\vs$ and $ \wt\vs $. 

To prove the third property, we note that for any $ \alpha\in[0,1] $, if $ (\vec x^1_1,\vec x^2_1)\in\cXa^{\alpha n}\times\cXb^{\alpha n} $ and $ (\vec x^1_2,\vec x^2_2)\in\cXa^{\alpha n}\times\cXb^{\alpha n} $ are confusable (via $ \vec s_1\in\cS^{\alpha n} $ and $\vec s_2\in\cS^{\alpha n} $), $ (\vec x^1_3,\vec x^2_3)\in\cXa^{(1-\alpha)n}\times\cXb^{(1-\alpha)n} $ and $ (\vec x^1_4,\vec x^2_4)\in\cXa^{(1-\alpha)n}\times\cXb^{(1-\alpha)n} $ are also confusable (via $ \vec s_3\in\cS^{(1-\alpha) n}$ and $\vec s_4\in\cS^{(1-\alpha) n} $), then $ ((\vec x^1_1,\vec x^1_3), (\vec x^2_1,\vec x^2_3))\in\cXa^n\times\cXb^n $ and $ ((\vec x^1_2,\vec x^1_4), (\vec x^2_2,\vec x^2_4))\in\cXa^n\times\cXb^n $ are confusable (via $ (\vec s_1,\vec s_3)\in\cS^n $ and $ (\vec s_2,\vec s_4)\in\cS^n $). 
Here for two vectors $ \vec v_1\in\cV^{n_1} $ and $ \vec v_2\in\cV^{n_2} $, we use the notation $ (\vec v_1,\vec v_2)\in\cV^{n_1+n_2} $ to denote the concatenation of $ \vec v_1 $ and $ \vec v_2 $. 
Therefore, by \Cref{lem:type-concac}, if $ \distraabb\in\cKab $ and $ P_{\wt{\bfxa_1},\wt{\bfxa_2},\wt{\bfxb_1},\wt{\bfxb_2}}\in\cKab $ then $ \alpha\distraabb + (1-\alpha)P_{\wt{\bfxa_1},\wt{\bfxa_2},\wt{\bfxb_1},\wt{\bfxb_2}}\in\cKab $ for any $ \alpha\in[0,1] $. 
\end{proof}

\begin{remark}
\label{rk:conf-relation}
If we define the relation $ \sim_{\mathrm{conf}} $ on the set of feasible input sequences as $ (\vxa,\vxb)\sim_{\mathrm{conf}}(\wt\vxa,\wt\vxb) $ (resp. $ (\vxa,\vxb)\sim_{\mathrm{conf}}(\wt\vxa,\vxb) $ or $ (\vxa,\vxb)\sim_{\mathrm{conf}}(\vxa,\wt\vxb) $) iff $ \tau_{\vxa,\wt\vxa,\vxb,\wt\vxb}\in\cKa $ (resp. $ \tau_{\vxa,\wt\vxa,\vxb}\in\cKa $ or $ \tau_{\vxa,\vxb,\wt\vxb}\in\cKb $), then \Cref{prop:prop-conf-set} implies that $ \sim_{\mathrm{conf}} $ is reflective and symmetric. 
However, $ \sim_{\mathrm{conf}} $ is not necessarily transitive. 
Therefore, it is not in general an equivalence relation. 
\end{remark}

\begin{claim}
\label{lem:cap-determined-by-conf-set}
Channels with the same confusability sets have the same capacity region. 
\end{claim}

\begin{proof}
Let $ \mactwo $ and $ \mactwo' $ be two adversarial MACs with the same input constraints $ \ipconstra,\ipconstrb $ and the same confusability sets $ \cKab,\cKa,\cKb $ for all $ (\ipdistra,\ipdistrb)\in\ipconstra\times\ipconstrb $. 
Note that $\mactwo$ and $ \mactwo' $ may have different state/output alphabets and channel laws. 
By \Cref{claim:distributional-nonconf}, any code $ (\cCa,\cCb) $ that attains zero error for $ \mactwo $ also attains zero error for $ \mactwo' $. 
Therefore, any achievable rate pair $ (R_1,R_2) $ for $ \mactwo $ is also achievable for $ \mactwo' $. 
\end{proof}

\section{The sets of good distributions and their properties}
\label{sec:good_distr}

The geometry of various sets of distributions/tensors is depicted in \Cref{fig:geom-of-sets}. 

\begin{figure}[htbp]
	\centering
	\includegraphics[width=0.4\textwidth]{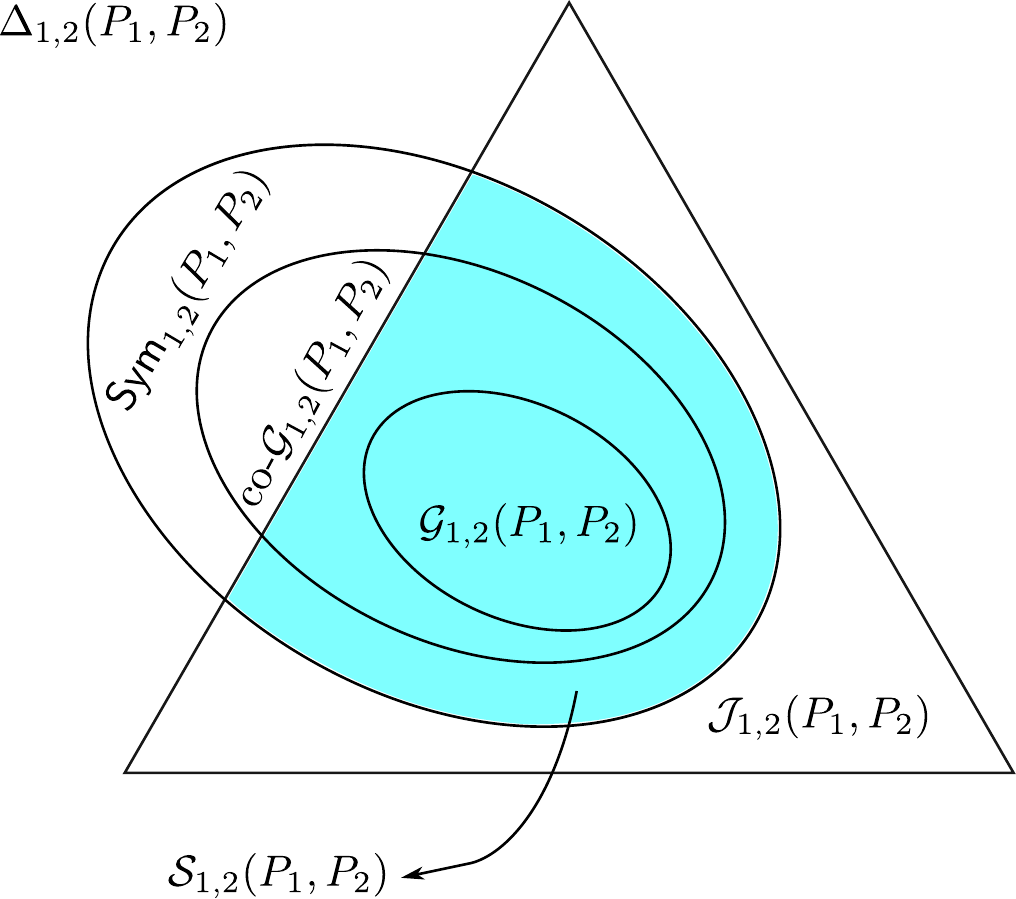}
	\caption{The geometry of various sets of distributions/tensors. We only draw sets of joint distributions/tensors. The geometry of the corresponding marginal distributions/tensors is similar. The ambient space is $ \Dab $ which is defined in \Cref{def:gen-self-coupling}. The set $ \cJab $ of self-couplings is defined in \Cref{def:self-coupling}. The set $ \symab $ of symmetric tensors is defined in \Cref{def:sym_tensor}. Inside $ \symab $, there is a pair of dual cones, viz.: $ \goodab $ (\Cref{def:good_distr}) and $ \cogoodab $ (\Cref{def:cogood-distr}). The {\color{cyan}blue} region denotes the set $ \cSab $ of symmetric distributions (\Cref{def:sym_distr}) which is the intersection of $ \symab $ and $ \cJab $. }
	\label{fig:geom-of-sets}
\end{figure}

\begin{definition}[Generalized self-couplings]
\label{def:gen-self-coupling}
\begin{align}
\Dab \coloneqq& \curbrkt{\taabb\in\bR^{\card{\cXa}^2\times\card{\cXb}^2}: \normone{\taabb} = 1,
\begin{array}{l}
\sqrbrkt{\taabb}_{\bfxa_1} = \sqrbrkt{\taabb}_{\bfxa_2} = \ipdistra, \\
\sqrbrkt{\taabb}_{\bfxb_1} = \sqrbrkt{\taabb}_{\bfxb_2} = \ipdistrb
\end{array}}, \notag \\
\Da \coloneqq& \curbrkt{
	\taab\in\bR^{\card{\cXa}^2\times\card{\cXb}}:
	\normtwo{\taab} = 1,
	\sqrbrkt{\taab}_{\bfxa_1} = \sqrbrkt{\taab}_{\bfxa_2} = \ipdistra, 
	\sqrbrkt{\taab}_{\bfxb} = \ipdistrb
} \notag \\
\Db \coloneqq& \curbrkt{
	\tabb\in\bR^{\card{\cXa}\times\card{\cXb}^2}:
	\normtwo{\tabb} = 1,
	\sqrbrkt{\tabb}_{\bfxa} = \ipdistra, 
	\sqrbrkt{\tabb}_{\bfxb_1} = \sqrbrkt{\tabb}_{\bfxb_2} = \ipdistrb
}. \notag 
\end{align}
\end{definition}

\begin{remark}
For a general tensor (not necessarily a distribution) $ T_{\bfa,\bfb}\in\bR^{\card{\cA}\times\card{\cB}} $, the marginalization of $ T_{\bfa,\bfb} $ onto the first variable $ \bfa $ is defined as $ \sqrbrkt{T_{\bfa,\bfb}}_{\bfa}(a) \coloneqq \sum_{b\in\cB}\abs{T_{\bfa,\bfb}(a,b)} $ for any $ a\in\cA $. 
\end{remark}

\begin{remark}
For the convenience of discussion, the above sets should be thought of as generalizations of distributions (\Cref{def:self-coupling}). 
\end{remark}

\begin{definition}[Symmetric tensors]
\label{def:sym_tensor}
\begin{align}
\symab\coloneqq& \curbrkt{\taabb\in\Dab :
\taabb = T_{\bfxa_2, \bfxa_1, \bfxb_2, \bfxb_1} = T_{\bfxa_2, \bfxa_1, \bfxb_1, \bfxb_2} = T_{\bfxa_1, \bfxa_2, \bfxb_2, \bfxb_1}
}, \notag \\
\syma\coloneqq& \curbrkt{
	\taab\in\Da:
	\taab = T_{\bfxa_2,\bfxa_1,\bfxb}
}, \notag \\
\symb\coloneqq& \curbrkt{
	\tabb\in\Db:
	\tabb = T_{\bfxa,\bfxb_2,\bfxb_1}
}. \notag 
\end{align}
\end{definition}

\begin{definition}[Symmetric distributions]
\label{def:sym_distr}
\begin{align}
\cSab \coloneqq& \cJab\cap\symab, \notag \\ 
\cSa \coloneqq& \cJa\cap\syma, \notag \\ 
\cSb \coloneqq& \cJb\cap\symb. \notag
\notag
\end{align}
\end{definition}

\begin{definition}[Good distributions]
\label{def:good_distr}
Let $ (\ipdistra, \ipdistrb) \in \ipconstra\times\ipconstrb $. 
The set of \emph{jointly good distributions} $ \goodab $, the set of \emph{first marginally good distributions} $ \gooda $ and the set of \emph{second marginally good distributions} $ \goodb $ w.r.t. $ \ipdistra $ and $ \ipdistrb $ are defined as follows:
\begin{align}
\goodab \coloneqq& \curbrkt{
	\distraabb \in \cJab  \colon 
	\begin{array}{l}
	\exists k\in\bZ_{\ge1}, \curbrkt{\lambda_i}_{i = 1}^k \subseteq [0,1] ,\curbrkt{P_{1, i}}_{i = 1}^k \subseteq \Delta(\cXa), \curbrkt{P_{2, i}}_{i = 1}^k \subseteq \Delta(\cXb),   \suchthat \\
	\displaystyle \sum_{i = 1}^k\lambda_i = 1, \distraabb = \sum_{i = 1}^k \lambda_iP_{1,i}^{\ot2} \ot P_{2, i}^{\ot2}
	\end{array}
}, \notag \\
\gooda \coloneqq& \curbrkt{
	\distraab \in \cJa  \colon 
	\begin{array}{l}
	\exists k\in\bZ_{\ge1}, \curbrkt{\lambda_i}_{i = 1}^k \subseteq [0,1] ,\curbrkt{P_{1, i}}_{i = 1}^k \subseteq \Delta(\cXa), \curbrkt{P_{2, i}}_{i = 1}^k \subseteq \Delta(\cXb),   \suchthat \\
	\displaystyle \sum_{i = 1}^k\lambda_i = 1, \distraab = \sum_{i = 1}^k \lambda_i P_{1, i}^{\ot2} \ot P_{2, i}
	\end{array}
}, \notag \\
\goodb \coloneqq& \curbrkt{
	\distrabb \in \cJb  \colon 
	\begin{array}{l}
	\exists k\in\bZ_{\ge1}, \curbrkt{\lambda_i}_{i = 1}^k \subseteq [0,1] ,\curbrkt{P_{1, i}}_{i = 1}^k \subseteq \Delta(\cXa), \curbrkt{P_{2, i}}_{i = 1}^k \subseteq \Delta(\cXb),   \suchthat \\
	\displaystyle \sum_{i = 1}^k\lambda_i = 1, \distrabb = \sum_{i = 1}^k \lambda_i P_{1, i} \ot P_{2, i}^{\ot2}
	\end{array}
}. \notag 
\end{align}
In addition, we define the set of \emph{simultaneously good} distributions $\good$ w.r.t. $ \ipdistra $ and $ \ipdistrb $ as 
\begin{align}
\good \coloneqq& \curbrkt{
	\begin{array}{rl}
	\distraabb\in\goodab\setminus\cKab:& \\
	\sqrbrkt{\distraabb}_{\bfxa_1,\bfxa_2,\bfxb_1} =& \sqrbrkt{\distraabb}_{\bfxa_1,\bfxa_2,\bfxb_2} \in\gooda\setminus\cKa \\
	\sqrbrkt{\distraabb}_{\bfxa_1,\bfxb_1,\bfxb_2} =& \sqrbrkt{\distraabb}_{\bfxa_2,\bfxb_1,\bfxb_2} \in\goodb\setminus\cKb
	\end{array}
}. \notag 
\end{align}
\end{definition}

\begin{proposition}[Properties of good distributions]
\label{prop:properties_good_distr}
The sets $ \gooda,\goodb $ and $ \goodab $ enjoy the following properties. 
\begin{enumerate}
	\item 
	Good distributions are symmetric. 
	\begin{align}
	\goodab\subset\cSab,\quad
	\gooda\subset\cSa,\quad
	\goodb\subset\cSb. \notag 
	\end{align}

	\item 
	For any $ \distraabb \in \goodab $,
	\begin{align}
	\sqrbrkt{\distraabb}_{\bfxa_1, \bfxa_2, \bfxb_1} =& \sqrbrkt{\distraabb}_{\bfxa_1, \bfxa_2, \bfxb_2}, \quad
	\sqrbrkt{\distraabb}_{\bfxa_1, \bfxb_1, \bfxb_2} = \sqrbrkt{\distraabb}_{\bfxa_2, \bfxb_1, \bfxb_2}. \notag
	\end{align}
	\item The sets $ \gooda $ and $ \goodb $ are projections of the set $ \goodab $.
	\begin{align}
	\gooda =& \curbrkt{\sqrbrkt{\distraabb}_{\bfxa_1,\bfxa_2,\bfxb_1}:\distraabb\in\goodab}, \notag \\
	\goodb =& \curbrkt{\sqrbrkt{\distraabb}_{\bfxa_1,\bfxb_1,\bfxb_2}:\distraabb\in\goodab}. \notag
	\end{align}
\end{enumerate}
\end{proposition}

\begin{remark}
Though the good sets $ \goodab,\gooda,\goodb $ are \emph{consistent under projections} (the third property of \Cref{prop:properties_good_distr}), the confusability sets $ \cKab,\cKa,\cKb $ are not. 
Operationally, this is because $ (\vxa_{i_1}, \vxa_{j_1}) $ (or $ (\vxa_{i_1}, \vxb_{j_2}) $) and $ (\vxa_{i_2}, \vxb_{j_1}) $ (or $ (\vxa_{i_2}, \vxb_{j_2}) $) are not necessarily confusable even if $ (\vxa_{i_1}, \vxb_{j_1}) $ and $ (\vxa_{i_2}, \vxb_{j_2}) $ are (for $ i_1\ne i_2 $ and $ j_1\ne j_2 $). 
Therefore, even the second property of \Cref{prop:properties_good_distr} is guaranteed to hold for $ \distraabb\in\cKab $, let alone the third one. 
\end{remark}

\begin{definition}[Co-good tensors]
\label{def:cogood-distr}
\begin{align}
\cogoodab \coloneqq& \curbrkt{
	\distraabb \in \symab \colon 
	\forall P_{\bfxa} \in \Delta(\cXa),
	\forall P_{\bfxb} \in \Delta(\cXb), 
	\inprod{P_{\bfxa}^{\ot2} \ot P_{\bfxb}^{\ot2}}{\distraabb} \ge 0
}, \notag \\
\cogooda \coloneqq& \curbrkt{
	\distraab \in \syma \colon 
	\forall P_{\bfxa} \in \Delta(\cXa),
	\forall P_{\bfxb} \in \Delta(\cXb), 
	\inprod{P_{\bfxa}^{\ot2} \ot P_{\bfxb}}{\distraab} \ge 0 
}, \notag \\
\cogoodb \coloneqq& \curbrkt{
	\distrabb \in \symb \colon 
	\forall P_{\bfxa} \in \Delta(\cXa),
	\forall P_{\bfxb} \in \Delta(\cXb), 
	\inprod{P_{\bfxa} \ot P_{\bfxb}^{\ot2}}{\distrabb} \ge 0 
}. \notag 
\end{align}
\end{definition}

\begin{remark}
\label{rk:cogood-not-distr}
Note that co-good tensors are not necessarily distributions. 
They may have negative entries. 
\end{remark}

\begin{remark}
\label{rk:cogood-subset-good}
It follows from definition that the sets of good distributions are subsets of the corresponding co-good distributions, i.e., 
\begin{align}
\goodab\subset\cogoodab,\quad\gooda\subset\cogooda,\quad\goodb\subset\cogoodb. \notag 
\end{align} 
\end{remark}

\begin{definition}[Dual cone]
\label{def:dual-cone}
The \emph{dual cone} $ \cB^* $ of a cone $ \cB $ in a Hilbert space $ \cH $ is defined as 
$\cB^* \coloneqq \curbrkt{b'\in\cH:\forall b\in\cB,\;\inprod{b}{b'}\ge0}$. 
\end{definition}

\begin{theorem}[Duality]
\label{thm:duality}
The sets $ \goodab $, $ \gooda $ and $ \goodb $ are all closed convex pointed cones with non-empty interior.
Furthermore, the following duality relations hold.
In $ \symab $, $ \goodab $ and $ \cogoodab $ are dual cones of each other.
In $ \syma $, $ \gooda $ and $ \cogooda $ are dual cones of each other.
In $ \symb $, $ \goodb $ and $ \cogoodb $ are dual cones of each other. 
\end{theorem}

\begin{proof}

We first prove the duality relations.
Intuitively, the duality follows since the extremal rays of $ \goodab $ (or $\gooda$, $ \goodb $ respectively) are distributions of the form $ P_{\bfxa}^{\ot2}\ot P_{\bfxb}^{\ot2} $ (or $P_{\bfxa}^{\ot2}\ot P_{\bfxb}$, $P_{\bfxa}\ot P_{\bfxb}^{\ot2}$ respectively). 
Indeed, it follows from \Cref{def:good_distr} that
\begin{align}
\goodab =& \conv\curbrkt{P_{\bfxa}^{\ot2}\ot P_{\bfxb}^{\ot2}:P_{\bfxa}\in\Delta(\cXa),P_{\bfxb}\in\Delta(\cXb)}\cap\cJab, \notag \\
\gooda =& \conv\curbrkt{P_{\bfxa}^{\ot2}\ot P_{\bfxb}:P_{\bfxa}\in\Delta(\cXa),P_{\bfxb}\in\Delta(\cXb)}\cap\cJa, \notag \\
\goodb =& \conv\curbrkt{P_{\bfxa}\ot P_{\bfxb}^{\ot2}:P_{\bfxa}\in\Delta(\cXa),P_{\bfxb}\in\Delta(\cXb)}\cap\cJb, \notag 
\end{align}
where $ \conv\curbrkt{\cdot} $ denotes the convex hull of a set. 
Therefore, one can replace $ \distraabb\in\goodab $ (or $ \distraab $, $ \distrabb $ respectively) in the definition of $ \goodab^* $ (or $ \gooda^* $, $ \goodb^* $ respectively) below with $ P_{\bfxa}^{\ot2}\ot P_{\bfxb}^{\ot2} $ (or $ P_{\bfxa}^{\ot2}\ot P_{\bfxb} $, $ P_{\bfxa}\ot P_{\bfxb}^{\ot2} $ respectively). 
\begin{align}
\goodab^* =& \curbrkt{Q_{\bfxa_1,\bfxa_2,\bfxb_1,\bfxb_2}\in\symab:\forall \distraabb\in\goodab,\;\inprod{\distraabb}{Q_{\bfxa_1,\bfxa_2,\bfxb_1,\bfxb_2}}\ge0}, \notag \\
\gooda^* =& \curbrkt{Q_{\bfxa_1,\bfxa_2,\bfxb}\in\syma:\forall \distraab\in\gooda,\;\inprod{\distraab}{Q_{\bfxa_1,\bfxa_2,\bfxb}}\ge0}, \notag \\
\goodb^* =& \curbrkt{Q_{\bfxa,\bfxb_1,\bfxb_2}\in\symb:\forall\distrabb\in\goodb,\;\inprod{\distrabb}{Q_{\bfxa,\bfxb_1,\bfxb_2}}\ge0}. \notag 
\end{align}
After the replacement, we get exactly $ \goodab $ (or $ \gooda $, $ \goodb $, respectively). 

To formalize this intuition, we prove two-sided set inclusions for $ \cogoodab $ and $ \cogooda $. 
The proof for $ \cogoodb $ is the same as that for $ \cogooda $ up to change of notation.

We first prove $ \cogoodab = \goodab^* $. 
\begin{enumerate}
	\item[$\subseteq$.] 
	Let $ Q_{\bfxa_1,\bfxa_2,\bfxb_1,\bfxb_2}\in\cogoodab $. 
	Let $ \distraabb = \sum_{i = 1}^k\lambda_i P_{1,i}^{\ot2}\ot P_{2,i}^{\ot2}\in\goodab $. 
	By \Cref{def:cogood-distr}, we have $ \inprod{Q_{\bfxa_1,\bfxa_2,\bfxb_1,\bfxb_2}}{P_{1,i}^{\ot2}\ot P_{2,i}^{\ot2}}\ge0 $ for all $ i\in[k] $. 
	Therefore, $ \inprod{\distraabb}{Q_{\bfxa_1,\bfxa_2,\bfxb_1,\bfxb_2}}\ge0 $, which means $ Q_{\bfxa_1,\bfxa_2,\bfxb_1,\bfxb_2}\in\goodab^* $. 
	This proves $ \cogoodab\subseteq\goodab^* $.

	\item[$\supseteq$.]
	Let $ Q_{\bfxa_1,\bfxa_2,\bfxb_1,\bfxb_2}\in\goodab^* $. 
	By \Cref{def:dual-cone}, for any $ P_{\bfxa}\in\Delta(\cXa) $ and $ P_{\bfxb}\in\Delta(\cXb) $, we have $ \inprod{Q_{\bfxa_1,\bfxa_2,\bfxb_1,\bfxb_2}}{P_{\bfxa}^{\ot2}\ot P_{\bfxb}^{\ot2}}\ge0 $ since $ P_{\bfxa}^{\ot2}\ot P_{\bfxb}^{\ot2}\in\goodab $. 
	Therefore, $ Q_{\bfxa_1,\bfxa_2,\bfxb_1,\bfxb_2}\in\cogoodab $ and $ \goodab^*\subseteq\cogoodab $. 
\end{enumerate}

We then prove $ \cogooda = \gooda^* $ in the same way.
\begin{enumerate}
	\item[$\subseteq$.]
	Let $ Q_{\bfxa_1,\bfxa_2,\bfxb}\in\cogoodab $. 
	Let $ \distraab = \sum_{i = 1}^k\lambda_i P_{1,i}^{\ot2}\ot P_{2,i}\in\goodab $.
	By \Cref{def:cogood-distr}, for each $ i\in[k] $, we have $ \inprod{Q_{\bfxa_1,\bfxa_2,\bfxb}}{P_{1,i}^{\ot2}\ot P_{2,i}}\ge0 $.
	Therefore, $ \inprod{\distraab}{Q_{\bfxa_1,\bfxa_2,\bfxb}}\ge0 $, which means $ Q_{\bfxa_1,\bfxa_2,\bfxb}\in\gooda^* $. 
	This proves $ \cogooda\subseteq\gooda^* $. 
	\item[$\supseteq$.]
	Let $ Q_{\bfxa_1,\bfxa_2,\bfxb}\in\gooda^* $. 
	By \Cref{def:dual-cone}, for any $ P_{\bfxa}\in\Delta(\cXa) $ and $ P_{\bfxb}\in\Delta(\cXb) $, we have $ \inprod{Q_{\bfxa_1,\bfxa_2,\bfxb}}{P_{\bfxa}^{\ot2}\ot P_{\bfxb}}\ge0 $ since $ P_{\bfxa}^{\ot2}\ot P_{\bfxb}\in\gooda $. 
	Therefore, $ Q_{\bfxa_1,\bfxa_2,\bfxb}\in\cogooda $ and $ \gooda^*\subseteq\cogoodab $. 
\end{enumerate}
This finishes the proof for duality.

The claimed convexity and conic property of $\cogoodab$, $\cogooda$ and $\cogoodb$ follow directly from \Cref{def:cogood-distr}.
The closedness of $ \goodab$, $\gooda$ and $\goodb $ follows from the fact that the dual cone of any convex cone is closed. 
One can easily find distributions that are in the interior of the cones under consideration. 
The pointedness of $ \goodab$, $\gooda$ and $\goodb $ follows from nonnegativity of the entries of their elements. 
Finally, the pointedness of $ \cogoodab$, $\cogooda$ and $\cogoodb $ follows from the fact that the dual cone of any convex cone with nonempty interior is pointed. 
\end{proof}

\section{A characterization of the shape of capacity region}
\label{sec:our_results}


\begin{theorem}
\label{thm:our-results}
Fix a pair of input distributions $ (\ipdistra, \ipdistrb) \in \ipconstra\times\ipconstrb $. 

\begin{enumerate}
	\item\label[case]{itm:result-1} If $ \good\ne\emptyset $, then the capacity region contains rate pairs $ (R_1,R_2) $ such that $ R_1>0,R_2>0 $ or $ R_1>0,R_2=0 $ or $ R_1=0,R_2>0 $ or $ R_1=0,R_2=0 $.
	\item\label[case]{itm:result-2} If $\good=\emptyset$, $ \gooda\setminus\cKa\ne\emptyset $ and $ \goodb\setminus\cKb\ne\emptyset $, then the capacity region only contains rate pairs $ (R_1,R_2) $ such that $ R_1>0,R_2=0 $ or $ R_1=0,R_2>0 $ or $ R_1=0,R_2=0 $.
	\item\label[case]{itm:result-3} If $ \gooda\setminus\cKa\ne\emptyset $ and $ \gooda\setminus\cKb=\emptyset $, then the capacity region only contains rate pairs $ (R_1,R_2) $ such that $ R_1>0,R_2=0 $ or $ R_1=0,R_2=0 $.
	\item\label[case]{itm:result-4} If $ \gooda\setminus\cKa=\emptyset $ and $ \gooda\setminus\cKb\ne\emptyset $, then the capacity region only contains rate pairs $ (R_1,R_2) $ such that $ R_1=0,R_2>0 $ or $ R_1=0,R_2=0 $.
	\item\label[case]{itm:result-5} If $ \gooda\setminus\cKa=\emptyset $ and $ \gooda\setminus\cKb=\emptyset $, then the capacity region only contains $ (0,0) $.
\end{enumerate}
\end{theorem}

\begin{table}[htbp]
\centering
\begin{tabular}{ccccc}
\hline
Cases & $ \good\ne\emptyset $ & $ \gooda\setminus\cKa\ne\emptyset $ & $ \goodb\setminus\cKb\ne\emptyset $ & Capacity region \\ \hline
Case (1) & {\color{red}$\checkmark$} & {\color{blue}$\checkmark$} & {\color{blue}$\checkmark$} & $ (+,+),(+,0),(0,+),(0,0) $ \\
Case (2) & $\times$ & $\checkmark$ & $\checkmark$ & $(+,0),(0,+),(0,0)$ \\
Case (3) & {\color{blue}$\times$} & $\checkmark$ & {\color{red}$\times$} & $(+,0),(0,0)$ \\
Case (4) & {\color{blue}$\times$} & {\color{red}$\times$} & $\checkmark$ & $(0,+),(0,0)$ \\
Case (5) & {\color{blue}$\times$} & {\color{red}$\times$} & {\color{red}$\times$} & $(0,0)$ \\
\hline
\end{tabular}
\caption{A characterization of the \emph{shape} of the capacity region of any omniscient adversarial two-user MAC. Note that the condition $ \good\ne\emptyset $ implies both $ \gooda\setminus\cKa\ne\emptyset $ and $ \goodb\setminus\cKb\ne\emptyset $. Indeed, the former condition is strictly stronger. In each case, we highlight the conditions in colors in such a way that {\color{red}red} conditions imply {\color{blue}blue} conditions.
Note that the table above covers all possible cases.}
\label{tab:our-results}
\end{table}

The proof of the above characterization is comprised of two parts: achievability (\Cref{thm:achievability}) and converse (\Cref{thm:converse}).

\begin{theorem*}[Achievability, restatement of \Cref{thm:achievability}]
Fix input distributions $ (\ipdistra, \ipdistrb)\in\ipconstra\times\ipconstrb $. 
\begin{enumerate}
	\item If $ \good \ne\emptyset $, then there exist achievable rate pairs $ (R_1, R_2) $ such that $ R_1>0,R_2>0 $.
	\item If $ \gooda\setminus\cKa\ne\emptyset $, then there exist achievable rate pairs $ (R_1, 0) $ such that $ R_1>0 $.
	\item If $ \goodb\setminus\cKb\ne\emptyset $, then there exist achievable rate pairs $ (0,R_2) $ such that $ R_2>0 $.
\end{enumerate}
\end{theorem*}

Various achievability results are proved in \Cref{sec:achievability}. 
Firstly, in \Cref{lem:achievability-prod}, we prove the existence of positive rates using \emph{product} distributions. 
Next, in \Cref{thm:achievability}, we refine this result using \emph{mixtures} of product distributions, i.e., good distributions (\Cref{def:good_distr}). 
Finally, in \Cref{lem:inner_bound_prod_distr} we present \emph{inner bounds} on the capacity region using product distributions.

\begin{theorem}[Converse]
\label{thm:converse}
Fix a pair of input distributions $ (\ipdistra, \ipdistrb) \in \ipconstra\times\ipconstrb $. 

\begin{enumerate}
	\item\label[case]{itm:conv-pos-pos} If $ \good = \emptyset $, then there does not exist achievable rate pair $ (R_1, R_2) $ such that $ R_1>0,R_2>0 $. 
	\item\label[case]{itm:conv-pos-zero} If $ \gooda\setminus\cKa=\emptyset $, then there does not exist achievable rate pair $ (R_1, R_2) $ such that $ R_1>0 $.
	\item\label[case]{itm:conv-zero-pos} If $ \goodb\setminus\cKb=\emptyset $, then there does not exist achievable rate pair $ (R_1, R_2) $ such that $ R_2>0 $.
\end{enumerate}
\end{theorem}

\begin{proof}
\Cref{itm:result-1} is proved in \Cref{sec:conv-pos-pos}.
\Cref{itm:result-2,itm:result-3} are proved in \Cref{sec:conv-pos-zero}.
\end{proof}

\begin{observation}
\label{obs:ach-rect}
For an omniscient two-user adversarial MAC, for $ i=1,2 $, if a rate $ R_i>0 $ is achievable for transmitter $i$, then any rate $ 0\le R_i'\le R_i $ is also achievable for transmitter $i$. 
\end{observation}

By \Cref{obs:ach-rect}, if the capacity region contains a rate pair $ (R_1,R_2) $ where $ R_1>0,R_2>0 $, then the rate pairs $ (R_1,0) $ and $ (0,R_2) $ are also in the capacity region. 

\subsection{A remark on nonconvexity of capacity region}
\label{rk:nonconvex}
As suggested by \Cref{thm:our-results}, the capacity region of an adversarial MAC can be \emph{nonconvex}.
E.g., if a MAC satisfies the conditions in \Cref{itm:result-2} of \Cref{thm:our-results}, then the capacity region only consists of two perpendicular line segments and is therefore nonconvex. 
However, the capacity region cannot be an arbitrary nonconvex region. 
Indeed, \Cref{obs:ach-rect} implies that if a rate pair $ (R_1,R_2) $ with $ R_1>0,R_2>0 $ is achievable, then all rate pairs in the (closed) rectangle with vertices $ (0,0),(R_1,0),(0,R_2),(R_1,R_2) $ are also achievable. 

For AVMACs (i.e., the \emph{oblivious} adversarial MACs), the nonconvexity of the capacity region was noted by Gubner--Hughes \cite{gubner-hughes-1995-nonconvex-avmac} and Pereg--Steinberg \cite{pereg-steinberg-2019-avmac} via the example of an (oblivious) erasure MAC. 
As a side note, for AVMACs equipped with common randomness, the capacity region may or may not be convex, depending on how the common randomness is instantiated. 
If each encoder shares an \emph{independent} secret key with the decoder, then the corresponding capacity region, known as the \emph{divided-randomness} capacity region, is not necessarily convex \cite{gubner-hughes-1995-nonconvex-avmac}. 
On the other hand, if all of two encoders and the decoder share the \emph{same} key, then the corresponding capacity region, known as the \emph{random code} capacity region, is always convex \cite{pereg-steinberg-2019-avmac}. 
In our work, we do not equip any party with shared randomness. 
See \cite{pereg-steinberg-2019-avmac} for a more detailed discussion on the nonconvexity of the capacity region of AVMACs. 

\subsection{Comparison of our results with \cite{pereg-steinberg-2019-avmac} on (oblivious) AVMACs}
\label{sec:comparison-our-peregsteinberg}

We compare below our results with the parallel results by Pereg and Steinberg on \emph{oblivious} AVMACs. 
For simplicity, we only compare the characterizations of \emph{positivity} of capacities. 
Specifically, an oblivious AVMAC is a general adversarial MAC with input and state constraints and an oblivious adversary who does \emph{not} know the transmitted sequences from any of the encoders. 
As many other results in the AVC literature, their characterization involves the oblivious analog of confusability known as \emph{symmetrizability}. 
Proper notions of \emph{first marginal symmetrizability}, \emph{second marginal symmetrizability} and \emph{joint symmetrizability} (denoted in their notation by \emph{symmetrizability-$\cXa|\cXb$}, \emph{symmetrizability-$\cXb|\cXa$} and \emph{symmetrizability-$\cXa\times\cXb$} respectively) were introduced and were shown to characterize the capacity positivity. 
See \Cref{tab:pereg-steinberg-results} below. 

\begin{table}[htbp]
\centering
\begin{tabular}{ccccc}
\hline
Cases & non-joint symmetrizability & non-first marginal symmetrizability & non-second marginal symmetrizability & Capacity region \\ \hline
Case (1) & {$\checkmark$} & {$\checkmark$} & {$\checkmark$} & $ (+,+),(+,0),(0,+),(0,0) $ \\
Case (2) & {$\checkmark$} & $\checkmark$ & {$\times$} & $(+,0),(0,0)$ \\
Case (3) & {$\checkmark$} & {$\times$} & $\checkmark$ & $(0,+),(0,0)$ \\
Case (4) & {$\times$} & {$?$} & $?$ & $(0,0)$ \\
Case (5) & {$\checkmark$} & {$\times$} & $\times$ & $(0,0)$ \\
\hline
\end{tabular}
\caption{Results in \cite{pereg-steinberg-2019-avmac} on capacity positivity of oblivious AVMACs. In the table, ``$\checkmark$'' (resp. ``$\times$'') means the corresponding non-symmetrizability condition is satisfied (resp. unsatisfied). Question marks ``$?$'' mean either satisfied or unsatisfied, regardlessly. As noted, non-joint symmetrizability is a necessary condition for any positive achievable rate.}
\label{tab:pereg-steinberg-results}
\end{table}

Intuitively, one should think of symmetrizability as the oblivious analog of confusability defined in \Cref{sec:conf-set}. 
However, in the AVMAC setting, due to the ``independence'' between the jammer and the encoders, the formal definition of symmetrizability does not appear to be a straightforward adjustment of \Cref{def:conf-set}. 
As a result, the characterization of positivity in \cite{pereg-steinberg-2019-avmac} does not exactly parallel ours. 
An informal analogy between the symmetrizability of Pereg and Steinberg's and the confusability of ours is as follows.
Non-first (resp. -second) marginal symmetrizability corresponds to $ \gooda\setminus\cKa \ne\emptyset $ (resp. $ \goodb\setminus\cKb \ne\emptyset $).
Non-joint symmetrizability corresponds to $ \goodab\setminus\cKab \ne\emptyset $. 
However, one gets \emph{wrong} results (for \Cref{itm:result-1,itm:result-2} in particular) if she/he verbatim translates the oblivious results to the omniscient setting using the aforementioned informal correspondence. 

In the AVMAC setting, non-joint symmetrizability is a necessary condition for the existence of $ R_1>0 $ or $ R_2>0 $. 
As a consequence, there does not exist situation where $ R_1>0 $ or $ R_2>0 $ can be achieved separately yet not simultaneously (\Cref{itm:result-2} in \Cref{thm:our-results}). 

In the omniscient setting, the condition that determines the possibility of $ (R_1,R_2) $ with $ R_1>0,R_2>0 $ is in terms of $ \good $ rather than $ \goodab\setminus\cKab $. 
Communication at positive rates for both encoders simultaneously may not be possible even if $ \goodab\setminus\cKab\ne\emptyset $. 
It is possible only if there is a \emph{single} good distribution (as per \Cref{def:good_distr}) that is simultaneously non-jointly symmetrizable and non-marginally symmetrizable (for both transmitters).

\section{Overview of proof techniques}
\label{sec:overview-techniques}
In this section we overview the proof techniques for establishing \Cref{thm:our-results}. 
Since there are cases where both/exactly one/none of the transmitters can achieve positive rates, we have to divide the analysis into several cases. 
Nevertheless, the proofs for different cases share roughly the same structure. 
In what follows, we briefly introduce the ideas behind the achievability part and the converse part separately. 

\subsection{Proof techniques for achievability}
\label{sec:tech-ach}
To show positive achievable rates under the conditions of \Cref{thm:achievability}, we use the standard method of random coding with expurgation. 
The conditions in \Cref{thm:achievability} can be intuitively interpreted as the existence of \emph{good} distributions (according to \Cref{def:good_distr}) that are not \emph{bad} (according to \Cref{def:conf-set}). 

If one is able to find a \emph{product} distribution (which is always good by definition) that is outside the confusability sets, then one can simply sample positive rate codes whose entries are i.i.d. according to the distribution. 
By concentration of measure, the joint type of any codeword tuple is tightly concentrated around the product distribution. 
In particular, any joint type is outside the confusability sets with high probability. 
Now by large deviation principle, if the code rates are sufficiently small, a union bound over all codeword tuples allows us to conclude that no joint type is confusable and hence the whole code pair attains zero error with high probability.  
This gives \Cref{lem:achievability-prod}. 

\Cref{lem:achievability-prod} can be strengthened in the following two ways. 

Firstly, even if product distributions are confusable, if one can find \emph{mixtures} of product distributions that are outside the confusability sets, then positive rates are still achievable. 
Here the additional idea is \emph{time-sharing}. 
Recall that a good distribution is a convex combination of product distributions.\footnote{Note that importantly, the components of such a convex combination do not have to satisfy the input constraints. This is why it is possible to find mixtures of product distributions that are non-confusable even if all \emph{feasible} product distributions are confusable. See \Cref{rk:coded-time-sharing}. } 
The coefficients of the convex combination can be regarded as giving a time-sharing sequence. 
We then sample random codes in the following way. 
All codewords are chopped up into chunks of lengths proportional to the convex combination coefficients.
Entries of all codewords in a particular chunk are i.i.d. according the corresponding component distribution of the convex combination. 
Effectively it is as if we convexly concatenate multiple codebooks of shorter lengths sampled from different product distributions. 
Again by a Chernoff-union argument, all joint types are tightly concentrated around the mixture distribution provided that the rates are sufficiently small. 
Since the mixture distribution itself is outside the confusability sets, the code pair attains zero error with high probability.
This gives \Cref{thm:achievability}. 
Such a code construction is known as \emph{coded time-sharing} (see \Cref{rk:coded-time-sharing}).

Secondly, by carefully analyzing the large deviation exponent, one can in fact obtain \emph{inner bounds} on the capacity region. 
To this end, one could not simply set the rates to be sufficiently small so as to admit a union bound. 
A standard trick is to \emph{remove} (a.k.a. \emph{expurgate}) one codeword from each confusable pair. 
Using Sanov's theorem (\Cref{lem:sanov}), one can get the exact exponent of the probability of sampling a confusable pair. 
One can then set the rates so as to guarantee that the (expected) number of expurgated codewords is at most, say, half of the code size. 
This ensures that the expurgation process does not hurt the rate. 
This gives \Cref{lem:inner_bound_prod_distr}. 
We remark that if one wishes to achieve a rate pair with two positive rates, then the above argument requires one to expurgate codewords that contribute to (at least one of) jointly confusable pairs, first marginally confusable pairs or second marginally confusable pairs. 
We believe that such an expurgation strategy is pessimistic and higher rates may be obtained using more clever expurgation strategies. 
See \Cref{itm:open-better-inner-bounds} in \Cref{sec:concl-rmk-open-prob}. 

\subsection{Proof techniques for converse}
\label{sec:tech-conv}
The converse part is considerably more involved. 
At a high level, it is inspired by the classical Plotkin bound in coding theory and follows a similar structure as \cite{wbbj-2019-omni-avc}. 
However, due to the multiuser nature of the channel, the case analysis is more delicate. 

The basic proof strategy is comprised of the following components. 
Given any code pair $ (\cCa,\cCb) $ that attains zero error, we would like to show that they have zero rate(s) once the conditions in \Cref{thm:converse} are satisfied. 
To this end, we follow the steps below. 
\begin{enumerate}
	\item First, we extract a subcode pair $ (\cCa',\cCb') $ which has nontrivial sizes and is ``equicoupled''. 
	More specifically, for one thing, the code sizes are mildly large in the sense that $ |\cC_i'|\xrightarrow{|\cC_i|\to\infty}\infty $ for $ i = 1,2 $. 
	In fact $ |\cC_i'| = f(|\cC_i|) $ where $ f(\cdot) $ is the inverse Ramsey number which grows extremely slow. 
	However, this is enough for our purposes since it will be ultimately proved that $ \max\curbrkt{|\cCa'|,|\cC'|}\le C $ for some \emph{constant} $ C>0 $ independent of $n$. 
	Then $ \max\curbrkt{|\cCa|,|\cCb|}\le f^{-1}(C) $ which is a huge constant. 
	However, this is already more than sufficient to imply zero rates. 
	For another (more important) thing, the subcode pair we obtained is highly structured in the sense that the joint type of any codeword tuple from the subcode pair is approximately the same (hence the subcodes are at times called \emph{equicoupled} in this paper). 
	This follows from Ramsey's theorem (\Cref{lem:hypergraph_ramsey}). 
	At the cost of losing rates (which is actually fine), we localize some highly regular structures into a tiny subcode pair.

	\item We then focus on the subcode pair. 
	It is unclear whether or not the distribution that all joint types are concentrated around is symmetric (as per \Cref{def:sym_distr}). 
	However, viewing the codebook as a sequence of random variables, we can show (in \Cref{sec:converse_asymm_case}) that the size of the equicoupled subcode must be small if the distribution is asymmetric. 
	This, after some preprocessing of the sequence of random variables, follows from a classical theorem by \komlos (\Cref{lem:komlos}). 
	
	\item Now we assume that the equicoupled subcode is equipped with a symmetric distribution. 
	Since we started with a code pair of zero error, all joint types are outside the confusability sets. 
	Hence by the equicoupledness property, the associated distribution is outside the confusability sets as well. 
	By the assumptions of \Cref{thm:converse}, this distribution cannot be good (as per \Cref{def:good_distr}) since the sets of good distributions are assumed to be subsets of the confusability sets. 
	By the duality (\Cref{thm:duality}) between the sets of good and ``co-good'' tensors (\Cref{def:cogood-distr}), we can find a \emph{witness} (which itself is a co-good tensor) of the non-goodness of the distribution. 
	This finally allows us to apply a Plotkin-type double counting trick. 
	Specifically, we upper and lower bound the following crucial quantity (\Cref{eqn:double_count_object}): the average inner product between the witness and the joint types in the subcodes. 
	Careful calculations give us upper and lower bounds on this quantity. 
	Contrasting these bounds further gives us an upper bound on the code sizes as promised. 
\end{enumerate}

Similar argument can be adapted to the marginal case where exactly one transmitter suffers from zero capacity.

\section{Achievability}
\label{sec:achievability}


We need the following lemma which concentrates the size of the constant composition component of a random code. 
The proof follows from the Chernoff bound (\Cref{lem:chernoff}) and can be found in, e.g., \cite{zhang-2020-generalized-listdec-itcs}. 

\begin{lemma}
\label{lem:const-comp-conc}
Let $ \cC\subseteq\cX^n $ be a random code that consists of codewords $ \vbfx_1,\cdots,\vbfx_M $ i.i.d. according to $ P_\bfx^\tn $ for some $ P_\bfx\in\Delta(\cX) $. 
Let $ \cC'\subseteq\cC $ be the $ P_\bfx $-constant composition subcode of $\cC$.
Then
\begin{align}
\prob{\card{\cC'}\notin(1\pm1/2) \frac M {\nu(P_\bfx,n)}} \le& 2\exp\paren{-\frac{M}{12\nu(P_\bfx,n)}}. \notag 
\end{align}
\end{lemma}

\subsection{Positive achievable rates via product distributions}
\label{sec:positive_rate_prod_distr}

\begin{lemma}[Positive achievable rates via product distributions]
\label{lem:achievability-prod}
Let $ (\ipdistra, \ipdistrb)\in\ipconstra\times\ipconstrb $. 
\begin{enumerate}
	\item\label[case]{itm:prod-pos-pos} If $ \ipdistra^{\ot2}\ot\ipdistrb^{\ot2} \notin\cKab $, $\ipdistra^{\ot2}\ot\ipdistrb \notin\cKa$ and $\ipdistra\ot\ipdistrb^{\ot2} \notin\cKb$, then there exist achievable rate pairs $ (R_1,R_2) $ such that $ R_1>0,R_2>0 $. 
	\item\label[case]{itm:prod-pos-zero} If $ \ipdistra^{\ot2}\ot\ipdistrb\notin\cKa $, then there exist achievable rate pairs $ (R_1,R_2) $ such that $ R_1>0,R_2=0 $. 
	\item\label[case]{itm:prod-zero-pos} If $ \ipdistra\ot\ipdistrb^{\ot2}\notin\cKa $, then there exist achievable rate pairs $ (R_1,R_2) $ such that $ R_1=0,R_2>0 $. 
\end{enumerate}
\end{lemma}

\begin{proof}[Proof of \Cref{itm:prod-pos-pos} in \Cref{lem:achievability-prod}]
	Assume that both $ \ipdistra $ and $ \ipdistrb $ have no zero atoms. 
	Sample a random code pair $ \paren{\cCa,\cCb}\subseteq\cXa^n\times\cXb^n $ of sizes $ (M_1, M_2) $, where $ \cC_i $ consists of codewords $ \vbfx^i_1,\cdots,\vbfx^i_{M_i} $ i.i.d. according to $ P_i^\tn $ ($i = 1,2$). 
	Note that for any $ 1\le i_1 < i_2\le M_1 $ and $ 1\le j_1 < j_2\le M_2 $, 
	\begin{align}
	\expt{\tau_{\vbfxa_{i_1}, \vbfxa_{i_2}, \vbfxb_{j_1}, \vbfxb_{j_2}}} =& \ipdistra^{\ot2}\ot\ipdistrb^{\ot2}. \label{eqn:expt-type-12} 
	\end{align}
	To see this, for any $ (\xa_1,\xa_2,\xb_1,\xb_2)\in\cXa^2\times\cXb^2 $, 
	\begin{align}
	\expt{\tau_{\vbfxa_{i_1}, \vbfxa_{i_2}, \vbfxb_{j_1}, \vbfxb_{j_2}}}(\xa_1,\xa_2,\xb_1,\xb_2) =& \frac{1}{n}\sum_{k = 1}^n\expt{\indicator{\vbfxa_{i_1}(k) = \xa_1,\vbfxa_{i_2}(k) = \xa_2,\vbfxb_{j_1}(k) = \xb_1,\vbfxb_{j_2}(k) = \xb_2}} \notag \\
	=& \frac{1}{n}\sum_{k = 1}^n \expt{\indicator{\vbfxa_{i_1}(k) = \xa_1}}\expt{\indicator{\vbfxa_{i_1}(k) = \xa_2}}\expt{\indicator{\vbfxb_{j_1}(k) = \xb_1}}\expt{\indicator{\vbfxb_{j_2}(k) = \xb_2}} \label{eqn:expt-type-indep} \\
	=& \frac{1}{n}\sum_{k=1}^n\prob{\vbfxa_{i_1}(k) = \xa_1}\prob{\vbfxa_{i_2}(k) = \xa_2}\prob{\vbfxb_{j_1}(k) = \xb_1}\prob{\vbfxb_{j_2}(k) = \xb_2} \notag \\
	=& \ipdistra(\xa_1)\ipdistra(\xa_2)\ipdistrb(\xb_1)\ipdistrb(\xb_2), \label{eqn:expt-type-identical}
	\end{align}
	where \Cref{eqn:expt-type-indep} follows since each codeword is sampled independent; 
	\Cref{eqn:expt-type-identical} follows since each component is identically distributed. 
	Similarly,
	\begin{align}
	\expt{\tau_{\vbfxa_{i_1}, \vbfxa_{i_2}, \vbfxb_{j_1}}} =& \ipdistra^{\ot2}\ot\ipdistrb, \quad
	\expt{\tau_{\vbfxa_{i_1}, \vbfxb_{j_1}, \vbfxb_{j_2}}} = \ipdistra\ot\ipdistrb^{\ot2}. \notag 
	\end{align}
	Let $ \cC_i' $ be the $ P_i $-constant composition subcode of $ \cC_i $ ($i = 1,2$). 
	By \Cref{lem:const-comp-conc}, for $ i = 1,2 $,
	\begin{align}
	\prob{\card{\cC_i'}\notin(1\pm1/2) \frac {M_i} {\nu(P_i,n)}} \le& 2\exp\paren{-\frac{M_i}{12\nu(P_i,n)}}. \label{eqn:positive-product-bound1} 
	\end{align}

	Let 
	\begin{equation}
	\begin{aligned}
	\rho_{1,2} \coloneqq& \distinf{\ipdistra^{\ot2}\ot\ipdistrb^{\ot2}}{\cKab} ,  \\
	\rho_1 \coloneqq& \distinf{\ipdistra^{\ot2}\ot\ipdistrb}{\cKa} ,  \\
	\rho_2 \coloneqq& \distinf{\ipdistra\ot\ipdistrb^{\ot2}}{\cKb} ,  \\
	\eps \coloneqq& \frac{1}{2}\min\curbrkt{\rho_{1,2}, \rho_1, \rho_2}. 
	\end{aligned}
	\label{eqn:prod-def-rho}
	\end{equation}
	By the assumptions of \Cref{itm:prod-pos-pos}, all the above quantities are \emph{strictly} positive.
	Since $ \eps<\rho_{1,2} $, 
	for any $ 1\le i_1<i_2\le M_1 $ and $ 1\le j_1<j_2\le M_2 $, 
	\begin{align}
	& \prob{\tau_{\vbfxa_{i_1}, \vbfxa_{i_2}, \vbfxb_{j_1}, \vbfxb_{j_2}} \in\cKab} \notag \\
	\le& \prob{\distinf{\tau_{\vbfxa_{i_1}, \vbfxa_{i_2}, \vbfxb_{j_1}, \vbfxb_{j_2}}}{\ipdistra^{\ot2}\ot\ipdistrb^{\ot2}}\ge \eps} \notag \\
	=& \prob{\exists (\xa_1,\xa_2,\xb_1,\xb_2)\in\cXa^2\times\cXb^2,\;\abs{\tau_{\vbfxa_{i_1}, \vbfxa_{i_2}, \vbfxb_{j_1}, \vbfxb_{j_2}}(\xa_1,\xa_2,\xb_1,\xb_2) - \ipdistra(\xa_1)\ipdistra(\xa_2)\ipdistrb(\xb_1)\ipdistrb(\xb_2)}\ge\eps} \notag \\
	\le& \sum_{(\xa_1,\xa_2,\xb_1,\xb_2)\in\cXa^2\times\cXb^2} \prob{\abs{\sum_{k = 1}^n\indicator{\vbfxa_{i_1}(k) = \xa_1, \vbfxa_{i_2}(k) = \xa_2, \vbfxb_{j_1}(k) = \xb_1, \vbfxb_{j_2}(k) = \xb_2} - n\ipdistra(\xa_1)\ipdistra(\xa_2)\ipdistrb(\xb_1)\ipdistrb(\xb_2)} \ge n\eps} \notag \\
	=& \sum_{(\xa_1,\xa_2,\xb_1,\xb_2)\in\cXa^2\times\cXb^2} \prob{\sum_{k = 1}^n\indicator{\vbfxa_{i_1}(k) = \xa_1, \vbfxa_{i_2}(k) = \xa_2, \vbfxb_{j_1}(k) = \xb_1, \vbfxb_{j_2}(k) = \xb_2} \notin \paren{1\pm\frac{n\eps}{\mu}}\mu} \label{eqn:define-mu} \\
	\le&\sum_{(\xa_1,\xa_2,\xb_1,\xb_2)\in\cXa^2\times\cXb^2} 2\exp\paren{-\frac{1}{3}\paren{\frac{n\eps}{\mu}}^2\mu} \label{eqn:pos-pos-chernoff} \\
	=& \sum_{(\xa_1,\xa_2,\xb_1,\xb_2)\in\cXa^2\times\cXb^2} 2\exp\paren{-\frac{n\eps^2}{3\ipdistra(\xa_1)\ipdistra(\xa_2)\ipdistrb(\xb_1)\ipdistrb(\xb_2)}} \label{eqn:pos-pos-type} \\
	\le& \card{\cXa}^2\card{\cXb}^2 \cdot2\exp\paren{-\frac{n\eps^2}{3}}. \label{eqn:pos-pos-one}
	\end{align}
	In \Cref{eqn:define-mu}, we define 
	\begin{align}
	\mu = \mu(\xa_1,\xa_2,\xb_1,\xb_2) \coloneqq \expt{\tau_{\vbfxa_{i_1}, \vbfxa_{i_2}, \vbfxb_{j_1}, \vbfxb_{j_2}}}(\xa_1,\xa_2,\xb_1,\xb_2) = \ipdistra(\xa_1)\ipdistra(\xa_2)\ipdistrb(\xb_1)\ipdistrb(\xb_2)>0. \notag 
	\end{align}
	\Cref{eqn:pos-pos-chernoff} is by \Cref{lem:chernoff}. 
	In \Cref{eqn:pos-pos-type}, we used \Cref{eqn:expt-type-12}. 
	In \Cref{eqn:pos-pos-one}, we used the trivial bound: for $i = 1,2$, $ P_i(x)\le1 $ for $ x\in\cX_i $. 

	We only need to consider ordered pairs $ i_1<i_2 $ and $ j_1<j_2 $, since by the \Cref{itm:conf-set-prop-transp-inv} of \Cref{prop:prop-conf-set}, if $ \tau_{\vxa_{i_1}, \vxa_{i_2}, \vxb_{j_1}, \vxb_{j_2}}\in\cKab $ then $ \tau_{\vxa_{i_2}, \vxa_{i_1}, \vxb_{j_2}, \vxb_{j_1}}\in\cKab $.
	By union bound,
	\begin{align}
	&\prob{\exists ((i_1,i_2),(j_1,j_2))\in\binom{[|\cCa'|]}{2}\times\binom{[|\cCb'|]}{2},\;\tau_{\vbfxa_{i_1}, \vbfxa_{i_2}, \vbfxb_{j_1}, \vbfxb_{j_2}} \in\cKab}\notag \\
	\le& \binom{M_1}{2}\binom{M_2}{2}\cdot\card{\cXa}^2\card{\cXb}^2 \cdot2\exp\paren{-\frac{n\eps^2}{3}} \notag \\
	\le& \exp\paren{n\paren{2R_1\ln\card{\cXa} + 2R_2\ln\card{\cXb} - \eps^2/3 + o(1)}}. \label{eqn:positive-product-bound2} 
	\end{align}

	Similar Chernoff-union argument yields
	\begin{align}
	\prob{\exists ((i_1, i_2), j)\in\binom{[|\cCa'|]}{2}\times[|\cCb'|],\;{\tau_{\vbfxa_{i_1}, \vbfxa_{i_2}, \vbfxb_j}}\in\cKa} \le& \exp\paren{n\paren{2R_1\ln\card{\cXa} + R_2\ln\card{\cXb} - \eps^2/3 + o(1)}}, \label{eqn:positive-product-bound3}  \\
	\prob{\exists (i, (j_1, j_2))\in[|\cCa'|]\times\binom{[|\cCb'|]}{2},\;{\tau_{\vbfxa_{i}, \vbfxa_{j_1}, \vbfxb_{j_2}}}\in\cKb} \le& \exp\paren{n\paren{R_1\ln\card{\cXa} + 2R_2\ln\card{\cXb} - \eps^2/3 + o(1)}}. \label{eqn:positive-product-bound4} 
	\end{align}
	It suffices to take $ (R_1, R_2) $ such that 
	$2R_1\ln\card{\cXa} + 2R_2\ln\card{\cXb} - \eps^2/3 <0$.
	For instance, one can take $ R_1 = \frac{\eps^2}{24\ln\card{\cXa}} $ and $ R_2 = \frac{\eps^2}{24\ln\card{\cXb}} $. 
	Then \Cref{eqn:positive-product-bound2,eqn:positive-product-bound3,eqn:positive-product-bound4} are all $ \exp\paren{-\Omega(n)} $. 
	Finally, combining \Cref{eqn:positive-product-bound1,eqn:positive-product-bound2,eqn:positive-product-bound3,eqn:positive-product-bound4}, we get that with probability  $ 1 - \exp(-\Omega(n)) $, $ (\cCa',\cCb') $ is a good code pair of rates $ R(\cCa')\asymp R_1>0 $ and $ R(\cCb')\asymp R_2>0 $. 
\end{proof}

\begin{proof}[Proof of \Cref{itm:prod-pos-zero,itm:prod-zero-pos} in \Cref{lem:achievability-prod}]
We only prove \Cref{itm:prod-pos-zero} and \Cref{itm:prod-zero-pos} follows similarly once the roles of user one and user two are interchanged. 

Suppose $ \ipdistra^{\ot2}\ot\ipdistrb \notin\cKa $. 
We construct a codebook pair $ (\cCa,\cCb) $ as follows. 
The codebook $ \cCb $ consists of only one (arbitrary) codeword $ \vxb\in\cXb^n $ of type $ \ipdistrb $. 
Apparently $ R(\cCb)\to0 $ as $ n\to0 $. 
Indeed, user two cannot even transmit a single bit reliably through the channel. 
The codebook $ \cCa \in\cXa^{M\times n} $ consists of $M$ codewords $ \vbfxa_1,\cdots,\vbfxa_M $ i.i.d. according to $ \ipdistra^\tn $. 
Note that for all $ 1\le i_1<i_2\le M $, $ \expt{\tau_{\vbfxa_{i_1}, \vbfxa_{i_2}, \vxb}} = \ipdistra^{\ot2}\ot\ipdistrb $. 
Indeed, for any $ (\xa_1,\xa_2,\xb)\in\cXa^2\times\cXb $,
\begin{align}
\expt{\tau_{\vbfxa_{i_1}, \vbfxa_{i_2}, \vxb}}(\xa_1,\xa_2,\xb) =& \frac{1}{n}\sum_{k = 1}^n\expt{\indicator{\vbfxa_{i_1}(k) = \xa_1,\vbfxa_{i_2}(k) = \xa_2,\vxb(k) = \xb}} \notag \\
=& \frac{1}{n}\sum_{k = 1}^n\expt{\indicator{\vbfxa_{i_1}(k) = \xa_1}}\expt{\indicator{\vbfxa_{i_2}(k) = \xa_2}}\indicator{\vxb(k) = \xb} \notag \\
=& \prob{\vbfxa(1) = \xa_1}\prob{\vbfxa(1) = \xa_2}\frac{1}{n}\sum_{k = 1}^n\indicator{\vxb(k) = \xb} \notag \\
=& \ipdistra(\xa_1)\ipdistra(\xa_2)\tau_{\vxb}(\xb) \notag \\
=& \paren{\ipdistra^{\ot2}\ot\ipdistrb}(\xa_1,\xa_2,\xb). \notag 
\end{align}
By \Cref{lem:const-comp-conc}, \Cref{eqn:positive-product-bound1} holds for the $ \ipdistra $-constant composition subcode of $ \cCa $, denoted by $ \cCa' $. 
Therefore, $ \cCa' $ has asymptotically the same rate as $ R(\cCa) $.

We define the gap $ \rho_1>0 $ between $ \ipdistra^{\ot2}\ot\ipdistrb $ and $ \cKa $ in the same way as in \Cref{eqn:prod-def-rho}.
Let $ \eps \coloneqq \rho_1/2 $.
Similar Chernoff-union-type argument as before yields
\begin{align}
\prob{\exists(i_1,i_2)\in\binom{[|\cCa'|]}{2},\;\tau_{\vbfxa_{i_1}, \vbfxa_{i_2}, \vxb} \in \cKa} \le& \binom{M}{2}\cdot \card{\cXa}^2\cdot2\exp\paren{-\frac{n\eps^2}{3}} \notag \\
\le& \exp\paren{n\paren{2R_1\ln\card{\cXa} - \eps^2/3 + o(1)}}. \label{eqn:prod-pos-zero-bound} 
\end{align}
Taking $ R_1 = \frac{\eps^2}{12\card{\cXa}} $, we get that with probability $ 1-2^{-\Omega(n)} $, the codebook pair $ (\cCa',\cCb) $ constructed above is good. 
\end{proof}

\subsection{Positive achievable rates via mixtures of product distributions}
\label{sec:positive_rate_mixture}
\begin{lemma}[Positive achievable rates via mixtures product distributions]
\label{thm:achievability}
Fix input distributions $ (\ipdistra, \ipdistrb)\in\ipconstra\times\ipconstrb $. 

\begin{enumerate}
	\item\label[case]{itm:ach-1} If $ \good \ne\emptyset $, then there exist achievable rate pairs $ (R_1, R_2) $ such that $ R_1>0,R_2>0 $.
	\item\label[case]{itm:ach-2} If $ \gooda\setminus\cKa\ne\emptyset $, then there exist achievable rate pairs $ (R_1, 0) $ such that $ R_1>0 $.
	\item\label[case]{itm:ach-3} If $ \goodb\setminus\cKb\ne\emptyset $, then there exist achievable rate pairs $ (0,R_2) $ such that $ R_2>0 $.
\end{enumerate}
\end{lemma}
\begin{proof}[Proof of \Cref{itm:ach-1}]
	By the condition in \Cref{itm:ach-1}, we are able to find a distribution $ \distraabb\in\good $. 
	Suppose $ \distraabb = \sum_{{\ell} = 1}^k \lambda_{\ell} P_{1,{\ell}}^{\ot2}\ot P_{2,{\ell}}^{\ot2} $ for some $ k\in\bZ_{\ge1} $, $ \curbrkt{\lambda_{\ell}}_{{\ell} = 1}^k\subset(0,1] $ with $ \sum_{{\ell} = 1}^k\lambda_{\ell} = 1 $ and distributions $ \curbrkt{P_{1,{\ell}}}_{{\ell} = 1}^k\subset\Delta(\cXa),\curbrkt{P_{2,{\ell}}}_{{\ell} = 1}^k\subset\Delta(\cXb) $. 
	It simultaneously holds that 
	\begin{equation}
	\begin{aligned}
	\distraabb\in& \goodab\setminus\cKab, \\
	\distraab \coloneqq \sum_{\ell = 1}^k\lambda_\ell P_{1,\ell}^{\ot2}\ot P_{2,\ell} \in& \gooda\setminus\cKa, \\
	\distrabb \coloneqq \sum_{\ell = 1}^k\lambda_\ell P_{1,\ell}\ot P_{2,\ell}^{\ot2} \in& \goodb\setminus\cKb. 
	\end{aligned}
	\label{eqn:exist-goodab-distr}
	\end{equation}
	See \Cref{fig:geom-ach-subfig} for the geometry of the aforementioned distributions. 

	Partition $ [n] $ into $k$ subsets $ \cI_1,\cdots,\cI_k $ such that $ |\cI_{\ell}| = \lambda_{\ell}n $ (${\ell}\in[k]$). 
	Now sample a codebook pair $ (\cCa,\cCb)\subseteq\cXa^n\times\cXb^n $ of sizes $ (M_1, M_2) $ in the following way. 
	For $ i = 1,2 $, $ {\ell}\in[k] $, the entries of each codeword of $ \cC_i $ that are in $ \cI_{\ell} $ are i.i.d. according to $ P_{i,{\ell}} $. 
	See \Cref{fig:time-sharing-subfig} for a pictorial explanation of the code construction.

	The proof is similar to that of \Cref{lem:achievability-prod} and the geometry of various distributions is depicted in \Cref{fig:geom-ach-dist-subfig}.
	\begin{figure}[htbp]
	 	\centering
	 	\begin{subfigure}[t]{0.9\linewidth}
	 		\centering
	 		\includegraphics[width=0.65\textwidth]{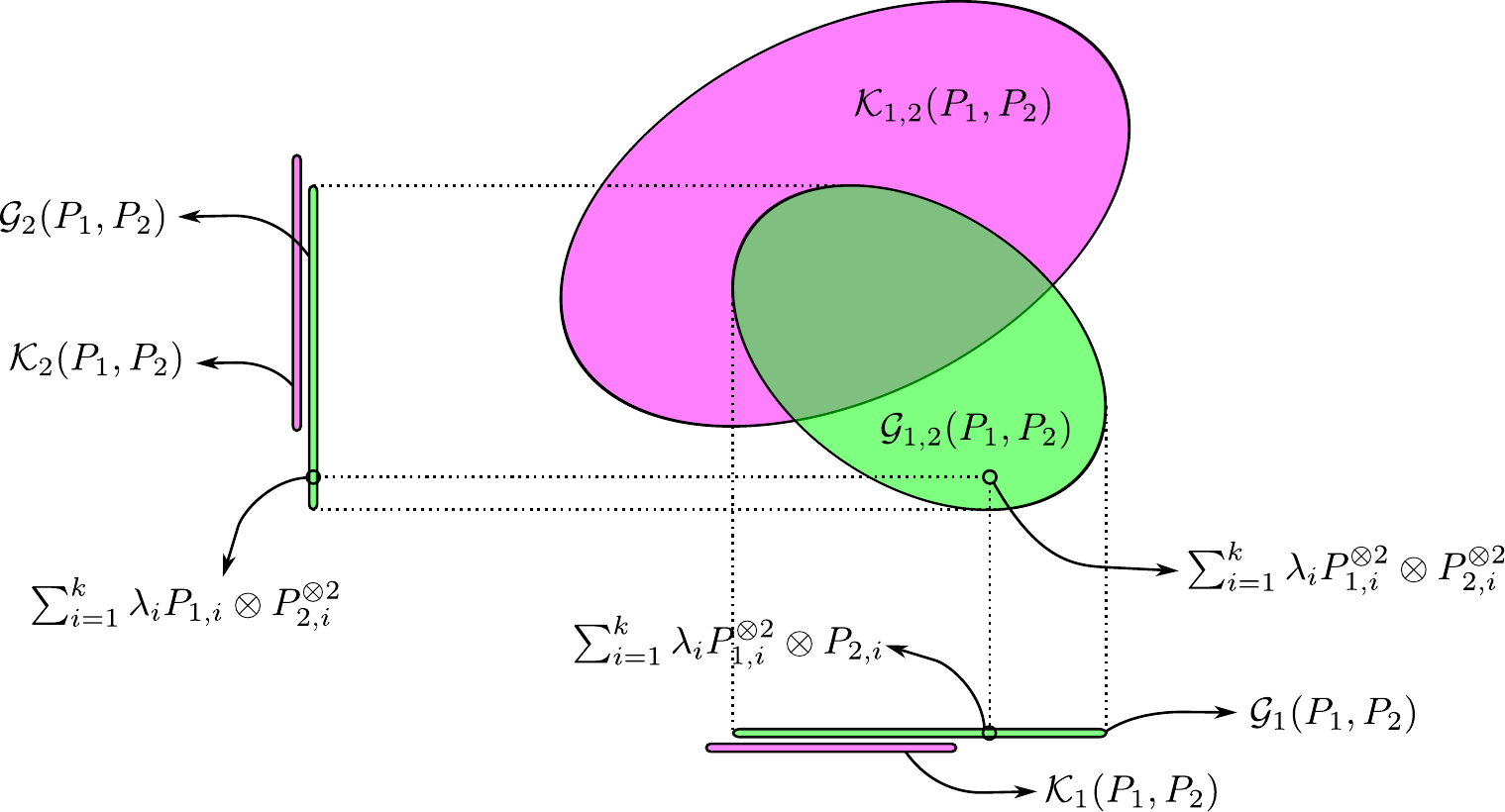}
	 		\caption{By the assumption $ \good \ne \emptyset $, there exists a distribution $ \sum_{i = 1}^k\lambda_iP_{1,i}^{\ot2}\ot P_{2,i}^{\ot2}\notin\cKab $ such that $ \sum_{i = 1}^k\lambda_iP_{1,i}^{\ot2}\ot P_{2,i}\notin\cKa $ and $ \sum_{i = 1}^k\lambda_iP_{1,i}\ot P_{2,i}^{\ot2}\notin\cKb $ (see \Cref{eqn:exist-goodab-distr}). }
	 		\label{fig:geom-ach-subfig}
	 	\end{subfigure}
	 	\\ \vspace{10pt}
	 	\begin{subfigure}[t]{0.9\linewidth}
	 		\centering
	 		\includegraphics[width=0.65\textwidth]{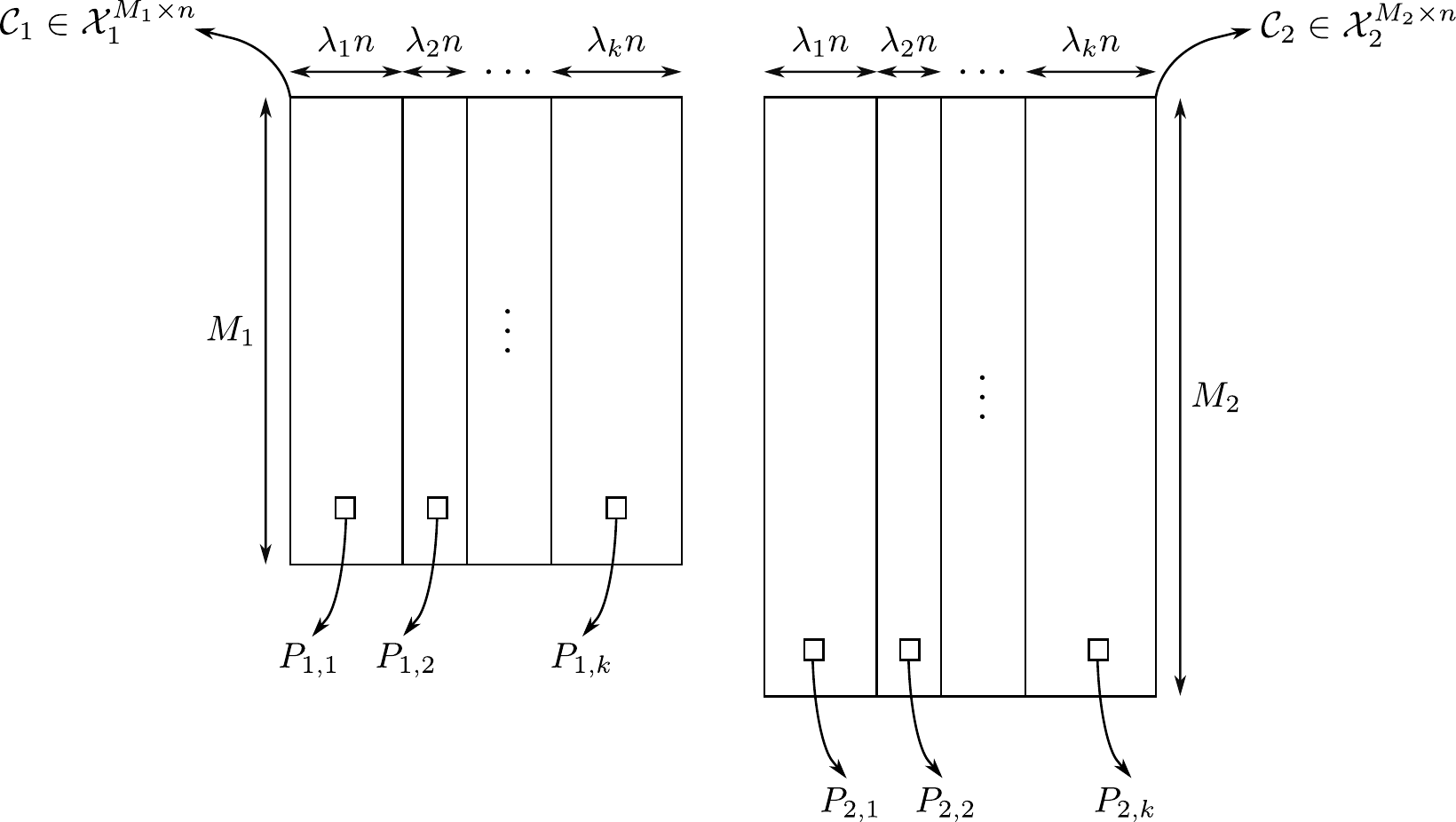}
	 		\caption{A pictorial explanation of our code construction from $ \sum_{i = 1}^k\lambda_iP_{1,i}^{\ot2}\ot P_{2,i}^{\ot2} $. The construction can be viewed as an application of coded time-sharing where the time-sharing sequence is given by the convex combination coefficients $ \curbrkt{\lambda_i}_{i = 1}^k $. For any fixed value $ \ell\in[k] $ of the time-sharing variable, each symbol of $ \cC_i $ is i.i.d. according to $ P_{\ell,i} $. }
	 		\label{fig:time-sharing-subfig}
	 	\end{subfigure}
	 	\\ \vspace{10pt}
	 	\begin{subfigure}[t]{0.9\linewidth}
	 		\centering
	 		\includegraphics[width=0.65\textwidth]{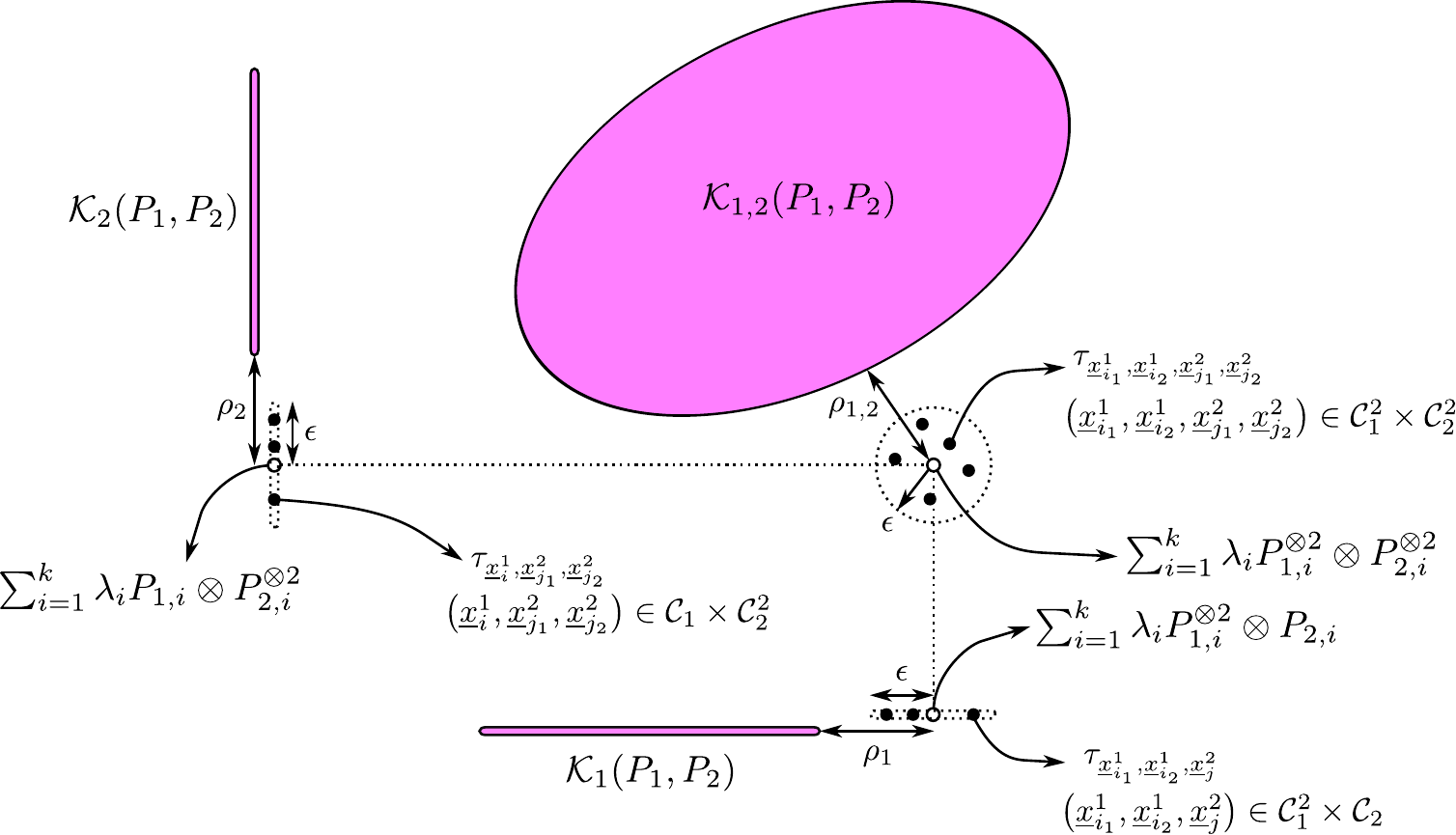}
	 		\caption{By the assumption that $ \sum_{i = 1}^k\lambda_iP_{1,i}^{\ot2}\ot P_{2,i}^{\ot2} $ is $ \rho_{1,2} $-far from $ \cKab $, $ \sum_{i = 1}^k\lambda_iP_{1,i}^{\ot2}\ot P_{2,i} $ is $ \rho_1 $-far from $\cKa$ and $ \sum_{i = 1}^k\lambda_iP_{1,i}\ot P_{2,i}^{\ot2} $ is $ \rho_2 $-far from $ \cKb $, one can show via a Chernoff-union-type argument that all joint types of $ (\cCa,\cCb) $ are $ \eps $-far from the confusability sets and hence $ (\cCa,\cCb) $ attains positive rates and zero error. The gap factors $ \rho_{1,2},\rho_1,\rho_2 $ and $ \eps $ are defined in \Cref{eqn:define-rho-mix}. }
	 		\label{fig:geom-ach-dist-subfig}
	 	\end{subfigure}
	 	\caption{Illustration of the proof of \Cref{itm:ach-1} of \Cref{thm:achievability}. Under the assumption $ \good \ne \emptyset $, the goal is to show the existence of zero-error code pairs $ (\cCa,\cCb) $ of positive rates. }
	 	\label{fig:geom-ach-fig}
	 \end{figure} 
	We can apply similar Chernoff-union argument to the $\ell$-th punctured codes of $ (\cCa,\cCb) $ for each $ \ell\in[k] $ and then take a union bound over $ {\ell} $. 
	Here by the $\ell$-th punctured codes we mean the codes obtained by restricting codewords to $ \cI_\ell $.
	We use $ \vbfx_{i,{\ell}}^{1}\in\cXa^{\lambda_{\ell}n} $ and $ \vbfx_{j,{\ell}}^2\in\cXb^{\lambda_{\ell}n} $ to denote respectively the subsequences of $ \vbfxa_i $ and $ \vbfxb_j $ whose components are in $ \cI_{\ell} $.
	Note that for any  $ 1\le i_1<i_2\le M_1 $ and $ 1\le j_1<j_2\le M_2 $, by \Cref{lem:type-concac},
	\begin{align}
	\expt{\tau_{\vbfxa_{i_1}, \vbfxa_{i_2}, \vbfxb_{j_1}, \vbfxb_{j_2}}} =& \sum_{\ell = 1}^k \lambda_\ell \expt{\tau_{\vbfxa_{i_1,\ell}, \vbfxa_{i_2,\ell}, \vbfxb_{j_1,\ell}, \vbfxb_{j_2,\ell}}} = \sum_{\ell = 1}^k\lambda_\ell P_{1,\ell}^{\ot2}\ot P_{2,\ell}^{\ot2} = \distraabb , \notag \\
	\expt{\tau_{\vbfxa_{i_1}, \vbfxa_{i_2}, \vbfxb_{j_1}}} =& \sum_{\ell = 1}^k \lambda_\ell \expt{\tau_{\vbfxa_{i_1,\ell}, \vbfxa_{i_2,\ell}, \vbfxb_{j_1,\ell}}} = \sum_{\ell = 1}^k\lambda_\ell P_{1,\ell}^{\ot2}\ot P_{2,\ell} = \distraab, \notag \\
	\expt{\tau_{\vbfxa_{i_1}, \vbfxb_{j_1}, \vbfxb_{j_2}}} =& \sum_{\ell = 1}^k \lambda_\ell \expt{\tau_{\vbfxa_{i_1,\ell}, \vbfxb_{j_1,\ell}, \vbfxb_{j_2,\ell}}} = \sum_{\ell = 1}^k\lambda_\ell P_{1,\ell}\ot P_{2,\ell}^{\ot2} = \distrabb. \notag 
	\end{align}

	Let $ \cC_i' $ be the subcode of $ \cC_i $ such that all codewords in $ \cC_i' $ restricted to $ \cI_{\ell} $ are $ P_{i,\ell} $-constant composition ($ i = 1,2,\ell \in[k] $). 
	The size of $ \cC_i' $ can be concentrated similarly as before. 
	\begin{align}
	\expt{|\cC_i'|} =& \sum_{j =1}^{M_i}\prob{\forall \ell\in[k],\;\tau_{\vbfx^i_{j,\ell}} = P_{i,\ell}} 
	= \sum_{j = 1}^{M_i}\prod_{\ell = 1}^k \prob{\tau_{\vbfx^i_{j,\ell}} = P_{i,\ell}} 
	\asymp M_i \prod_{\ell = 1}^k\nu(P_{i,\ell}, \lambda_\ell n)^{-1}. \notag 
	\end{align}
	By \Cref{lem:chernoff},
	\begin{align}
	\prob{|\cC_i'|\notin(1\pm1/2)\expt{|\cC_i'|}} \le& 2\exp\paren{-\frac{M_i}{12\prod_{\ell = 1}^k\nu(P_{i,\ell}, \lambda_\ell n)}}. \label{eqn:const-cc-mix} 
	\end{align}

	Let 
	\begin{equation}
	\begin{aligned}
	\rho_{1,2} \coloneqq& \distinf{\distraabb}{\cKab} > 0,  \\
	\rho_1 \coloneqq& \distinf{\distraab}{\cKa} > 0,  \\
	\rho_2 \coloneqq& \distinf{\distrabb}{\cKb} > 0,  \\
	\eps \coloneqq& \frac{1}{2}\min\curbrkt{\rho_{1,2}, \rho_1, \rho_2}>0. 
	\end{aligned}
	\label{eqn:define-rho-mix}
	\end{equation}
	For any $ 1\le i_1<i_2\le M_1 $ and $ 1\le j_1<j_2\le M_2 $, 
	\begin{align}
	 \prob{\tau_{\vbfxa_{i_1}, \vbfxa_{i_2}, \vbfxb_{j_1}, \vbfxb_{j_2}} \in\cKab} 
	\le& \prob{\exists \ell\in[k],\; \distinf{\tau_{\vbfxa_{i_1,\ell}, \vbfxa_{i_2,\ell}, \vbfxb_{j_1,\ell}, \vbfxb_{j_2,\ell}}}{P_{1,\ell}^{\ot2}\ot P_{2,\ell}^{\ot2}} \ge \eps} \label[ineq]{eqn:reason} \\
	\le& k \cdot \card{\cXa}^2\card{\cXb}^2 \cdot2\exp\paren{-\frac{n\eps^2}{3}}. \label{eqn:mix-case1-joint} 
	\end{align}
	\Cref{eqn:reason} follows since $ \distinf{\tau_{\vbfxa_{i_1,\ell}, \vbfxa_{i_2,\ell}, \vbfxb_{j_1,\ell}, \vbfxb_{j_2,\ell}}}{P_{1,\ell}^{\ot2}\ot P_{2,\ell}^{\ot2}} < \eps $ for all $ \ell\in[k] $ implies 
	\begin{align}
	& \distinf{\tau_{\vbfxa_{i_1}, \vbfxa_{i_2}, \vbfxb_{j_1}, \vbfxb_{j_2}}}{\distraabb} \notag \\
	=& \max_{(\xa_1,\xa_2,\xb_1,\xb_2)\in\cXa^2\times\cXb^2}\abs{\sum_{\ell = 1}^k\lambda_\ell \tau_{\vbfxa_{i_1,\ell}, \vbfxa_{i_2,\ell}, \vbfxb_{j_1,\ell}, \vbfxb_{j_2,\ell}}(\xa_1,\xa_2,\xb_1,\xb_2) - \sum_{\ell = 1}^k\lambda_\ell P_{1,\ell}^{\ot2}\ot P_{2,\ell}^{\ot2}(\xa_1,\xa_2,\xb_1,\xb_2) } \notag \\
	\le& \sum_{\ell = 1}^k\lambda_\ell \max_{(\xa_1,\xa_2,\xb_1,\xb_2)\in\cXa^2\times\cXb^2} \abs{\tau_{\vbfxa_{i_1,\ell}, \vbfxa_{i_2,\ell}, \vbfxb_{j_1,\ell}, \vbfxb_{j_2,\ell}}(\xa_1,\xa_2,\xb_1,\xb_2) - P_{1,\ell}^{\ot2}\ot P_{2,\ell}^{\ot2}(\xa_1,\xa_2,\xb_1,\xb_2)} \notag \\
	=& \sum_{\ell = 1}^k\lambda_\ell \distinf{\tau_{\vbfxa_{i_1,\ell}, \vbfxa_{i_2,\ell}, \vbfxb_{j_1,\ell}, \vbfxb_{j_2,\ell}}}{P_{1,\ell}^{\ot2}\ot P_{2,\ell}^{\ot2}} 
	< \eps < \rho_{1,2}, \notag 
	\end{align}
	which in turn implies $ \tau_{\vbfxa_{i_1}, \vbfxa_{i_2}, \vbfxb_{j_1}, \vbfxb_{j_2}} \notin\cKab $. 
	In \Cref{eqn:mix-case1-joint}, we took a union bound over $ \ell\in[k] $ where $ k = \cO(1) $. 

	Similarly, we have
	\begin{align}
	\prob{\tau_{\vbfxa_{i_1}, \vbfxa_{i_2}, \vbfxb_{j}} \in\cKa} 
	\le& k \cdot \card{\cXa}^2\card{\cXb} \cdot2\exp\paren{-\frac{n\eps^2}{3}}, \label{eqn:mix-case1-marg1}
	\end{align}
	for all $ 1\le i_1<i_2\le M_1 $ and $ 1\le j\le M_2 $; 
	and 
	\begin{align}
	\prob{\tau_{\vbfxa_{i}, \vbfxb_{j_1}, \vbfxb_{j_2}} \in\cKb} 
	\le& k \cdot \card{\cXa}\card{\cXb}^2 \cdot2\exp\paren{-\frac{n\eps^2}{3}}, \label{eqn:mix-case1-marg2}
	\end{align}
	for all $ 1\le i\le M_1 $ and $ 1\le j_1<j_2\le M_2 $. 
	Taking further union bounds on \Cref{eqn:mix-case1-joint,eqn:mix-case1-marg1,eqn:mix-case1-marg2} over $ ((i_1,i_2), (j_1,j_2)) $, $ ((i_1,i_2), j) $ and $ (i,(j_1,j_2)) $ respectively ensures that \Cref{eqn:positive-product-bound2,eqn:positive-product-bound3,eqn:positive-product-bound4} still hold. 
	The rest of the proof remains the same and we get a good code pair $ (\cCa',\cCb') $ of rate $ R(\cCa')>0,R(\cCb')>0 $. 
\end{proof}

\begin{proof}[Proof of \Cref{itm:ach-2,itm:ach-3}]
We only prove \Cref{itm:ach-2} since \Cref{itm:ach-3} is the same once the roles of the first and second users are swapped.

Suppose $ \distraab\in\gooda\setminus\cKa $ has a decomposition
$ \distraab = \sum_{\ell = 1}^k\lambda_\ell P_{1,\ell}^{\ot2}\ot P_{2,\ell} $ for some $ k\in\bZ_{\ge1} $, $ \curbrkt{\lambda_\ell}_{\ell = 1}^k\subset(0,1] $ with $ \sum_{\ell = 1}^k\lambda_\ell = 1 $ and $ \curbrkt{P_{1,\ell}}_{\ell = 1}^k\subset\Delta(\cXa^2) $, $ \curbrkt{P_{2,\ell}}_{\ell = 1}^k\subset\Delta(\cXb^2) $. 

Partition $ [n] $ into $k$ subsets $ \cI_1,\cdots,\cI_k $ such that $ |\cI_{\ell}| = \lambda_{\ell}n $ (${\ell}\in[k]$). 
Construct a codebook pair $ (\cCa,\cCb) $ as follows. 
The second codebook $ \cCb $ only consists of one (arbitrary) codeword $ \vxb\in\cXb^n $ satisfying the following property. 
Let $ \vxb_{\ell}\in\cXb^{\lambda_\ell n} $ denote the subsequence of $ \vxb $ restricted to $ \cI_\ell $. 
For each $ \ell\in[k] $, $ \tau_{\vxb_\ell} = P_{2,\ell} $. 
The first codebook $ \cCa\in\cXa^{M\times n} $ consists of $M$ codewords $ \vbfxa_1,\cdots,\vbfxa_M $, where for each $ i\in [M] $ and $ \ell\in[k] $, $ \vbfxa_{i,\ell}\iid P_{1,\ell}^{\ot(\lambda_\ell n)} $. 
Note that for all $ 1\le i_1<i_2\le M $, 
$\expt{\tau_{\vbfxa_{i_1}, \vbfxa_{i_2}, \vxb}} = \distraab$. 
Let $ \cCa' $ be the subcode of $ \cCa $ whose codewords restricted to $ \cI_\ell $ are all $ P_{1,\ell} $-constant composition ($ \ell\in[k] $). 
For $ \cCa' $, \Cref{eqn:const-cc-mix} still holds. 
Therefore, $ R(\cCa')\asymp R(\cCa) $ ($ n\to\infty $). 
We define $\rho_1 $ in the same way as in \Cref{eqn:define-rho-mix}.
Let $ \eps\coloneqq\rho_1/2 $. 
Since 
\begin{align}
\distinf{\tau_{\vbfxa_{i_1}, \vbfxa_{i_2}, \vxb}}{\distraab} \le& \sum_{\ell = 1}^k\lambda_\ell\distinf{\tau_{\vbfxa_{i_1,\ell}, \vbfxa_{i_2,\ell}, \vxb_{\ell}}}{P_{1,\ell}^{\ot2}\ot P_{2,\ell}}, \notag 
\end{align}
a Chernoff-union bound gives
\begin{align}
\prob{\tau_{\vbfxa_{i_1}, \vbfxa_{i_2}, \vxb}\in\cKa}\le& 
\prob{\distinf{\tau_{\vbfxa_{i_1}, \vbfxa_{i_2}, \vxb}}{\distraab}\ge\eps} \notag \\
\le& \prob{\exists \ell\in[k],\;\distinf{\tau_{\vbfxa_{i_1,\ell}, \vbfxa_{i_2,\ell}, \vxb_{\ell}}}{P_{1,\ell}^{\ot2}\ot P_{2,\ell}}\ge\eps} \notag \\
\le& k\cdot\card{\cXa}^2\cdot2\exp\paren{-\frac{n\eps^2}{3}}. \notag
\end{align}
Since $ k $ is a constant independent of $n$, a union bound over $ (i_1,i_2)\in\binom{[|\cCa'|]}{2} $ gives \Cref{eqn:prod-pos-zero-bound}. 
Under a proper choice of $ R_1>0 $, we get that $ (\cCa',\cCb) $ is a good codebook pair with probability at least $ 1 - 2^{-\Omega(n)} $. 
\end{proof}

\begin{remark}
\label{rk:coded-time-sharing}
In the above proof of \Cref{thm:achievability}, the partition $ \curbrkt{\cI_\ell}_{\ell = 1}^k $ can be thought of as a \emph{time-sharing} sequence $ \vu\in[k]^n $ of type $ P_\bfu $ given by the coefficients $ \curbrkt{\lambda_i}_{i = 1}^k $ of the convex combination. 
That is, $ P_\bfu(u) = \lambda_u $ for any $ u\in[k] $. 
This particular type of time-sharing scheme is known as the \emph{coded} time-sharing in the literature \cite{pereg-steinberg-2019-avmac}. 
As explained in \cite[Remark 6]{pereg-steinberg-2019-avmac}, the classical \emph{operational} time-sharing in network information theory does not work for (oblivious) arbitrarily varying channels with constraints. 
This is because the adversary can concentrate his power on coordinates in a single $ \cI_\ell $. 
This effectively increases the noise level in $ \cI_\ell $ significantly and the $\ell$-th component codebook in the time-sharing is not necessarily resilient to this effective level of noise. 
The above argument also applies to the omniscient adversarial channel model. 
More discussions on the ``non-tensorization'' of good codes for adversarial channels and its implications to single-letterization of capacity expressions can be found in \Cref{itm:open-tensorization} of \Cref{sec:concl-rmk-open-prob}. 
These phenomena suggest that the capacity region of adversarial channels does not have to be convex in general (see \Cref{rk:nonconvex}). 

Furthermore, we emphasize the following point in the above achievability proof. 
Each component $ P_{1,\ell} $ and $ P_{2,\ell} $ of the convex combinations is not necessarily non-confusable, i.e., $ P_{1,\ell}^{\ot2}\ot P_{2,\ell}^{\ot2} $, $ P_{1,\ell}^{\ot2}\ot P_{2,\ell} $ or $ P_{1,\ell}\ot P_{2,\ell}^{\ot2} $ may be confusable. 
Nonetheless, it is only desired that their convex combinations are non-confusable. 
\end{remark}

\begin{remark}
\label{rk:one-codeword}
In the above proof of \Cref{itm:ach-2,itm:ach-3}, the transmitter with zero capacity cannot even reliably transmit a single bit through the MAC since the codebook contains only one codeword. 
Such achievability proofs go through as long as there exist non-marginally confusable distributions. 
In contrast, in the AVMAC setting \cite{pereg-steinberg-2019-avmac}, besides non-marginal symmetrizability, non-joint symmetrizability is a necessary condition for achieving any positive rate even individually instead of jointly. 
More discussions on the differences between our results and Pereg--Steinberg's \cite{pereg-steinberg-2019-avmac} can be found in \Cref{sec:comparison-our-peregsteinberg}. 
\end{remark}

\subsection{Inner bounds via product distributions}
\label{sec:inner_bound_prod_distr}
\begin{lemma}[Inner bounds via product distributions]
\label{lem:inner_bound_prod_distr}
Fix input distributions $ (\ipdistra, \ipdistrb)\in\ipconstra\times\ipconstrb $. 
\begin{enumerate}
	\item\label[case]{itm:rate-prod-pos-pos} If $ \prodab\notin\cKab $, $ \ipdistra\notin\cKa $ and $ \ipdistrb\notin\cKb $, then rate pairs $ (R_1, R_2)\in\bR_{\ge0}^2 $ satisfying 
	\begin{equation}
	\begin{aligned}
	R_1 \le& D(\ipdistra,\ipdistrb) - \wh D(\ipdistra,\ipdistrb)  \\
	R_2 \le& D(\ipdistra,\ipdistrb) - \wh D(\ipdistra,\ipdistrb) \\
	R_1+R_2\le& \wh D(\ipdistra,\ipdistrb) 
	\end{aligned}
	\label{eqn:achieve-product-12}
	\end{equation}
	are achievable, where 
	\begin{align}
	D(\ipdistra,\ipdistrb)\coloneqq& \min_{ \distraabb\in \cKab } \kl{ \distraabb}{ \ipdistra^{\ot2}\ot \ipdistrb^{\ot2} } \notag \\
	\wh D(\ipdistra,\ipdistrb) \coloneqq& \min\curbrkt{
	\min_{P_{\bfxa_1,\bfxa_2,\bfxb} \in \cKa} \kl{P_{\bfxa_1,\bfxa_2,\bfxb} }{ \ipdistra^{\ot2}\ot \ipdistrb },
	\min_{P_{\bfxa,\bfxb_1,\bfxb_2} \in \cKb} \kl{ P_{\bfxa,\bfxb_1,\bfxb_2} }{ \ipdistra\ot \ipdistrb^{\ot2} } }. \notag 
	\end{align}

	\item\label[case]{itm:rate-prod-pos-zero} If $ \ipdistra^{\ot2}\ot \ipdistrb^{\ot2} \notin\cKab $, $ \ipdistra^{\ot2} \notin \cKa $ and $\ipdistrb^{\ot2}\in \cKb $, then rate pairs $ (R_1, 0) $ satisfying 
	\begin{align}
	0\le R_1 \le \min_{P_{\bfxa_1, \bfxa_2, \bfxb} \in \cKa} \kl{P_{\bfxa_1, \bfxa_2, \bfxb}}{\ipdistra^{\ot2}\ot \ipdistrb} \label{eqn:rate-prod-pos-zero} 
	\end{align}
	are achievable.

	\item\label[case]{itm:rate-prod-zero-pos} If $ \ipdistra^{\ot2} \ot\ipdistrb^{\ot2} \notin\cKab $, $ \ipdistra^{\ot2}\in\cKa $ and $ \ipdistrb^{\ot2}\notin\cKb $, then rate pairs $ (0, R_2) $ satisfying 
	\begin{align}
	0 \le R_2 \le \min_{P_{\bfxa,  \bfxb_1, \bfxb_2} \in\cKb }\kl{P_{\bfxa,  \bfxb_1, \bfxb_2}}{\ipdistra\ot\ipdistrb^{\ot2}} \label{eqn:rate-prod-zero-pos} 
	\end{align}
	are achievable. 
\end{enumerate}
\end{lemma}

\begin{corollary}[Inner bounds on capacity region]
\label{cor:in-bd-prod}
Let $ \mactwofull $ be a two-user omniscient adversarial MAC. 
The capacity region of $ \mactwo $ contains as a subset the following region
\begin{align}
\bigcup_{\substack{(\ipdistra,\ipdistrb)\in\ipconstra\times\ipconstrb \\ \text{conditions in \Cref{itm:rate-prod-pos-pos} are satisfied}}}\curbrkt{(R_1,R_2):(R_1,R_2)\text{ satisfies }\Cref{eqn:achieve-product-12}} \notag \\
\cup\bigcup_{\substack{(\ipdistra,\ipdistrb)\in\ipconstra\times\ipconstrb \\ \text{conditions in \Cref{itm:rate-prod-pos-zero} are satisfied}}}\curbrkt{(R_1,0):R_1\text{ satisfies \Cref{eqn:rate-prod-pos-zero}}} \notag \\
\cup\bigcup_{\substack{(\ipdistra,\ipdistrb)\in\ipconstra\times\ipconstrb \\ \text{conditions in \Cref{itm:rate-prod-zero-pos} are satisfied}}}\curbrkt{(0,R_2):R_2\text{ satisfies \Cref{eqn:rate-prod-zero-pos}}}. \notag 
\end{align}
\end{corollary}

\begin{proof}[Proof of \Cref{itm:rate-prod-pos-pos}]
	Sample a random code pair $ \paren{\cCa,\cCb}\subseteq\cXa^n\times\cXb^n $ of sizes $ (M_1, M_2) $, where $ \cC_i $ consists of codewords $ \vbfx^i_1,\cdots,\vbfx^i_{M_i} $ i.i.d. according to $ P_i^\tn $ ($i = 1,2$). 
	By \Cref{lem:type-class-conc}, the the expected number of codewords in $ \cC_i $ of type $ P_i $ is asymptotically $ M_i/\nu(P_i, n) $.
	For any $ 1\le i_1<i_2\le M_1 $ and $ 1\le j_1<j_2\le M_2 $, by \Cref{lem:sanov}, 
	\begin{align}
	\prob{\tau_{\vbfxa_{i_1}, \vbfxa_{i_2}, \vbfxb_{j_1}, \vbfxb_{j_2}}\in\cKab} \doteq& \sup_{P_{\bfxa_1,\bfxa_2,\bfxb_1,\bfxb_2}\in\cKab} 2^{-n\kl{P_{\bfxa_1,\bfxa_2,\bfxb_1,\bfxb_2}}{\ipdistra^{\ot2}\ot\ipdistrb^{\ot2}}}, \notag \\
	\prob{\tau_{\vbfxa_{i_1}, \vbfxa_{i_2}, \vbfxb_{j_1}}\in\cKa} \doteq& \sup_{P_{\bfxa_1,\bfxa_2,\bfxb}\in\cKa} 2^{-n\kl{P_{\bfxa_1,\bfxa_2,\bfxb}}{\ipdistra^{\ot2}\ot\ipdistrb}}, \notag \\
	\prob{\tau_{\vbfxa_{i_1}, \vbfxb_{j_1}, \vbfxb_{j_2}}\in\cKb} \doteq& \sup_{P_{\bfxa,\bfxb_1,\bfxb_2}\in\cKb} 2^{-n\kl{P_{\bfxa,\bfxb_1,\bfxb_2}}{\ipdistra\ot\ipdistrb^{\ot2}}}. \notag 
	\end{align}
	Hence the expected number of confusable tuples $ (\vbfxa_{i_1}, \vbfxa_{i_2}, \vbfxb_{j_1}, \vbfxb_{j_2}) $, $ (\vbfxa_{i_1}, \vbfxa_{i_2}, \vbfxb_j) $ and $ (\vbfxa_i, \vbfxb_{j_1}, \vbfxb_{j_2}) $ is respectively
	\begin{align}
	\binom{M_1}{2}\binom{M_2}{2} 2^{-n\inf \kl{P_{\bfxa_1,\bfxa_2,\bfxb_1,\bfxb_2}}{\ipdistra^{\ot2}\ot\ipdistrb^{\ot2}}} \le& M_1^2M_2^22^{-n\inf \kl{P_{\bfxa_1,\bfxa_2,\bfxb_1,\bfxb_2}}{\ipdistra^{\ot2}\ot\ipdistrb^{\ot2}}}, \notag \\
	\binom{M_1}{2}M_22^{-n \inf \kl{P_{\bfxa_1,\bfxa_2,\bfxb}}{\ipdistra^{\ot2}\ot\ipdistrb}} \le& M_1^2M_2 2^{-n \inf \kl{P_{\bfxa_1,\bfxa_2,\bfxb}}{\ipdistra^{\ot2}\ot\ipdistrb}}, \notag \\
	M_1\binom{M_2}{2}2^{-n\inf \kl{P_{\bfxa,\bfxb_1,\bfxb_2}}{\ipdistra\ot\ipdistrb^{\ot2}}} \le& M_1M_2^22^{-n\inf \kl{P_{\bfxa,\bfxb_1,\bfxb_2}}{\ipdistra\ot\ipdistrb^{\ot2}}} . \notag
	\end{align}
	Pick $ M_1,M_2 $ such that 
	\begin{align}
	M_1^2M_2^22^{-n\inf \kl{P_{\bfxa_1,\bfxa_2,\bfxb_1,\bfxb_2}}{\ipdistra^{\ot2}\ot\ipdistrb^{\ot2}}} \le& \min\curbrkt{ \frac{M_1}{3\nu(\ipdistra, n)}, \frac{M_2}{3\nu(\ipdistrb, n)} }, \notag \\
	M_1^2M_2 2^{-n \inf \kl{P_{\bfxa_1,\bfxa_2,\bfxb}}{\ipdistra^{\ot2}\ot\ipdistrb}} \le& \frac{M_1}{3\nu(\ipdistra, n)} \notag \\
	M_1M_2^22^{-n\inf \kl{P_{\bfxa,\bfxb_1,\bfxb_2}}{\ipdistra\ot\ipdistrb^{\ot2}}} \le& \frac{M_2}{3\nu(\ipdistrb, n)}. \notag 
	\end{align}
	This can be satisfied if
	\begin{align}
	2R_1 + 2R_2 - \inf\kl{P_{\bfxa_1,\bfxa_2,\bfxb_1,\bfxb_2}}{\ipdistra^{\ot2}\ot\ipdistrb^{\ot2}} \le& \min\curbrkt{ R_1 , R_2 } - o(1), \notag \\
	2R_1 + R_2 - \inf \kl{P_{\bfxa_1,\bfxa_2,\bfxb}}{\ipdistra^{\ot2}\ot\ipdistrb} \le& R_1 - o(1), \notag \\
	R_1 + 2R_2 - \inf \kl{P_{\bfxa,\bfxb_1,\bfxb_2}}{\ipdistra\ot\ipdistrb^{\ot2}} \le& R_2 - o(1), \notag 
	\end{align}
	i.e.,
	\begin{align}
	R_1+2R_2\le&\inf\kl{P_{\bfxa_1,\bfxa_2,\bfxb_1,\bfxb_2}}{\ipdistra^{\ot2}\ot\ipdistrb^{\ot2}} -o(1), \notag \\
	2R_1+R_2\le&\inf\kl{P_{\bfxa_1,\bfxa_2,\bfxb_1,\bfxb_2}}{\ipdistra^{\ot2}\ot\ipdistrb^{\ot2}} -o(1), \notag \\
	R_1+R_2\le&\min\curbrkt{\inf \kl{P_{\bfxa_1,\bfxa_2,\bfxb}}{\ipdistra^{\ot2}\ot\ipdistrb} -o(1), \inf \kl{P_{\bfxa,\bfxb_1,\bfxb_2}}{\ipdistra\ot\ipdistrb^{\ot2}} -o(1) }. \notag 
	\end{align}
	That is, it suffices to take $ (R_1, R_2) $ satisfying \Cref{eqn:achieve-product-12} (as $ n\to\infty $). 

	Now, we remove all codewords from $ \cCa $ and $ \cCb $ whose types are not $ \ipdistra $ and $ \ipdistrb $ respectively.
	For all $ 1\le i_1<i_2\le M_1 $ and $ 1\le j_1<j_2\le M_2 $, we also remove 
	\begin{enumerate}
		\item one of $ (\vbfxa_{i_1}, \vbfxa_{j_2}) $ from $ \cCa $ and one of $ (\vbfxb_{j_1}, \vbfxb_{j_2}) $ from $ \cCb $ if $ \tau_{\vbfxa_{i_1}, \vbfxa_{i_2}, \vbfxb_{j_1}, \vbfxb_{j_2}} \in\cKab $;
		\item one of $ (\vbfxa_{i_1}, \vbfxa_{i_2}) $ from $ \cCa $ if $ \tau_{\vbfxa_{i_1}, \vbfxa_{i_2}, \vbfxb_{j_1}}\in\cKa $;
		\item one of $ (\vbfxb_{j_1}, \vbfxb_{j_2}) $ from $ \cCb $ if $ \tau_{\vbfxa_{i_1}, \vbfxb_{j_1}, \vbfxb_{j_2}}\in\cKb $.
	\end{enumerate}
	After the removal, $ \paren{\cCa,\cCb} $ becomes a good code pair. 
	In total, the expected number of codewords we removed from $ \cC_i $ is at most 
	\begin{align}
	M_i - \frac{M_i}{\nu(P_i,n)} + \frac{M_i}{3\nu(P_i,n)} + \frac{M_i}{3\nu(P_i,n)} = M_i - \frac{M_i}{3\nu(P_i,n)} \notag 
	\end{align}
	for $ i = 1,2 $. 
	Therefore, $(R_1,R_2) $ is preserved after the removal. 
	Noting that we have exhibited the existence of code pairs that attain zero error for $ \mactwo $ with desired rates, we finish the proof.
\end{proof}

\begin{proof}[Proof of \Cref{itm:rate-prod-pos-zero,itm:rate-prod-zero-pos}]
We only prove \Cref{itm:rate-prod-pos-zero}. 
\Cref{itm:rate-prod-zero-pos} will follow verbatim. 
Let $ \vxb\in\cXb^n $ be an arbitrary codeword of type $ \ipdistrb $. 
The codebook $ \cCb $ only consists of $ \vxb $. 
The codebook $ \cCa $ consists of $ M $ codewords $ \vbfxa_1,\cdots,\vbfxa_M $ i.i.d. according to $ \ipdistra^\tn $. 
Again, the expected number of codewords in $ \cCa $ of type $ \ipdistra $ is asymptotically $ M/\nu(\ipdistra,n) $.  
By \Cref{lem:sanov}, for any $ 1\le i_1<i_2\le M $, 
\begin{align}
\prob{\tau_{\vbfxa_{i_1}, \vbfxa_{i_2}, \vxb} \in\cKa} \doteq& \sup_{\distraab\in\cKa}2^{-n\kl{\distraab}{\ipdistra^{\ot2}\ot\ipdistrb}}. \notag 
\end{align}
Hence the expected number of confusable tuples $ (\vbfxa_{i_1},\vbfxa_{i_2},\vxb) $ is 
\begin{align}
\binom{M}{2}2^{-n\inf\kl{\distraab}{\ipdistra^{\ot2}\ot\ipdistrb}} \le& M^22^{-n\inf\kl{\distraab}{\ipdistra^{\ot2}\ot\ipdistrb}}. \notag 
\end{align}
Pick $M$ such that
\begin{align}
M^22^{-n\inf\kl{\distraab}{\ipdistra^{\ot2}\ot\ipdistrb}} \le& \frac{M}{2\nu(\ipdistra,n)}. \notag 
\end{align}
It suffices to take
\begin{align}
2R_1 - \inf\kl{\distraab}{\ipdistra^{\ot2}\ot\ipdistrb} \le& R_1 - o(1), \notag 
\end{align}
i.e., $ R_1 $ asymptotically satisfies \Cref{eqn:rate-prod-pos-zero}. 

We then remove all codewords from $ \cCa $ which have type different from $ \ipdistra $. 
We also remove $ \vbfxa_{i_1} $ if $ \tau_{\vbfxa_{i_1},\vbfxa_{i_2},\vxb}\in\cKa $ for some $ i_1<i_2\le M $. 
After removal we get a constant composition codebook pair that attains zero error.
The expected number of codewords we removed from $ \cCa $ is at most $ M - M/\nu(\ipdistra,n) + M/2\nu(\ipdistra,n) = M - M/2\nu(\ipdistra,n) $. 
Therefore, the removal does not (asymptotically) change the rate. 
This finishes the proof. 
\end{proof}

\begin{remark}
In \Cref{lem:inner_bound_prod_distr}, we did not obtain a pentagon region defined by three mutual information terms as is commonly seen in problems regarding MACs. 
It is perhaps due to our crude expurgation strategy. 
We believe that our inner bounds can be improved by employing more careful expurgation strategies (see \Cref{itm:open-better-inner-bounds} in \Cref{sec:concl-rmk-open-prob}). 
\end{remark}




\section{Converse, \Cref{itm:conv-pos-pos} in \Cref{thm:converse}}
\label{sec:conv-pos-pos}



In this section, we assume that $ \good = \emptyset $. 
Let $ (\cCa,\cCb)\subseteq\cXa^n\times\cXb^n $ be any good codebook pair. 
Without loss of rate, we assume that $ \cCa $ is $\ipdistra$-constant composition and $ \cCb $ is $ \ipdistrb $-constant composition. 
Our goal is to show that $ R(\cCa) $ and $ R(\cCb) $ cannot be simultaneously positive. 
In fact, we will show that at least one of $ M_1\coloneqq|\cCa| $ and $ M_2\coloneqq|\cCb| $ is bounded from above by a \emph{constant} (independent of $n$). 

\subsection{Subcode pair extraction}
\label{sec:subcodepair_extraction_ramsey}

\begin{definition}[Bipartite, uniform, complete hypergraphs]
\label{def:bipartite_unif_complete_hypergraph}
A hypergraph $ \cH = (\cV, \cE) $ is called \emph{$ (N_1 ,N_2) $-bipartite} if it is bipartite with $ \cV = \cV_1 \sqcup \cV_2 $ where $|\cV_1| = N_1 $ and $ |\cV_2| = N_2 $.
It is called \emph{$ ( k_1 , k_2) $-uniform} if every hyperedge contains $ k_1 $ vertices in $ \cV_1 $ and $ k_2 $ vertices in $ \cV_2 $. 
It is called \emph{complete} if every $  k_1  $-tuple of vertices in $ \cV_1 $ and every $ k_2 $-tuple of vertices in $ \cV_2 $ are connected. 
\end{definition}

\begin{theorem}[Bipartite hypergraph Ramsey's theorem \cite{bipartite-ramsey}]
\label{lem:hypergraph_ramsey}
Let $ N_1, N_2, D $ be integers that are at least 2. 
There exist  constants $ K_1 = K_2(N_1, N_2, D) $ and $ K_2 = K_2(N_1, N_2, D) $ such that for every $ (M_1, M_2) $-bipartite $(2,2)$-uniform complete hypergraph $ \cH = ((\cV_1, \cV_2), \cE ) $ such that $ |\cV_1| = M_1 \ge K_1 $ and $ |\cV_2| = M_2 \ge K_2 $, for every  $D$-coloring of $ \cE $, there must exist $ \cV_1' \subseteq \cV_1 $ and $ \cV_2' \subseteq \cV_2 $ such that $ |\cV_1'| \ge N_1, |\cV_2'| \ge N_2 $ and all hyperedges crossing $ \cV_1' $ and $ \cV_2' $ have the same color. 
\end{theorem}

\begin{lemma}[Subcode pair extraction]
\label{lem:subcodepair_extraction}
For any code pair $ \paren{\cCa, \cCb} = \paren{ \curbrkt{\vxa_{k}}_{k = 1}^{M_1}, \curbrkt{\vxb_{\ell}}_{\ell = 1}^{M_2} } $ of sizes $ M_1 $ and $ M_2 $, respectively, there exists a subcode pair $ \paren{\cCa', \cCb'} = \paren{ \curbrkt{\vxa_{i}}_{i = 1}^{M_1'}, \curbrkt{\vxb_{j}}_{j = 1}^{M_2'} } $ of sizes $ M_1' \ge f_1(|\cXa|,|\cXb|, \eta, M_1,M_2)\xrightarrow{M_1\to\infty}\infty $ and $ M_2' \ge f_2(|\cXa|,|\cXb|, \eta, M_1,M_2)\xrightarrow{M_2\to\infty}\infty $, respectively, and there exists a distribution $ \distraabb \in \cJab $ such that, for all $ 1 \le i_1 < i_2 \le M_1' $ and $ 1 \le j_1 < j_2 \le M_2' $, it holds that $ \distinf{\tau_{ \vxa_{i_1}, \vxa_{i_2}, \vxb_{j_1}, \vxb_{j_2} }}{ \distraabb } \le \eta $. 
\end{lemma}

\begin{proof}
To apply \Cref{lem:hypergraph_ramsey}, we build an $ (M_1,M_2) $-bipartite $ (2,2) $-uniform complete hypergraph $ \cH $. 
The left and right vertex sets of $ \cH $ are the codewords in $ \cCa $ and the codewords in $ \cCb $ respectively. 
Every pair of codewords $ (\vxa_{i_1}, \vxa_{i_2})\in\binom{\cCa}{2} $ (where $ 1\le i_1<i_2<M_1 $) in the left vertex set is connected to all pairs of codewords $ (\vxb_{j_1}, \vxb_{j_2})\in\binom{\cCb}{2} $ (for all $ 1\le j_1<j_2\le M_2 $) in the right vertex set. 

We now color all hyperedges of $ \cH $ using distributions in $ \cJab $. 
To this end, 
we first take an $\eta$-net $\cN$ of $ \cJab $ with respect to $ d_{\infty} $.
By \Cref{lem:bound-net}, $ D\coloneqq\card{\cN} $ can be made no larger than $ \paren{\frac{\card{\cXa}^2\times\card{\cXb}^2}{2\eta} + 1}^{\card{\cXa}^2\times\card{\cXb}^2} $. 
The hyperedges in $ \cH $ are colored in the following way.
If an hyperedge $ ((\vxa_{i_1}, \vxa_{i_2}), (\vxb_{j_1}, \vxb_{j_2})) $ (where $ 1\le i_1<i_2<M_1 $ and $ 1\le j_1<j_2\le M_2 $) satisfies $ \distinf{\tau_{\vxa_{i_1}, \vxa_{i_2}, \vxb_{j_1}, \vxb_{j_2}}}{\distraabb}\le\eta $ for some $ \distraabb\in\cN $, then we color this hyperedge by $ \distraabb $. 
Note that by the covering property of $ \cN $, such a distribution must exist. 

By \Cref{lem:hypergraph_ramsey}, 
there exist subcodes $ (\cCa',\cCb') $ of $ (\cCa,\cCb) $ satisfying 
\begin{enumerate}
	\item $ M_1'\coloneqq|\cCa'|\ge N_1,M_2'\coloneqq|\cCb'|\ge N_2 $ for $ N_1 = N_1(M_1,M_2,D),N_2 = N_2(M_1,M_2,D) $ with $ N_1\xrightarrow{M_1\to\infty}\infty,N_2\xrightarrow{M_2\to\infty}\infty $;
	\item all hyperedges between $ \cCa' $ and $ \cCb' $ are monochromatic.
\end{enumerate}
In other words, according to the way we colored the hyperedges, there is a distribution $ \distraabb\in\cJab $ such that for all $ 1\le i_1<i_2\le M_1' $ and $ 1\le j_1<j_2\le M_2' $, we have $ \distinf{\tau_{\vxa_{i_1}, \vxa_{i_2}, \vxb_{j_1}, \vxb_{j_2}}}{\distraabb}\le\eta $. 
This completes the proof.
\end{proof}

In what follows, we will prove that the ``equicoupled'' subcode pair $ (\cCa',\cCb') $ must have at least one zero rate. 
We do so by treating separately the case where $ \distraabb $ is (almost) symmetric and the case where it is (significantly) asymmetric. 
We will actually show that\footnote{Hereafter we use the simplified notation $ M_1' = f(M_1) $ and $ M_2' = f(M_2) $ (where $ f(\cdot) $ is an increasing function) to emphasize the respective dependence of $ |\cCa'| $ and $ |\cCb'| $ on $|\cCa|$ and $ \cCb $, ignoring the dependence on other parameters. Indeed, noting $ M_1,M_2\ge1 $ and treating $ \card{\cXa},\card{\cXb},\eta $ as constants, one can take $ f(\cdot) = \min\curbrkt{f_1(\card{\cXa},\card{\cXb},\eta;\cdot,1), f_2(\card{\cXa},\card{\cXb},\eta;1,\cdot)} $ where $ f_1 $ and $ f_2 $ are from \Cref{lem:subcodepair_extraction}. } $ M_1' = f(M_1) \le C_1 $ or $ M_2' = f(M_2) \le C_2 $ for some constants (independent of $n$) $ C_1>0 $ and $ C_2>0 $. 
Since $f(\cdot)$ is a (slowly) increasing function, this implies that the original code pair $ (\cCa,\cCb) $ has sizes $ M_1\le f^{-1}(C_1) $ and $ M_2\le f^{-1}(C_2) $ which are still constants (though enormous). 
This is a stronger statement than that $ (\cCa,\cCb) $ have at least one zero rate.

\subsection{Asymmetric case}
\label{sec:converse_asymm_case}

\begin{definition}[Asymmetry and approximate symmetry]
\label{def:asymm}
The \emph{$ \curbrkt{1,2} $-asymmetry}, the \emph{$ \curbrkt{1} $-asymmetry}, the \emph{$ \curbrkt{2} $-asymmetry} and the \emph{asymmetry} of a distribution $ \distraabb \in\Delta(\cXa^2\times\cXb^2) $ is respectively defined as
\begin{align}
\asymmab(\distraabb) \coloneqq& \max_{(\xa_1, \xa_2)\in{\cXa}^{2}} \max_{(\xb_1, \xb_2)\in{\cXb}^{2}} \abs{ \distraabb(\xa_1, \xa_2, \xb_1, \xb_2) - \distraabb(\xa_2, \xa_1, \xb_2, \xb_1) }, \notag \\
\asymma(\distraabb) \coloneqq& \max_{(\xa_1, \xa_2)\in{\cXa}^{2}} \max_{(\xb_1, \xb_2)\in{\cXb}^{2}} \abs{ \distraabb(\xa_1, \xa_2, \xb_1, \xb_2) - \distraabb(\xa_2, \xa_1, \xb_1, \xb_2) }, \notag \\
\asymmb(\distraabb) \coloneqq& \max_{(\xa_1, \xa_2)\in{\cXa}^{2}} \max_{(\xb_1, \xb_2)\in{\cXb}^{2}} \abs{ \distraabb(\xa_1, \xa_2, \xb_1, \xb_2) - \distraabb(\xa_1, \xa_2, \xb_2, \xb_1) }, \notag \\
\asymm(\distraabb) \coloneqq& \max\curbrkt{\asymmab(\distraabb), \asymma(\distraabb), \asymmb(\distraabb)}. \notag 
\end{align}
A distribution $ \distraabb $ is called \emph{$ \alpha $-symmetric} if 
$\asymm(\distraabb) \le \alpha$. 
\end{definition}

\begin{remark}
By definition, a self-coupling $ \distraabb\in\cJab $ is in $ \cSab $ if and only if $ \asymm(\distraabb) = 0 $.
\end{remark}

According to \Cref{def:asymm}, the asymmetry of $\distraabb$ that was extracted in \Cref{lem:subcodepair_extraction} can be divided into eight different cases as shown in \Cref{tab:asymm} below. 
Case (1) in \Cref{tab:asymm} corresponds to the case where $ \distraabb $ is $\alpha$-symmetric. 
This case will be treated in \Cref{sec:converse_symm_case}. 
Other cases correspond to when $ \distraabb $ is asymmetric with asymmetry larger than $\alpha$. 
They will be treated in \Cref{sec:asymm-reduce-p2p,sec:asymm-78,sec:asymm-234}.

\begin{table}[htbp]
\centering
\begin{tabular}{ccccc}
\hline
Cases & $ \asymmab(\distraabb) \stackrel{?}{\le} \alpha $ & $ \asymma(\distraabb)\stackrel{?}{\le} \alpha $ & $ \asymmb(\distraabb)\stackrel{?}{\le} \alpha $ & Section \\ \hline 
Case (1) & $\le$ & $\le$ & $\le$ & \Cref{sec:converse_symm_case} \\
Case (2) & $>$ & $\le$ & $\le$ & \Cref{sec:asymm-234} \\
Case (3) & $\le$ & $>$ & $\le$ & \Cref{sec:asymm-234} \\
Case (4) & $\le$ & $\le$ & $>$ & \Cref{sec:asymm-234} \\
Case (5) & $>$ & $>$ & $\le$ & \Cref{sec:asymm-reduce-p2p} \\
Case (6) & $>$ & $\le$ & $>$ & \Cref{sec:asymm-reduce-p2p} \\
Case (7) & $\le$ & $>$ & $>$ & \Cref{sec:asymm-78} \\
Case (8) & $>$ & $>$ & $>$ & \Cref{sec:asymm-78} \\
\hline
\end{tabular}
\caption{The asymmetric case can be divided into several sub-cases.}
\label{tab:asymm}
\end{table}

For the asymmetric cases (Cases (5)-(8) in \Cref{tab:asymm}), we prove the following lemma. 

\begin{lemma}
\label{lem:joint-asymm}
If a code pair $ (\cCa',\cCb')\in\cXa^{M_1'\times n}\times\cXb^{M_2'\times n} $ satisfies that there exists a distribution $ \distraabb\in\cJab $ such that
\begin{enumerate}
	\item $ \cC_i $ is $ P_i $-constant composition for $ i = 1,2 $;
	\item for all $ 1\le i_1<i_2\le M_1' $ and $ 1\le j_1<j_2\le M_2' $, $ \distinf{\tau_{\vxa_{i_1},\vxa_{i_2},\vxb_{j_1},\vxb_{j_2}}}{\distraabb}\le\eta $;
	\item $ \asymm(\distraabb)\ge\alpha $,
\end{enumerate}
then at least one of $ M_1' $ and $ M_2' $ is at most a constant $ C(\alpha,\eta)>0 $. 
\end{lemma}

\begin{proof}
The proof is divided into several cases.
As we shall see in \Cref{sec:asymm-reduce-p2p,sec:asymm-78}, only Cases (5)-(8) in \Cref{tab:asymm} are asymmetric cases. 
Cases (2)-(4), handled in \Cref{sec:asymm-234}, can be reduced to the symmetric case (Case (1)). 
The symmetric Case (1) will be treated in the next section (\Cref{sec:converse_symm_case}). 
\end{proof}

The following lemma will be crucial in the proceeding subsections.
\begin{theorem}[\cite{komlos-1990-strange-pigeon-hole}]
\label{lem:komlos}
Let $ \bfv_1,\cdots,\bfv_M $ be a sequence of random variables over a common finite alphabet $\cW$. 
If there exist a distribution $ P_{\bfw_1,\bfw_2}\in\Delta(\cW^2) $ and a constant $ \eta\ge0 $ such that $ \norminf{P_{\bfv_i, \bfv_j} - P_{\bfw_1, \bfw_2}} \le \eta $ for all $ 1\le i<j\le M $, then 
$\asymm(P_{\bfw_1,\bfw_2}) \le {6}/{\sqrt{M}} + 4\sqrt{\eta} + 2\eta$, where
\begin{align}
\asymm(P_{\bfw_1, \bfw_2}) \coloneqq& \max_{(w_1, w_2)\in\cW\times\cW} \abs{P_{\bfw_1, \bfw_2}(w_1,w_2 ) - P_{\bfw_1, \bfw_2}(w_2, w_1 )}. \notag 
\end{align}
\end{theorem}

\subsubsection{Cases (5) \& (6) in \Cref{tab:asymm}}
\label{sec:asymm-reduce-p2p}
We only prove Case (5) since Case (6) is the same up to change of notation. 
We will show that $ M_1'\coloneqq|\cCa'| $ is at most a constant. 

We identify $ \distraabb $ with $ P_{\bfxa_1, \bfxa_2, \bfz^2} $ where $ \bfz^2 = (\bfxb_1, \bfxb_2) $ is a random variable over $ \cZ_2 \coloneqq \cXb^2 $. 
Immediately, $ \alpha<\asymma(\distraabb) = \asymma(P_{\bfxa_1, \bfxa_2, \bfz^2}) $ where $\asymma(\distraabb)$ is naturally defined as
\begin{align}
\asymma(P_{\bfxa_1, \bfxa_2, \bfz^2}) \coloneqq& \max_{(\xa_1, \xa_2)\in\cXa^2} \max_{z^2\in\cZ_2} \abs{P_{\bfxa_1, \bfxa_2, \bfz^2}(\xa_1, \xa_2, z^2) - P_{\bfxa_1, \bfxa_2, \bfz^2}(\xa_2, \xa_1, z^2)}. \notag 
\end{align} 
We then have the following simple lemma. 
\begin{lemma}
\label{lem:asymmw12}
If a distribution $ P_{\bfxa_1, \bfxa_2, \bfz^2} $ satisfies $ \asymma(P_{\bfxa_1, \bfxa_2, \bfz^2})>\alpha $, then 
$\asymm(P_{\bfw_1, \bfw_2}) > \alpha$, 
where $ \bfw_i \coloneqq (\bfxa_i, \bfz^2) \in \cW \coloneqq \cXa\times\cZ_2 $ for $ i=1,2 $.
\end{lemma}
\begin{proof}
The lemma follows from the following simple (in)equalities:
\begin{align}
\abs{P_{\bfw_1, \bfw_2}(w_1,w_2 ) - P_{\bfw_1, \bfw_2}(w_2, w_1 )} =& \abs{P_{(\bfxa_1, \bfz^2), (\bfxa_2, \bfz^2)}((\xa_1,z^2),(\xa_2,z^2) ) - P_{(\bfxa_1, \bfz^2), (\bfxa_2, \bfz^2)}((\xa_2,z^2), (\xa_1,z^2) )} \notag \\
=& \abs{ P_{\bfxa_1, \bfxa_2, \bfz^2}(\xa_1, \xa_2, z^2) - P_{\bfxa_1, \bfxa_2, \bfz^2}(\xa_2, \xa_1, z^2) } > \alpha. \qedhere 
\end{align}
\end{proof}

Recall $ M_1'\coloneqq\card{\cCa'},M_2'\coloneqq\card{\cCb'} $. 
By equicoupledness, for any fixed $ 1\le j_1< j_2<M_2' $, we have 
$\distinf{\tau_{\vxa_{i_1}, \vxa_{i_2}, \vxb_{j_1}, \vxb_{j_2}}}{ \distraabb} \le \eta$ for all $ 1\le i_1< i_2 \le M_1' $. 
Identify the codewords $ \vxa_1,\cdots,\vxa_{M_1'} $ in $\cCa'$ together with $ \vxb_{j_1}, \vxb_{j_2}\in\cCb' $ with a sequence of random variables $ \bfchi_1,\cdots,\bfchi_{M_1'},\bfzeta^2 \in \cXa^{M_1'}\times\cZ_2 $. 
That is
$P_{\bfchi_1,\cdots,\bfchi_{M_1'},\bfzeta^2}
\coloneqq \tau_{\vxa_1,\cdots,\vxa_{M_1'}, (\vxb_{j_1}, \vxb_{j_2})}$. 
Arrange this sequence in the following way: $ \bfv_1, \cdots,\bfv_{M_1'} $ where $ \bfv_i = (\bfchi_i, \bfzeta^2)\in\cW \coloneqq \cXa\times\cZ_2 $. 
This sequence satisfies
$ \distinf{P_{\bfv_{i_1}, \bfv_{i_2}}}{ P_{\bfw_1,\bfw_2}} \le\eta $ for every $ 1\le i_1<i_2\le M_1' $. 
To see this, 
\begin{align}
\distinf{P_{\bfv_{i_1}, \bfv_{i_2}}}{ P_{\bfw_1,\bfw_2}} =& \distinf{P_{(\bfchi_{i_1}, \bfzeta^2), (\bfchi_{i_2}, \bfzeta^2)}}{ P_{(\bfxa_1, \bfz^2), (\bfxa_2, \bfz^2)}} \notag \\
=& \distinf{P_{\bfchi_{i_1}, \bfchi_{i_2}, \bfzeta^2}}{ P_{\bfxa_1, \bfxa_2, \bfz^2}} \notag \\
=& \distinf{\tau_{\vxa_{i_1}, \vxa_{i_2}, \vxb_{j_1}, \vxb_{j_2}}}{ \distraabb} \le \eta. \label[ineq]{eqn:asymm-second-assump}
\end{align}
\Cref{eqn:asymm-second-assump} is by the second assumption of \Cref{lem:joint-asymm}. 
Now by \Cref{lem:komlos} and \Cref{lem:asymmw12}, 
$\alpha < \asymm(P_{\bfw_1, \bfw_2}) \le 6/\sqrt{M_1'} + 4\sqrt{\eta} + 2\eta$, 
i.e., $ M_1'< 36/(\alpha-4\sqrt{\eta} - 2\eta)^2 $. 
This finishes the proof of this case. 



\subsubsection{Cases (7) \& (8) in \Cref{tab:asymm}}
\label{sec:asymm-78}

In both Cases (7) \& (8), it simultaneously holds that $ \asymma(\distraabb)>\alpha $ and $ \asymmb(\distraabb)>\alpha $. 
By the analysis of the previous case, we have $ M_1'<36/(\alpha - 4\sqrt{\eta} - 2\eta)^2 $ and $ M_2'<36/(\alpha - 4\sqrt{\eta} - 2\eta)^2 $.

\subsubsection{Cases (2)-(4) in \Cref{tab:asymm}}
\label{sec:asymm-234}

We apply the following lemma to handle Cases (2)-(4).

\begin{lemma}
\label{lem:asymm-reduction}
The following relations hold.
\begin{align}
\asymmab(\distraabb) \le& \asymma(\distraabb) + \asymmb(\distraabb), \label{eqn:ab-le-a-b} \\
\asymma(\distraabb) \le& \asymmab(\distraabb) + \asymmb(\distraabb), \label{eqn:a-le-ab-b} \\
\asymmb(\distraabb) \le& \asymmab(\distraabb) + \asymma(\distraabb). \label{eqn:b-le-ab-a}
\end{align}
\end{lemma}
\begin{proof}
The lemma is a simple consequence of the triangle inequality.
We only prove \Cref{eqn:ab-le-a-b}. 
\Cref{eqn:a-le-ab-b,eqn:b-le-ab-a} follow similarly.
\begin{align}
\asymmab(\distraabb) =& \max_{(\xa_1, \xa_2)\in{\cXa}^{2}} \max_{(\xb_1, \xb_2)\in{\cXb}^{2}} \abs{ \distraabb(\xa_1, \xa_2, \xb_1, \xb_2) - \distraabb(\xa_2, \xa_1, \xb_2, \xb_1) } \notag \\
\le& \max_{(\xa_1, \xa_2)\in{\cXa}^{2}} \max_{(\xb_1, \xb_2)\in{\cXb}^{2}} \left(\abs{ \distraabb(\xa_1, \xa_2, \xb_1, \xb_2) - \distraabb(\xa_1, \xa_2, \xb_2, \xb_1) } \right. \notag \\
& + \left.\abs{ \distraabb(\xa_1, \xa_2, \xb_2, \xb_1) - \distraabb(\xa_2, \xa_1, \xb_2, \xb_1) }\right) \notag \\
\le& \max_{(\xa_1, \xa_2)\in{\cXa}^{2}} \max_{(\xb_1, \xb_2)\in{\cXb}^{2}} \abs{ \distraabb(\xa_1, \xa_2, \xb_1, \xb_2) - \distraabb(\xa_1, \xa_2, \xb_2, \xb_1) } \notag \\
& + \max_{(\xa_1, \xa_2)\in{\cXa}^{2}} \max_{(\xb_1, \xb_2)\in{\cXb}^{2}} \abs{ \distraabb(\xa_1, \xa_2, \xb_2, \xb_1) - \distraabb(\xa_2, \xa_1, \xb_2, \xb_1) } \notag \\
=& \asymma(\distraabb) + \asymmb(\distraabb). \qedhere
\end{align}
\end{proof}

By \Cref{lem:asymm-reduction}, we can reduce Cases (2)-(4) to the symmetric case (Case (1)) with $\alpha$ replaced by $2\alpha$.
Indeed, in Case (2), 
\begin{align}
\alpha< \asymmab(\distraabb) \le& \asymma(\distraabb) + \asymmb(\distraabb) \le 2\alpha. \notag 
\end{align}
Cases (3) and (4) are similar. 

\subsubsection{Case (1) in \Cref{tab:asymm}} Case (1) is treated in the next section.

\subsection{Symmetric case}
\label{sec:converse_symm_case}

In the previous section, we showed that $ \distraabb $ associated to the subcode pair $ (\cCa',\cCb') $ must be approximately symmetric (in the sense of \Cref{def:asymm}) for both $ |\cCa'| $ and $ |\cCb'| $ to be large, regardless of the channel structure. 
Therefore, in this section we focus on the case where 
\begin{align}
\asymm(\distraabb)\le\alpha . \label{eqn:symm-assump}
\end{align} 

Though we assume $ \good = \emptyset $ in \Cref{itm:conv-pos-pos} of \Cref{thm:converse}, the set $ \goodab\setminus\cKab $ may or may not be empty (see \Cref{fig:g12-minus-k12-nonempty}). 
We treat these two cases separately in the subsequent two subsections (\Cref{sec:symm-g12-minus-k12-empty,sec:symm-g12-minus-k12-nonempty}).

\subsubsection{The case where $ \goodab\setminus\cKab =\emptyset $}
\label{sec:symm-g12-minus-k12-empty}
In this subsection, we show that if $\goodab\setminus\cKab = \emptyset$, then both $ M_1' $ and $ M_2' $ are bounded from above by a constant. 
Therefore, any good code pair $ (\cCa,\cCb) $ has rates $ R_1 = 0 $ and $ R_2 = 0 $. 
The geometry of various sets of distributions that are involved in the following proof is depicted in \Cref{fig:geom-converse}. 
\begin{figure}[htbp]
	\centering
	\includegraphics[width=0.8\textwidth]{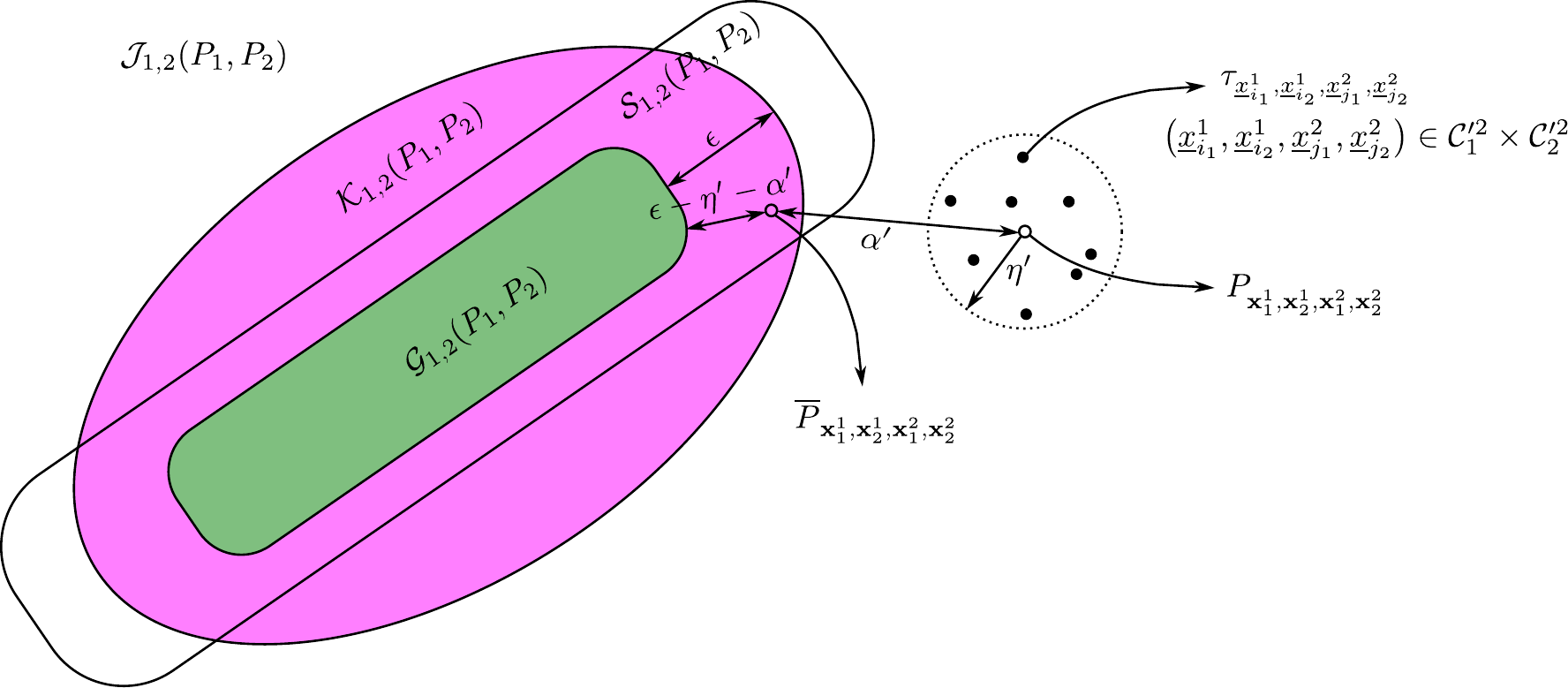}
	\caption{The geometry of various sets of distributions in the converse proof in \Cref{sec:symm-g12-minus-k12-empty}.  We assume $ \goodab\setminus\cKab = \emptyset $ and would like to show that any zero-error code pair has rate $ R_1=0 $ and $ R_2 = 0 $. In the above figure, the ambient space is the set $ \cJab $ of self-couplings equipped with $ \ell^1 $ metric. The set $ \goodab $ is a strict subset of $ \cKab $ such that they are $\eps$-separated (see \Cref{eqn:kab-gab-separation}). The joint types of the equicoupled subcode pair $ (\cCa',\cCb') $ are in an $ \eta' $-ball (see \Cref{eqn:equicoupled-lone}) around a distribution $ \distraabb $ which is assumed to be $ \alpha $-symmetric (see \Cref{eqn:symm-assump}). We then project $ \distraabb $ to obtain a symmetric distribution $ \oldistraabb $ defined in \Cref{eqn:define-oldistraabb}. (Note that $ \oldistraabb $ may be slight inside $ \cKab $.) It can be shown that $ \distraabb $ and $ \oldistraabb $ are $ \alpha' $-close (see \Cref{eqn:p12-olp12}). Since $ (\cCa',\cCb') $ attains zero error and all joint types are outside $ \cKab $, one can show that $ \oldistraabb $ is $ (\eps - \eta' - \alpha') $-far from $ \goodab $ (see \Cref{lem:oldistraabb-notin-goodab}). This allows us to proceed with the double counting argument. }
	\label{fig:geom-converse}
\end{figure}

We assume that $ \goodab $ is a \emph{proper} subset of $ \cKab $. 
Specifically, we assume that there exists a constant $ \eps>0 $ such that 
\begin{align}
\distone{\goodab}{\cJab\setminus\cKab} \ge\eps. 
\label{eqn:kab-gab-separation}
\end{align}
We first project $ \distraabb $ to $ \cSab $ and obtain an \emph{exactly} symmetric distribution $ \oldistraabb $, 
\begin{align}
\oldistraabb \coloneqq \frac{1}{4}\paren{\distraabb + P_{\bfxa_2, \bfxa_1, \bfxb_1, \bfxb_2} + P_{\bfxa_1, \bfxa_2, \bfxb_2, \bfxb_1} + P_{\bfxa_2, \bfxa_1, \bfxb_2, \bfxb_1}}. \label{eqn:define-oldistraabb} 
\end{align}
Since the four summands are all in $ \cJab $, $ \oldistraabb $ is also in $ \cJab $. 
Also, one can easily check that it is indeed symmetric in the sense of \Cref{def:sym_distr}. 
Furthermore, $ \oldistraabb $ and $ \distraabb $ are close to each other. 
\begin{align}
\distone{\oldistraabb}{\distraabb} =& \sum_{(\xa_1, \xa_2, \xb_1, \xb_2) \in \cXa^2\times\cXb^2} \abs{\oldistraabb(\xa_1, \xa_2, \xb_1, \xb_2) - \distraabb(\xa_1, \xa_2, \xb_1, \xb_2)} \notag \\
\le& \sum_{(\xa_1, \xa_2, \xb_1, \xb_2) \in \cXa^2\times\cXb^2} \frac{1}{4} \left( \abs{\distraabb(\xa_1, \xa_2, \xb_1, \xb_2) - P_{\bfxa_2, \bfxa_1, \bfxb_1, \bfxb_2}(\xa_1, \xa_2, \xb_1, \xb_2)} \right. \notag \\
&+ \abs{\distraabb(\xa_1, \xa_2, \xb_1, \xb_2) - P_{\bfxa_1, \bfxa_2, \bfxb_2, \bfxb_1}(\xa_1, \xa_2, \xb_1, \xb_2)} \notag \\
&+ \left. \abs{\distraabb(\xa_1, \xa_2, \xb_1, \xb_2) - P_{\bfxa_2, \bfxa_1, \bfxb_2, \bfxb_1}(\xa_1, \xa_2, \xb_1, \xb_2)} \right) \notag \\
\le& \frac{3}{4}|\cXa|^2|\cXb|^2\alpha \eqqcolon \alpha'. \label{eqn:p12-olp12} 
\end{align}
\Cref{eqn:p12-olp12} follows from the assumption given by \Cref{eqn:symm-assump}. 
Though we will not use it, the above bound can be slightly improved to $ \alpha' = \frac{1}{4}\paren{3|\cXa|^2|\cXb|^2 - |\cXa||\cXb|^2 - |\cXa|^2|\cXb| - |\cXa||\cXb|} $ by noting that some terms corresponding to $ \vxa_1 = \vxa_2 $ or $ \vxb_1 = \vxb_2 $ do not contributed to the sum.

\begin{claim}
\label{lem:oldistraabb-notin-goodab}
Under the assumptions of \Cref{sec:symm-g12-minus-k12-empty}, $ \oldistraabb $ is not in $ \goodab $:
\begin{align}
\distone{\oldistraabb}{\goodab} \ge& \eps - \eta' - \alpha', 
\label{eqn:olp12-notin-g12}
\end{align}
where $ \eta'\coloneqq\card{\cXa}^2\card{\cXb}^2\eta $ and $ \alpha' $ was defined in \Cref{eqn:p12-olp12}. 
\end{claim}

\begin{proof}
To prove the claim, first recall that for any $ 1\le i_1<i_2\le M_1' $ and $ 1\le j_1<j_2\le M_2' $, we have (by \Cref{fact:distinf-distone}) 
\begin{align}
\distone{\tau_{\vxa_{i_1}, \vxa_{i_2}, \vxb_{j_1}, \vxb_{j_2}}}{\distraabb} \le& \card{\cXa}^2\card{\cXb}^2 \distinf{\tau_{\vxa_{i_1}, \vxa_{i_2}, \vxb_{j_1}, \vxb_{j_2}}}{\distraabb}\le \card{\cXa}^2\card{\cXb}^2\eta \eqqcolon \eta'. \label{eqn:equicoupled-lone} 
\end{align}
Since $ (\cCa,\cCb) $ is a good code pair, $ (\cCa',\cCb') $ is also good. 
Hence $ \tau_{\vxa_{i_1}, \vxa_{i_2}, \vxb_{j_1}, \vxb_{j_2}} $ is not confusable, i.e.,  
\begin{align}
\tau_{\vxa_{i_1}, \vxa_{i_2}, \vxb_{j_1}, \vxb_{j_2}}\in\cJab\setminus\cKab. 
\label{eqn:non-confusable}
\end{align}
We get that $ \tau_{\vxa_{i_1}, \vxa_{i_2}, \vxb_{j_1}, \vxb_{j_2}} $ is strictly bounded away from $ \goodab $. 
\begin{align}
\distone{\tau_{\vxa_{i_1}, \vxa_{i_2}, \vxb_{j_1}, \vxb_{j_2}}}{\goodab}
\ge \distone{\goodab}{\cJab\setminus\cKab}
\ge \eps. 
\label{eqn:tauaabb-gab} 
\end{align}
The first inequality is by \Cref{eqn:non-confusable} and the second one follows from the assumption given by \Cref{eqn:kab-gab-separation}. 
\Cref{eqn:tauaabb-gab,eqn:equicoupled-lone} imply that $ \distraabb $ is strictly outside $ \goodab $.
\begin{align}
\distone{\distraabb}{\goodab} \ge& \distone{\tau_{\vxa_{i_1}, \vxa_{i_2}, \vxb_{j_1}, \vxb_{j_2}}}{\goodab} - \distone{\tau_{\vxa_{i_1}, \vxa_{i_2}, \vxb_{j_1}, \vxb_{j_2}}}{\distraabb} 
\ge \eps - \eta'. \label{eqn:dist-p12-g12} 
\end{align}
Combining \Cref{eqn:dist-p12-g12,eqn:p12-olp12}, we further have
\begin{align}
\distone{\oldistraabb}{\goodab} \ge& \distone{\distraabb}{\goodab} - \distone{\distraabb}{\oldistraabb} \ge \eps - \eta' - \alpha'. 
\notag
\end{align}
This finishes the proof of \Cref{lem:oldistraabb-notin-goodab}. 
\end{proof}

Since $ \oldistraabb\notin\goodab $ by \Cref{eqn:olp12-notin-g12}, we can apply \Cref{thm:duality}.
There exists $ Q_{\bfxa_1, \bfxa_2, \bfxb_1, \bfxb_2} \in \cogoodab $ such that 
\begin{align}
\inprod{\oldistraabb }{Q_{\bfxa_1, \bfxa_2, \bfxb_1, \bfxb_2} } \le& -\eps' <0 
\label{eqn:witness}
\end{align}
for some constant $ \eps'>0 $. 

To prove upper bounds on $ M_1' $ and $ M_2' $, the trick is to upper and lower bound the following quantity
\begin{align}
\sum_{(i_1, i_2) \in [M_1']\times[M_2']} \sum_{(j_1, j_2) \in [M_2'] \times[M_2']} \inprod{\tau_{\vxa_{i_1}, \vxa_{i_2}, \vxb_{j_1}, \vxb_{j_2}}}{Q_{\bfxa_1, \bfxa_2, \bfxb_1, \bfxb_2}}. \label{eqn:double_count_object}
\end{align}
Contrasting the upper and lower bounds on \Cref{eqn:double_count_object} will give us an upper bound on $ \max\curbrkt{M_1',M_2'} $. 
We first prove an upper bound on \Cref{eqn:double_count_object}.

\begin{claim}
\label{claim:joint-dc-upper}
\Cref{eqn:double_count_object} is at most 
\begin{align}
& \sum_{(i_1, i_2) \in [M_1']\times[M_2']} \sum_{(j_1, j_2) \in [M_2'] \times[M_2']} \inprod{\tau_{\vxa_{i_1}, \vxa_{i_2}, \vxb_{j_1}, \vxb_{j_2}}}{Q_{\bfxa_1, \bfxa_2, \bfxb_1, \bfxb_2}} \notag \\
\le& M_1'(M_1' - 1)M_2'(M_2'-1)(\eta' + \alpha' - \eps') + M_1'^2 M_2' + M_1' M_2'^2 + M_1' M_2'. \label{eqn:dc-upper-bound}
\end{align}
\end{claim}

\begin{proof}
Expanding the summation in \Cref{eqn:double_count_object}, we have 
\begin{align}
& \sum_{(i_1, i_2) \in [M_1']\times[M_1']} \sum_{(j_1, j_2) \in [M_2'] \times[M_2']} \inprod{\tau_{\vxa_{i_1}, \vxa_{i_2}, \vxb_{j_1}, \vxb_{j_2}}}{Q_{\bfxa_1, \bfxa_2, \bfxb_1, \bfxb_2}} \notag \\
=& \sum_{ \substack{(i_1, i_2, j_1, j_2) \in [M_1']^2\times[M_2']^2 \\ i_1 \ne i_2, j_1 \ne j_2} } + \sum_{\substack{(i_1, i_2, j_1, j_2) \in [M_1']^2\times[M_2']^2 \\ i_1 = i_2 \mathrm{\ or\ } j_1 = j_2}} \inprod{\tau_{\vxa_{i_1}, \vxa_{i_2}, \vxb_{j_1}, \vxb_{j_2}}}{Q_{\bfxa_1, \bfxa_2, \bfxb_1, \bfxb_2}} \notag \\
=& \sum_{ i_1 \ne i_2, j_1 \ne j_2 } + \sum_{i_1\ne i_2, j_1=j_2} + \sum_{i_1=i_2,j_1\ne j_2} + \sum_{i_1=i_2,j_1=j_2} \inprod{\tau_{\vxa_{i_1}, \vxa_{i_2}, \vxb_{j_1}, \vxb_{j_2}}}{Q_{\bfxa_1, \bfxa_2, \bfxb_1, \bfxb_2}} . \label{eqn:upper_bound_todo}
\end{align}
Note that 
\begin{align}
\inprod{\tau_{\vxa_{i_1}, \vxa_{i_2}, \vxb_{j_1}, \vxb_{j_2}}}{Q_{\bfxa_1, \bfxa_2, \bfxb_1, \bfxb_2}} \le& \normtwo{\tau_{\vxa_{i_1}, \vxa_{i_2}, \vxb_{j_1}, \vxb_{j_2}}} \normtwo{Q_{\bfxa_1, \bfxa_2, \bfxb_1, \bfxb_2}} \le \normone{\tau_{\vxa_{i_1}, \vxa_{i_2}, \vxb_{j_1}, \vxb_{j_2}}} \normone{Q_{\bfxa_1, \bfxa_2, \bfxb_1, \bfxb_2}} =1. \notag 
\end{align}
The last three terms in \Cref{eqn:upper_bound_todo} is at most
\begin{align}
M_1'^2 M_2' + M_1' M_2'^2 + M_1' M_2'. \label{eqn:last-three-terms-bound} 
\end{align}
The first term in \Cref{eqn:upper_bound_todo} can be bounded as follows. 
\begin{align}
& \sum_{ i_1 \ne i_2, j_1 \ne j_2 } \inprod{\tau_{\vxa_{i_1}, \vxa_{i_2}, \vxb_{j_1}, \vxb_{j_2}}}{Q_{\bfxa_1, \bfxa_2, \bfxb_1, \bfxb_2}} \notag \\
=& \sum_{ i_1 \ne i_2, j_1 \ne j_2 } \paren{\inprod{\tau_{\vxa_{i_1}, \vxa_{i_2}, \vxb_{j_1}, \vxb_{j_2}} - \oldistraabb}{Q_{\bfxa_1, \bfxa_2, \bfxb_1, \bfxb_2}} + \inprod{\oldistraabb}{Q_{\bfxa_1, \bfxa_2, \bfxb_1, \bfxb_2}}}. \label{eqn:upper_bound_distinct}
\end{align}
For any $ i_1\ne i_2 $ and $ j_1\ne j_2 $, the first term of the summand in \Cref{eqn:upper_bound_distinct} is at most
\begin{align}
\inprod{\tau_{\vxa_{i_1}, \vxa_{i_2}, \vxb_{j_1}, \vxb_{j_2}} - \oldistraabb}{Q_{\bfxa_1, \bfxa_2, \bfxb_1, \bfxb_2}} \le& \normone{\tau_{\vxa_{i_1}, \vxa_{i_2}, \vxb_{j_1}, \vxb_{j_2}} - \oldistraabb}\norminf{Q_{\bfxa_1, \bfxa_2, \bfxb_1, \bfxb_2}}  \label{eqn:apply-symm-12} \\
\le& \distone{\tau_{\vxa_{i_1}, \vxa_{i_2}, \vxb_{j_1}, \vxb_{j_2}} }{ \distraabb} + \distone{\distraabb}{\oldistraabb} \label{eqn:bound-q-linfty} \\
\le& \eta' + \alpha'. \label{eqn:first-term}
\end{align}
In \Cref{eqn:apply-symm-12},  we used the symmetry\footnote{The double counting argument crucially relies on the symmetry of $ \oldistraabb $ which is the reason why we treat the symmetric and asymmetric cases separately. } (as per \Cref{def:sym_distr}) of $ \oldistraabb $. 
Specifically, since $ \oldistraabb\in\cSab $, we have
\begin{align}
\distone{\tau_{\vxa_{i_2}, \vxa_{i_1}, \vxb_{j_2}, \vxb_{j_1}} }{ \oldistraabb} =& \distone{\tau_{\vxa_{i_1}, \vxa_{i_2}, \vxb_{j_1}, \vxb_{j_2}} }{ \ol{P}_{\bfxa_2,\bfxa_1,\bfxb_2,\bfxb_1} } = \distone{\tau_{\vxa_{i_1}, \vxa_{i_2}, \vxb_{j_1}, \vxb_{j_2}} }{ \oldistraabb}, \notag \\
\distone{\tau_{\vxa_{i_2}, \vxa_{i_1}, \vxb_{j_1}, \vxb_{j_2}} }{ \oldistraabb} =& \distone{\tau_{\vxa_{i_1}, \vxa_{i_2}, \vxb_{j_1}, \vxb_{j_2}} }{ \ol{P}_{\bfxa_2,\bfxa_1,\bfxb_1,\bfxb_2} } = \distone{\tau_{\vxa_{i_1}, \vxa_{i_2}, \vxb_{j_1}, \vxb_{j_2}} }{ \oldistraabb}, \notag \\
\distone{\tau_{\vxa_{i_1}, \vxa_{i_2}, \vxb_{j_2}, \vxb_{j_1}} }{ \oldistraabb} =& \distone{\tau_{\vxa_{i_1}, \vxa_{i_2}, \vxb_{j_1}, \vxb_{j_2}} }{ \ol{P}_{\bfxa_1,\bfxa_2,\bfxb_2,\bfxb_1} } = \distone{\tau_{\vxa_{i_1}, \vxa_{i_2}, \vxb_{j_1}, \vxb_{j_2}} }{ \oldistraabb}. \notag 
\end{align}
Hence, the bound in \Cref{eqn:first-term} holds for all $ i_1\ne i_2 $ and $ j_1\ne j_2 $ (not only for $ i_1<i_2 $ and $ j_1<j_2 $). 
In \Cref{eqn:bound-q-linfty}, we used the trivial bound $ \norminf{Q_{\bfxa_1, \bfxa_2, \bfxb_1, \bfxb_2}}\le1 $ since $ Q_{\bfxa_1, \bfxa_2, \bfxb_1, \bfxb_2} $ is a probability distribution. 
\Cref{eqn:first-term} is by \Cref{eqn:p12-olp12,eqn:equicoupled-lone}.
Combining \Cref{eqn:first-term,eqn:witness}, we get that the first term in \Cref{eqn:upper_bound_todo} is at most
\begin{align}
M_1'^2M_2'^2(\eta' + \alpha' - \eps'). \label{eqn:first-term-bound}
\end{align}
Overall, by \Cref{eqn:last-three-terms-bound,eqn:first-term-bound}, we get an upper bound on \Cref{eqn:double_count_object}:
\begin{align}
M_1'(M_1' - 1)M_2'(M_2'-1)(\eta' + \alpha' - \eps') + M_1'^2 M_2' + M_1' M_2'^2 + M_1' M_2', \notag
\end{align}
which completes the proof of \Cref{claim:joint-dc-upper}. 
\end{proof}

On the other hand, a lower bound on \Cref{eqn:double_count_object} follows from a direct calculation.

\begin{claim}
\label{claim:joint-dc-lower}
\Cref{eqn:double_count_object} is nonnegative, i.e., 
\begin{align}
\sum_{(i_1, i_2) \in [M_1']\times[M_2']} \sum_{(j_1, j_2) \in [M_2'] \times[M_2']} \inprod{\tau_{\vxa_{i_1}, \vxa_{i_2}, \vxb_{j_1}, \vxb_{j_2}}}{Q_{\bfxa_1, \bfxa_2, \bfxb_1, \bfxb_2}} \ge 0. 
\label{eqn:dc-lower-bound}
\end{align}
\end{claim}

\begin{proof}
We compute \Cref{eqn:double_count_object} from the first principle and interchange the summations.  

\begin{align}
& \sum_{(i_1, i_2) \in [M_1']\times[M_2']} \sum_{(j_1, j_2) \in [M_2'] \times[M_2']} \inprod{\tau_{\vxa_{i_1}, \vxa_{i_2}, \vxb_{j_1}, \vxb_{j_2}}}{Q_{\bfxa_1, \bfxa_2, \bfxb_1, \bfxb_2}} \notag \\
=& \sum_{(i_1, i_2)\in[M_1']^2}\sum_{(j_1, j_2)\in[M_2']^2} 
\sum_{(\xa_1, \xa_2)\in \cXa^2} \sum_{(\xb_1, \xb_2)\in \cXb^2}
\tau_{\vxa_{i_1}, \vxa_{i_2}, \vxb_{j_1}, \vxb_{j_2}} (\xa_1, \xa_2, \xb_1, \xb_2) Q (\xa_1, \xa_2, \xb_1, \xb_2) \notag \\
=& \sum_{\substack{(\xa_1, \xa_2)\in \cXa^2 \\ (\xb_1, \xb_2)\in \cXb^2}}
\sum_{\substack{(i_1, i_2)\in[M_1']^2 \\ (j_1, j_2)\in[M_2']^2}}
\frac{1}{n}\sum_{k\in[n]} \indicator{\vxa_{i_1}(k) = \xa_1} 
\indicator{\vxa_{i_2}(k) = \xa_2} 
\indicator{\vxb_{j_1}(k) = \xb_1} 
\indicator{\vxb_{j_2}(k) = \xb_2} 
Q (\xa_1, \xa_2, \xb_1, \xb_2) \notag \\
=& \frac{1}{n} \sum_{\substack{(\xa_1, \xa_2)\in \cXa^2 \\ (\xb_1, \xb_2)\in \cXb^2}}
\sum_{k\in[n]}
\paren{\sum_{i_1\in[M_1']}\indicator{\vxa_{i_1}(k) = \xa_1} }
\paren{\sum_{i_2\in[M_1']}\indicator{\vxa_{i_2}(k) = \xa_2} } \notag \\
& \paren{\sum_{j_1\in[M_2']}\indicator{\vxb_{j_1}(k) = \xb_1} }
\paren{\sum_{j_2\in[M_2']}\indicator{\vxb_{j_2}(k) = \xb_2} }
Q (\xa_1, \xa_2, \xb_1, \xb_2) \notag \\
=& \frac{M_1'^2M_2'^2}{n} \sum_{k\in[n]} 
\sum_{\substack{(\xa_1, \xa_2)\in \cXa^2 \\ (\xb_1, \xb_2)\in \cXb^2}}
\ipdistra^{(k)}(\xa_1)\ipdistra^{(k)}(\xa_2)\ipdistrb^{(k)}(\xb_1)\ipdistrb^{(k)}(\xb_2)
Q (\xa_1, \xa_2, \xb_1, \xb_2) \label{eqn:define-col-distr} \\
=& M_1'^2M_2'^2\inprod{\frac{1}{n}\sum_{k\in[n]} \paren{\ipdistra^{(k)}}^{\ot2} \ot \paren{\ipdistrb^{(k)}}^{\ot2}}{Q_{\bfxa_1, \bfxa_2, \bfxb_1, \bfxb_2}} \notag \\
\ge& 0. \label{eqn:dc-nonnegative}
\end{align}
In \Cref{eqn:define-col-distr}, $ P_i^{(k)} $ denotes the empirical distribution of the $k$-th \emph{column} of the codebook $ \cC_i'\in\cX_i^{M_i'\times n} $, i.e., for any $ x^i\in\cX_i $, 
\begin{align}
P_i^{(k)}(x^i) = \frac{1}{M_i'}\card{\curbrkt{\ell\in[M_i']:\vx^i_\ell(k) = x^i}}. \label{eqn:def-col-distr} 
\end{align}
\Cref{eqn:dc-nonnegative} follows from \Cref{thm:duality} since $ \frac{1}{n}\sum_{k\in[n]} \paren{\ipdistra^{(k)}}^{\ot2} \ot \paren{\ipdistrb^{(k)}}^{\ot2} \in\goodab $ and $ Q_{\bfxa_1,\bfxa_2,\bfxb_1,\bfxb_2}\in\cogoodab $.

This finishes the proof of \Cref{claim:joint-dc-lower}.
\end{proof}

Finally, \Cref{eqn:dc-upper-bound,eqn:dc-lower-bound} yield
\begin{align}
& & 0\le& M_1'^2M_2'^2(\eta'+\alpha'-\eps') + M_1'^2M_2'+M_1'M_2'^2+M_1'M_2' \notag \\
&\implies & 0\le& M_1'M_2'(\eta'+\alpha'-\eps') + M_1'+M_2'+1 \notag \\
&\implies & 0\le& -\delta M'^2 + 2M' + 1 \label{eqn:define-delta} \\
&\implies & M'\le& \frac{1+\sqrt{1+\delta}}{\delta}
\label{eqn:bound-max-m1-m2}
\end{align}
In \Cref{eqn:define-delta}, we let $ M'\coloneqq\max\curbrkt{M_1',M_2'} $ and $ \delta \coloneqq \eps' - \eta'-\alpha'>0 $. 
\Cref{eqn:bound-max-m1-m2} gives us the desired bound $ \max\curbrkt{M_1',M_2'}\le C $ for some constant $ C>0 $ independent of $n$. 

\subsubsection{The case where $ \goodab\setminus\cKab \ne \emptyset $}
\label{sec:symm-g12-minus-k12-nonempty}

Intuitively, this case is impossible for the following reasons.
In the last subsection, we have shown that for any of $ |\cCa'| $ and $ |\cCb'| $ to be large, the distribution $ \oldistraabb $ should (approximately) belong to $ \goodab\setminus\cKab $. 
For one thing, since $ \oldistraabb\in\gooda $, by the second property in \Cref{prop:properties_good_distr}, $ \sqrbrkt{\oldistraabb}_{\bfxa_1,\bfxa_2,\bfxb_1} $ (approximately) belongs to $ \gooda $ and $ \sqrbrkt{\oldistraabb}_{\bfxa_1,\bfxb_1,\bfxb_2} $ (approximately) belongs to $ \goodb $. 
For another thing, since the code pair $ (\cCa',\cCb') $ is assumed to attain zero error in the first place, we have that $ \sqrbrkt{\oldistraabb}_{\bfxa_1,\bfxa_2,\bfxb_1} $ which is close to $ \tau_{\vxa_{i_1}, \vxa_{i_2}, \vxb_{j_1}} $ is (approximately) outside $ \cKa $ and $ \sqrbrkt{\oldistraabb}_{\bfxa_1,\bfxb_1,\bfxb_2} $ which is close to $ \tau_{\vxa_{i_1}, \vxb_{j_1}, \vxb_{j_2}} $ is (approximately) outside $ \cKb $. 
In summary, we found a distribution $ \oldistraab\in\goodab\setminus\cKab $ with $ \sqrbrkt{\oldistraabb}_{\bfxa_1,\bfxa_2,\bfxb_1}\in\gooda\setminus\cKa $ and $ \sqrbrkt{\oldistraabb}_{\bfxa_1,\bfxb_1,\bfxb_2}\in\goodb\setminus\cKb $. 
This, nevertheless, contradicts the assumption $ \good = \emptyset $ of \Cref{itm:conv-pos-pos} in \Cref{thm:converse}.

The above intuition can be formalized by taking a good care of various slack factors. 
We flesh out the details below. 

In the previous section, we showed that for any constant $ \gamma>0 $,
if the distribution $ \oldistraabb $ (which is the symmetrized version of $ \distraabb $, as defined in \Cref{eqn:define-oldistraabb}) associated to $ (\cCa',\cCb') $ is $ \gamma $-far (in $ \ell^1 $ distance) from $ \goodab $, then both $ M_1' $ and $ M_2' $ are at most a constant $ g(\gamma) $ for some function $ g(\gamma)\xrightarrow{\gamma\to0}0 $.\footnote{
In the previous section, $ \gamma = \eps- \eta' - \alpha' $ as we got in \Cref{eqn:olp12-notin-g12} and $ g(\gamma) = g(\eps,\eta',\alpha') = \frac{1+\sqrt{1+\eps' - \eta' - \alpha'}}{\eps' - \eta' - \alpha'} $ where $ \eps' = \eps'(\eps) $ as we obtained in \Cref{eqn:bound-max-m1-m2}. } 
In other words, for $ M_1' $ or $ M_2' $ to be sufficiently large, $ \oldistraabb $ has to be $ \gamma $-close (in $ \ell^1 $ distance) to $ \goodab $ for an arbitrarily small constant $ \gamma>0 $.
Note also that unlike $ \tau_{\vxa_{i_1}, \vxa_{i_2}, \vxb_{j_1}, \vxb_{j_2}} $, the distribution $\oldistraabb$ can be \emph{slightly} inside $ \cKab $. 
However, it cannot be \emph{significantly} inside $ \cKab $. 
Specifically, for any $ 1\le i_1<i_2\le M_1' $ and $ 1\le j_1<j_2\le M_2' $, 
\begin{align}
& \distone{\oldistraabb}{\cJab\setminus\cKab} \notag \\
\le& \distone{\oldistraabb}{\distraabb} + \distone{\distraabb}{\tau_{\vxa_{i_1}, \vxa_{i_2}, \vxb_{j_1}, \vxb_{j_2}}} + \distone{\tau_{\vxa_{i_1}, \vxa_{i_2}, \vxb_{j_1}, \vxb_{j_2}}}{\cJab\setminus\cKab} \label{eqn:last-term-zero} \\
\le& \alpha' + \eta'. \label{eqn:symm-pab-kab} 
\end{align}
In \Cref{eqn:symm-pab-kab}, we used \Cref{eqn:p12-olp12,eqn:equicoupled-lone}. 
Also, the last term in \Cref{eqn:last-term-zero} is zero due to \Cref{eqn:non-confusable}. 
Overall, we have that for any good code pair $ (\cCa',\cCa')\in\cXa^{M_1'\times n}\times \cXb^{M_2'\times n} $ extracted from \Cref{lem:subcodepair_extraction}, for either $ M_1' $ or $ M_2' $ to be sufficiently large, $ \oldistraabb $ has to be $(\eps - \eta' - \alpha')$-close to $ \goodab $ and $ (\alpha'+\eta') $-close to $ \cJab\setminus\cKab $ for arbitrarily small constants $ \eps,\eta',\alpha'>0 $. 

Therefore, we can without loss of rigor drop these slack factors and assume for convenience  
\begin{align}
\oldistraabb\in\goodab\setminus\cKab. 
\label{eqn:convenient-assump}
\end{align}
For this to be possible, in this subsection we consider the case where $ \goodab \setminus \cKab \ne\emptyset $.

Let 
\begin{align}
\oldistraab \coloneqq& \sqrbrkt{\oldistraabb}_{\bfxa_1,\bfxa_2,\bfxb_1} = \sqrbrkt{\oldistraabb}_{\bfxa_1,\bfxa_2,\bfxb_2}, \notag \\
\oldistrabb \coloneqq& \sqrbrkt{\oldistraabb}_{\bfxa_1,\bfxb_1,\bfxb_2} = \sqrbrkt{\oldistraabb}_{\bfxa_2,\bfxb_1,\bfxb_2}. \notag 
\end{align}
Since $ \oldistraabb\in\goodab $, the equality of the respective marginals above is by the second property of \Cref{prop:properties_good_distr}. 
Furthermore, by \Cref{eqn:symm-pab-kab} and \Cref{lem:marg-doesnot-increase-dist}, 
\begin{align}
\distone{\oldistraab}{\cJa\setminus\cKa} \le& \alpha' + \eta', \label{eqn:dist-oldistraab-ka} \\
\distone{\oldistrabb}{\cJb\setminus\cKb} \le& \alpha' + \eta'. \label{eqn:dist-oldistrabb-kb}
\end{align}

We further divide the analysis into two cases (as shown in \Cref{fig:g12-minus-k12-nonempty}). 

\begin{figure}[htbp]
	\centering
	\begin{subfigure}[t]{0.45\linewidth}
		\centering
		\includegraphics[width=0.95\linewidth]{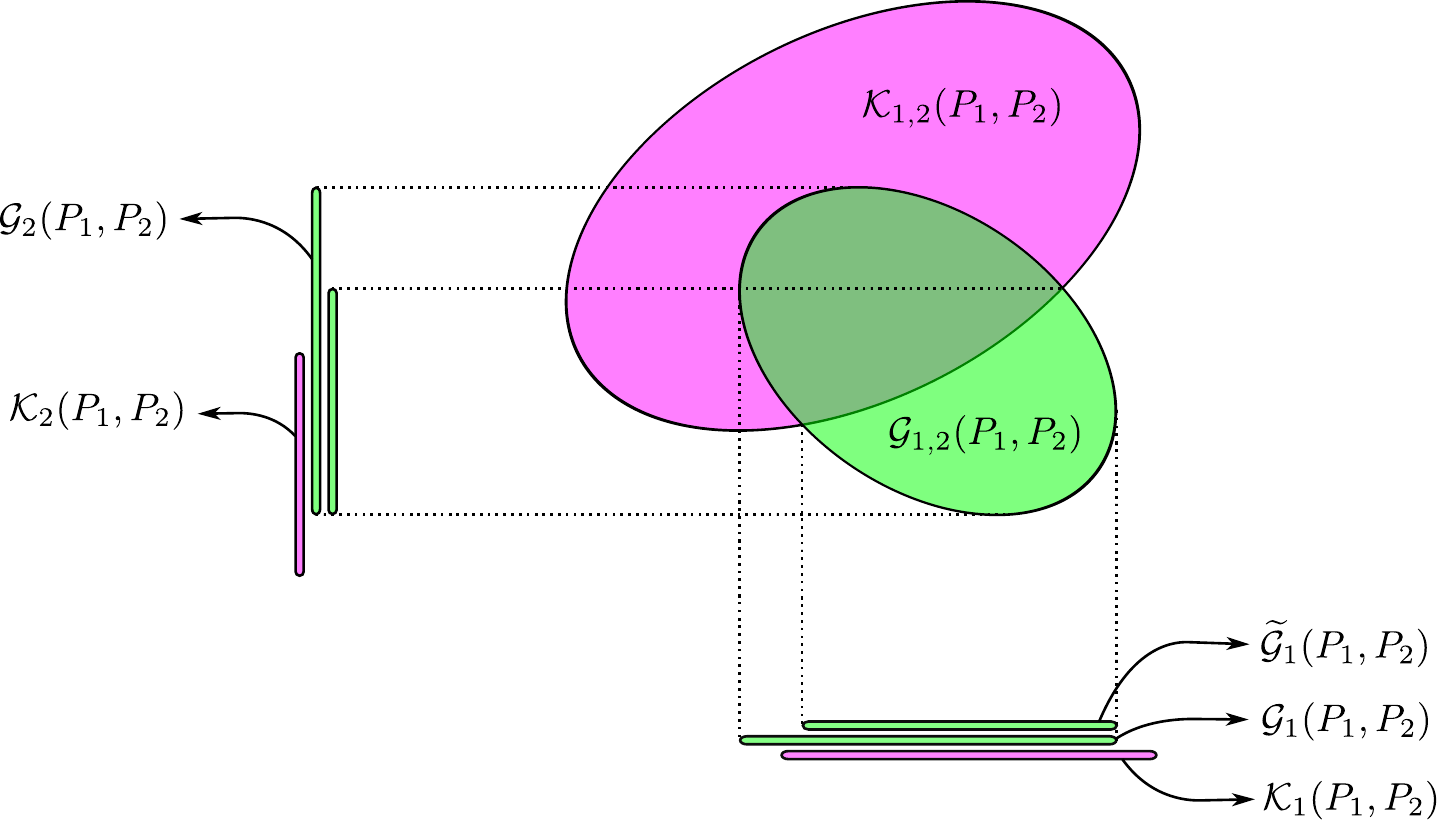}
		\caption{The case where $ \wtgooda\setminus\cKa = \emptyset $ where $ \wtgooda $ is defined in \Cref{eqn:def-wtgooda}. }
		\label{fig:joint-good_marg1-bad}
	\end{subfigure}
	\quad
	\begin{subfigure}[t]{0.45\linewidth}
		\centering
		\includegraphics[width=0.95\linewidth]{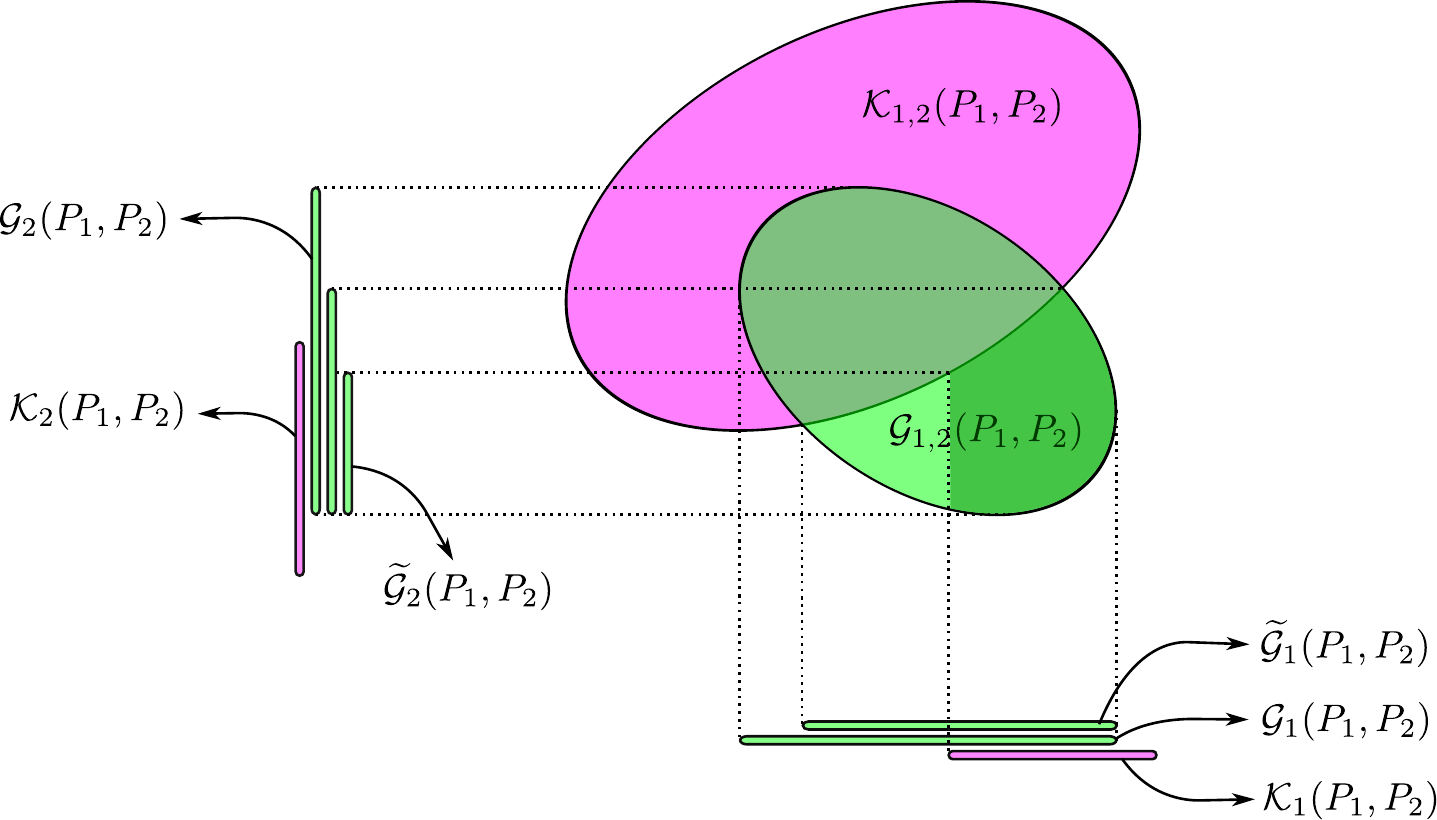}
		\caption{The case where $ \wtgooda\setminus\cKa \ne\emptyset $ while $ \wtgoodb\setminus\cKb = \emptyset $ where $ \wtgoodb $ is defined in \Cref{eqn:wtgoodb-alter-def}. }
		\label{fig:joint-good_marg1-good_marg2-bad}
	\end{subfigure}
	\caption{Under the assumptions $ \good = \emptyset $ and $ \goodab\setminus\cKab\ne\emptyset $, we further divide the converse analysis into two cases. The goal is to show that in these cases there do not exist zero-error code pairs of rates $ R_1>0 $ and $ R_2>0 $. In the above figures, {\color{magenta}pink} sets are confusability sets and {\color{forestgreen}green} sets are sets of good distributions. Two-dimensional sets are sets of joint distributions (e.g., $ \goodab,\cKab $) and one-dimensional sets are sets of marginal distributions (e.g., $ \gooda,\goodb,\cKa,\cKb $, etc.). }
	\label{fig:g12-minus-k12-nonempty}
\end{figure}

\begin{enumerate}
	\item Define
	\begin{align}
	\wtgooda \coloneqq& \curbrkt{\sqrbrkt{\distraabb}_{\bfxa_1,\bfxa_2,\bfxb_1}:\distraabb\in\goodab\setminus\cKab}\subseteq\gooda . \label{eqn:def-wtgooda} 
	\end{align}
	Note that by \Cref{eqn:convenient-assump}, 
	\begin{align}
	\oldistraab\in\wtgooda. 
	\label{eqn:oldistraab-in-wtgooda}
	\end{align}
	Since we assume $ \good = \emptyset $ in \Cref{itm:ach-1} of \Cref{thm:converse}, $ \wtgooda\setminus\cKa $ may or may not be empty. 
	We first handle the case where $ \wtgooda\setminus\cKa = \emptyset $. 
	In fact, let us assume 
	\begin{align}
	\distone{\wtgooda}{\cJa\setminus\cKa}\ge\eps_1 \label{eqn:wtgooda-assump}
	\end{align}
	for some $ \eps_1>0 $. 
	See \Cref{fig:joint-good_marg1-bad}. 
	However, \Cref{eqn:oldistraab-in-wtgooda,eqn:dist-oldistraab-ka} lead to a contradiction if $ \alpha' $ and $\eta'$ and sufficiently small so that $ \alpha' + \eta' < \eps_1 $. 
	Therefore, it is impossible for this case to happen. 

	\item 
	Now we assume 
	\begin{align}
	\wtgooda\setminus\cKa \ne \emptyset. 
	\label{eqn:wtgooda-assump-complement}
	\end{align}
	The analysis of the above case shows that $ \oldistraab\in\wtgooda $ has to be $ (\alpha' + \eta') $-close to $ \cJa\setminus\cKa $ for arbitrarily small $ \alpha' $ and $ \eta' $. 
	Similar to the assumption given by \Cref{eqn:convenient-assump}, in the present case we may as well assume for convenience
	\begin{align}
	\oldistraab\in\wtgooda\setminus\cKa. 
	\label{eqn:convenient-assump-1}
	\end{align}
	Now define
	\begin{align}
	\wtgoodb \coloneqq& \curbrkt{\sqrbrkt{\distraabb}_{\bfxa_1,\bfxb_1,\bfxb_2}:
	\begin{array}{rl}
	\distraabb\in& \goodab\setminus\cKab\\
	\sqrbrkt{\distraabb}_{\bfxa_1,\bfxa_2,\bfxb_1}\in& \gooda\setminus\cKa
	\end{array}
	} \notag \\
	=& \curbrkt{\sqrbrkt{\distraabb}_{\bfxa_1,\bfxb_1,\bfxb_2}:
	\begin{array}{rl}
	\distraabb\in& \goodab\setminus\cKab\\
	\sqrbrkt{\distraabb}_{\bfxa_1,\bfxa_2,\bfxb_1}\in& \wtgooda\setminus\cKa
	\end{array}
	}\subseteq\goodb . \label{eqn:wtgoodb-alter-def}
	\end{align}
	\Cref{eqn:wtgoodb-alter-def} is by \Cref{eqn:def-wtgooda}. 
	By the assumption given by \Cref{eqn:wtgooda-assump-complement}, $ \wtgoodb\ne\emptyset $. 
	Note that by \Cref{eqn:convenient-assump,eqn:convenient-assump-1}, 
	\begin{align}
	\oldistrabb\in\wtgoodb. 
	\label{eqn:oldistrabb-in-wtgoodb}
	\end{align}
	On the other hand, by the assumption $ \good = \emptyset $ and \Cref{eqn:wtgooda-assump-complement}, $ \wtgoodb\setminus\cKb $ must be empty. 
	In fact let us assume 
	\begin{align}
	\distone{\wtgoodb}{\cJb\setminus\cKb} \ge \eps_2 \label{eqn:wtgoodb-assump} 
	\end{align}
	for some $ \eps_2>0 $. 
	See \Cref{fig:joint-good_marg1-good_marg2-bad}. 
	Now by \Cref{eqn:dist-oldistrabb-kb,eqn:oldistrabb-in-wtgoodb}, we again reach a contradiction if $ \alpha+ \eta' < \eps_2 $. 
	Therefore, this case is also impossible to happen. 
\end{enumerate}

\section{Converse, \Cref{itm:conv-pos-zero,itm:conv-zero-pos} in \Cref{thm:converse}}
\label{sec:conv-pos-zero}

In this section, we only prove \Cref{itm:conv-pos-zero}. 
\Cref{itm:conv-zero-pos} follows by interchanging notation. 
We assume that $ \gooda\setminus\cKa = \emptyset $. 
More precisely, we assume 
\begin{align}
\distone{\gooda}{\cJa\setminus\cKa} \ge \eps \label{eqn:marg1-separation-assump}
\end{align}
for some $ \eps>0 $. 
Let $ (\cCa,\cCb) $ be any good code pair.
Suppose $ R_1>0 $. 
Our goal is to derive a contradiction. 

\subsection{Subcode extraction}
\label{sec:marginal-subcode-extraction}

\begin{theorem}[Ramsey's theorem \cite{ramsey}]
\label{thm:regular-ramsey}
Let $ \cK_M $ denote the (undirected) complete graph on $ M $ vertices. 
Let $ N\in\bZ_{\ge1}, D\in\bZ_{\ge2} $. 
Then there exists a constant $ K = K(N,D) $ such that for every $D$-coloring of the edges of $ \cK_M $ with $ M\ge K $, there is a monochromatic clique in $ \cK_M $ of size at least $N$. 
\end{theorem}

\begin{lemma}[Subcode extraction]
\label{thm:marginal-subcode-extraction}
Let $ (\cCa,\cCb)\subseteq\cXa^n\times\cXb^n $ be any $ (\ipdistra,\ipdistrb) $-constant composition code pair of sizes $ M_1,M_2 $, respectively. 
Let $ j\in[M_2] $. 
Then there exists a subcode $ \cCa'\subseteq\cC $ of size $ M_1'\ge f(\card{\cXa}, \card{\cXb}, \eta,  M_1)\xrightarrow{M_1\to\infty}\infty $ and a distribution $ \distraab\in\cJa $ such that for all $ 1\le i_1<i_1\le M_1' $, we have $ \distinf{\tau_{\vxa_{i_1}, \vxa_{i_2}, \vxb_j}}{\distraab}\le\eta $. 
\end{lemma}


\begin{proof}
The proof is similar to that of \Cref{lem:subcodepair_extraction} and follows readily from \Cref{thm:regular-ramsey}. 
We first build a complete graph $ \cK_{M_1} $ whose vertex set is $ \cCa $. 
We then color the edges of $ \cK_{M_1} $ using distributions in $ \cJa $.
Let $ \cN $ be an $ \eta $-net of $ \cJa $ of size at most $ |\cN|\le \paren{\frac{\card{\cXa}^2\times\card{\cXb}}{2\eta} + 1}^{\card{\cXa}^2\times\card{\cXb}} \eqqcolon D $ (by \Cref{lem:bound-net}). 
An edge $ (\vxa_{i_1}, \vxa_{i_2}) $ ($ 1\le i_1<i_2\le M_1 $) is colored by a distribution $ \distraab\in\cN $ if $ \distinf{\tau_{\vxa_{i_1}, \vxa_{i_2}, \vxb_j}}{\distraab}\le\eta $. 
Now by \Cref{thm:regular-ramsey}, there is a monochromatic subcode $ \cCa'\subseteq\cCa $ of size at least $ M_1'\ge f(\card{\cXa}, \card{\cXb},  \eta,  M_1) $, where $ f(\card{\cXa},\card{\cXb},\eta,M_1)\xrightarrow{M_1\to\infty}\infty $. 
According to the way we colored the edges, this means that for all $ 1\le i_1<i_2\le M_1' $, $ \distinf{\tau_{\vxa_{i_1}, \vxa_{i_2}, \vxb_j}}{\distraab}\le\eta $. 
\end{proof}

Fix any $ j\in[M_2] $. 
By \Cref{thm:marginal-subcode-extraction}, there is a subcode $ \cCa' \subseteq \cCa $ of size $ M_1'\xrightarrow{M_1\to\infty}\infty $ such that for some distribution $ \distraab\in\cJa $, we have 
\begin{align}
\distinf{\tau_{\vxa_{i_1}, \vxa_{i_2}, \vxb_{j}}}{\distraab}\le\eta \label{eqn:marg1-equicoupled}
\end{align}
for all $ 1\le i_1<i_2\le M_1' $.
\Cref{eqn:marg1-equicoupled} implies, by \Cref{fact:distinf-distone}, that
\begin{align}
\distone{\tau_{\vxa_{i_1}, \vxa_{i_2}, \vxb_{j}}}{\distraab}\le \card{\cXa}^2\card{\cXb} \eta \eqqcolon\eta'. \label{eqn:marg1-equicoupled-lone}
\end{align}
In the following two sections (\Cref{sec:marginal-asymmetric-case,sec:marginal-symmetric-case}) we treat the cases where $ \distraab $ is (noticeably) asymmetric and (approximately) symmetric (in the sense of \Cref{def:sym_distr}) separately.

\subsection{Asymmetric case}
\label{sec:marginal-asymmetric-case}

Reusing the proof for Cases (5) \& (6) of \Cref{lem:joint-asymm} with $ \bfz^2 $ being $ \bfxb $ (instead of $ (\bfxb_1,\bfxb_2) $ as in \Cref{sec:converse_asymm_case}) and $ \bfzeta^2 $ corresponding to $ \vxb_j $ (instead of $ (\vxb_{j_1},\vxb_{j_2}) $ as in \Cref{sec:converse_asymm_case}), we get that $ \asymm(\distraab)\le\alpha $ as long as $ M_1'\ge36/(\alpha-4\sqrt{\eta} - 2\eta)^2 $.

\subsection{Symmetric case}
\label{sec:marginal-symmetric-case}
As we saw in the last section, for $ M_1' $ to be sufficiently large, $ \asymm(\distraab)\le\alpha $. 
Under such an approximate symmetry condition, we then pass to an \emph{exactly} symmetric distribution $ \oldistraab\in\cSa $ defined as
\begin{align}
\oldistraab \coloneqq& \frac{1}{2}\paren{\distraab + P_{\bfxa_2,\bfxa_1,\bfxb}}. \notag 
\end{align}
Furthermore,
\begin{align}
\distone{\oldistraab}{\distraab} =& \sum_{(\xa_1,\xa_2,\xb)\in\cXa^2\times\cXb} \abs{\oldistraab(\xa_1,\xa_2,\xb) - \distraab(\xa_1,\xa_2,\xb)} \notag \\
\le& \frac{1}{2} \sum_{(\xa_1,\xa_2,\xb)\in\cXa^2\times\cXb} \abs{P_{\bfxa_2,\bfxa_1,\bfxb}(\xa_1,\xa_2,\xb) - \distraab(\xa_1,\xa_2,\xb)} \notag \\
\le& \frac{1}{2}\card{\cXa}^2\card{\cXb}\alpha \eqqcolon\alpha'. \label{eqn:p-olp-marg1} 
\end{align}

To apply the duality theorem (\Cref{thm:duality}), we argue that $ \oldistraab $ is not in $ \gooda $. 
\begin{align}
\distone{\oldistraab}{\gooda} \ge& \distone{\distraab}{\gooda} - \distone{\distraab}{\oldistraab} \notag \\
\ge& \distone{\tau_{\vxa_{i_1}, \vxa_{i_2}, \vxb_j}}{\gooda} - \distone{\tau_{\vxa_{i_1}, \vxa_{i_2}, \vxb_j}}{\distraab} - \alpha' \label{eqn:apply-p-olp-marg1} \\
\ge& \distone{\gooda}{\cJa\setminus\cKa} - \eta' - \alpha' \label{eqn:apply-marg1-equicoupled} \\
\ge& \eps - \eta'-\alpha'. \label{eqn:apply-marg-sep-assump} 
\end{align}
\Cref{eqn:apply-p-olp-marg1} is by \Cref{eqn:p-olp-marg1}. 
\Cref{eqn:apply-marg1-equicoupled} is by \Cref{eqn:marg1-equicoupled-lone} and the fact that $ \tau_{\vxa_{i_1}, \vxa_{i_2}, \vxb_j}\notin\cKa $. 
\Cref{eqn:apply-marg-sep-assump} follows from \Cref{eqn:marg1-separation-assump}. 
By \Cref{thm:duality}, there exists $ Q_{\bfxa_1,\bfxa_2,\bfxb}\in\cogooda $, such that
\begin{align}
\inprod{\oldistraab}{Q_{\bfxa_1,\bfxa_2,\bfxb}} \le -\eps'
\label{eqn:witness1}
\end{align}
for some constant $ \eps'>0 $.
The strategy is to bound
\begin{align}
\sum_{(i_1,i_2)\in[M_1']^2} \inprod{\tau_{\vxa_{i_1}, \vxa_{i_2}, \vxb_j}}{Q_{\bfxa_1,\bfxa_2,\bfxb}}.
\label{eqn:dc1}
\end{align}
For an upper bound,
\begin{align}
& \sum_{(i_1,i_2)\in[M_1']^2} \inprod{\tau_{\vxa_{i_1}, \vxa_{i_2}, \vxb_j}}{Q_{\bfxa_1,\bfxa_2,\bfxb}} \notag \\
=& \sum_{\substack{(i_1,i_2)\in[M_1']^2 \\ i_1\ne i_2}} \inprod{\tau_{\vxa_{i_1}, \vxa_{i_2}, \vxb_j}}{Q_{\bfxa_1,\bfxa_2,\bfxb}} 
+ \sum_{i\in[M_1']} \inprod{\tau_{\vxa_i,\vxa_i,\vxb_j}}{Q_{\bfxa_1,\bfxa_2,\bfxb}} \notag \\
=& \sum_{\substack{(i_1,i_2)\in[M_1']^2 \\ i_1\ne i_2}} 
\paren{\inprod{\tau_{\vxa_{i_1}, \vxa_{i_2}, \vxb_j} - \oldistraab}{ Q_{\bfxa_1,\bfxa_2,\bfxb} } - 
\inprod{\oldistraab}{Q_{\bfxa_1,\bfxa_2,\bfxb}} }
+ \sum_{i\in[M_1']} \inprod{\tau_{\vxa_i,\vxa_i,\vxb_j}}{Q_{\bfxa_1,\bfxa_2,\bfxb}} \notag \\
\le& M_1'^2(\eta' + \alpha' - \eps') + M_1'. \label[ineq]{eqn:ub1} 
\end{align}
In the above \Cref{eqn:ub1}, besides \Cref{eqn:marg1-equicoupled-lone,eqn:p-olp-marg1,eqn:witness1}, we also used the fact that $ \oldistraab\in\cSa $ and hence by \Cref{def:sym_distr}
\begin{align}
\distone{\tau_{\vxa_{i_2}, \vxa_{i_1}, \vxb_j}}{\oldistraab} = \distone{\tau_{\vxa_{i_1}, \vxa_{i_2}, \vxb_j}}{\ol{P}_{\bfxa_2,\bfxa_1,\bfxb}} = \distone{\tau_{\vxa_{i_1}, \vxa_{i_2}, \vxb_j}}{\oldistraab}. \notag 
\end{align}

For a lower bound,
\begin{align}
& \sum_{(i_1,i_2)\in[M_1']^2} \inprod{\tau_{\vxa_{i_1}, \vxa_{i_2}, \vxb_j}}{Q_{\bfxa_1,\bfxa_2,\bfxb}} \notag \\
=& \sum_{(i_1,i_2)\in[M_1']^2} \sum_{(\xa_1,\xa_2,\xb)\in\cXa^2\times\cXb} \tau_{\vxa_{i_1}, \vxa_{i_2}, \vxb_j}(\xa_1,\xa_2,\xb) Q_{\bfxa_1,\bfxa_2,\bfxb}(\xa_1,\xa_2,\xb) \notag \\
=& \sum_{(i_1,i_2)\in[M_1']^2}\sum_{(\xa_1,\xa_2,\xb)\in\cXa^2\times\cXb}
\frac{1}{n}\sum_{k\in[n]} \indicator{\vxa_{i_1}(k) = \xa_1,\vxa_{i_2}(k) = \xa_2,\vxb_j(k) = \xb}Q_{\bfxa_1,\bfxa_2,\bfxb}(\xa_1,\xa_2,\xb) \notag \\
=& M_1'^2 \sum_{(\xa_1,\xa_2,\xb)\in\cXa^2\times\cXb}\frac{1}{n} \sum_{k\in[n]} \ipdistra^{(k)}(\xa_1)\ipdistra^{(k)}(\xa_2)\ipdistrb^{(k)}(\xb)Q_{\bfxa_1,\bfxa_2,\bfxb}(\xa_1,\xa_2,\xb) \label{eqn:marginal-define-col-distr} \\
=& M_1'^2\inprod{\frac{1}{n}\sum_{k\in[n]}\paren{\ipdistra^{(k)}}^{\ot2}\ot\ipdistrb^{(k)}}{Q_{\bfxa_1,\bfxa_2,\bfxb}} \ge 0. \label[ineq]{eqn:lb1}
\end{align}
In \Cref{eqn:marginal-define-col-distr}, $ \ipdistra^{(k)} $ denotes the empirical distribution of the $k$-th column of $ \cCa' $ as defined in \Cref{eqn:def-col-distr} for $ i = 1 $; $ \ipdistrb^{(k)} $ is the indicator distribution $ \ipdistrb^{(k)}(\xb) \coloneqq \indicator{\vxb_j(k) = \xb} $ for all $ \xb\in\cXb $. 
\Cref{eqn:lb1} is by duality (\Cref{thm:duality}). 

\Cref{eqn:ub1,eqn:lb1} jointly yield
\begin{align}
M_1'^2(\eta' + \alpha' - \eps') + M_1' \ge 0, \notag 
\end{align}
i.e., 
\begin{align}
M_1'\le& \frac{1}{\eps' - \eta' - \alpha'} . \notag 
\end{align}

\begin{remark}
The marginal cases (\Cref{itm:conv-pos-zero,itm:conv-zero-pos}) of \Cref{thm:converse} proved in this section do \emph{not} directly follow from the point-to-point results by Wang et al. \cite{wbbj-2019-omni-avc} in a black-box manner. 
Unlike in the achievability proof (see proofs of \Cref{itm:prod-pos-zero,itm:prod-zero-pos} of \Cref{lem:achievability-prod}, proofs of \Cref{itm:ach-2,itm:ach-3} of \Cref{thm:achievability} and proofs of \Cref{itm:rate-prod-pos-zero,itm:rate-prod-zero-pos} of \Cref{lem:inner_bound_prod_distr}), we cannot assume in a converse argument that a zero-rate codebook only contains one codeword. 
Indeed, a rateless code may contain subexponentially many codewords. 
Consequently, the adversary may leverage his knowledge of this small code and jam the communication in a potentially more malicious way than as if he was not aware of the existence of the small code (in which case the problem reduces to the point-to-point setting). 
Incorporating such strength of the adversary requires a more tender care of the converse argument as we did in this section. 

Finally, we reiterate the nontriviality of the marginal cases of MACs even given the point-to-point results. 
Indeed, similar issues also arise in the study of AVMACs (where the adversary is oblivious) -- another adversarial model that received more attention than ours over the past years.
The corner cases where exactly one of the transmitters has zero capacity was left as a gap in Ahlswede and Cai's paper \cite{ahlswedecai-1999-obli-avmac-no-constr}, though the point-to-point results \cite{ahlswede-1978-avc-no-constr,csiszar-narayan-it1988-obliviousavc} were known for long by then. 
The gap was later noticed by Wiese and Boche \cite{wiese-2012-avmac-coop} and recently filled by Pereg and Steinberg \cite{pereg-steinberg-2019-avmac}, more than twenty years after \cite{ahlswedecai-1999-obli-avmac-no-constr}. 
\end{remark}





\section{Concluding remarks and open problems}
\label{sec:concl-rmk-open-prob}

In the following remarks we reflect on the results we obtained and the techniques we leveraged in this paper, and interleave them with several promising/interesting open questions. 

\begin{enumerate}
	\item \label{itm:open-error-criterion}
	Another highly related yet different model that is not considered in this paper is the adversarial MACs with \emph{average} probability of error. 
	As briefly discussed in \Cref{rk:error-criteria}, even for \emph{stochastic} MACs, the capacity region exhibits different behaviours under average error criterion than maximum error criterion. 
	Therefore, we do not believe that average error criterion behaves the same (at least under deterministic encoding) as the maximum one (which is equivalent to the zero error criterion under deterministic encoding) under our omniscient \emph{adversarial} MAC model. 
	Characterizing the capacity positivity and proving inner and outer bounds on the capacity region with \emph{average} probability of error are left for future research.
	In contrast, for point-to-point AVCs, the capacity remains the same under average probability of error (with deterministic encoding) and maximum probability of error (with stochastic encoding) \cite{csiszar-narayan-it1988-obliviousavc}.

	\item \label{itm:open-det-sto-mac}
	For technical simplicity, this paper only handles deterministic MACs. 
	For general (potentially stochastic) MACs, 
	maximum error criterion is \emph{not} equivalent to zero error criterion (though they are for deterministic MACs). 
	Techniques along the lines of \cite{csiszar-korner-1981} are of relevance for extending our results to general adversarial MACs.

	\item \label{itm:open-multiuser-mac}
	It is possible to generalize our results on capacity positivity to $t$-user MACs with $ t>2 $, though the case analysis may become baroque.

	\item \label{itm:open-better-inner-bounds}
	We believe that the capacity inner bounds obtained in \Cref{lem:inner_bound_prod_distr} can be improved. 
	In particular, the expurgation method we employed is crude -- we expurgated one codeword from each user's codebook for \emph{every pair} of confusable pairs $ ((\vbfxa_{i_1}, \vbfxb_{j_1}), (\vbfxa_{i_2}, \vbfxb_{j_2})) $. 
	Noting that a pair of codewords $ (\vbfxa_{i_1}, \vbfxb_{j_1}) $ participates in $ \Theta(M_1M_2) $ many pairs $((\vbfxa_{i_1}, \vbfxb_{j_1}), (\vbfxa_{i_2}, \vbfxb_{j_2}))$, we might have over-expurgated a more-than-desired number of codewords. 
	We believe that more careful expurgation strategy may lead to improved inner bounds. 
	For example, in \cite{gu-2018-zero-error-mac}, a nontrivial lower bound for $t$-user binary adder MACs\footnote{One caveat is that Gu \cite{gu-2018-zero-error-mac} was dealing with $t$-user MACs in which all transmitters use the \emph{same} codebook. Such codes are also known as $ B_t $ codes. } was obtained by only expurgating codewords with \emph{minimal violation} of the zero error criterion. 
	A naive expurgation as ours does \emph{not} yield such a bound.

	\item\label{itm:open-tensorization} 
	In classical zero-error information theory where channels under consideration are non-adversarial (or equivalently, unconstrainedly adversarial under our framework), there is a well-known $n$-letter expression for the capacity of a general DMC with zero error. 
	The expression involves the independence number of the $n$-fold strong product of the confusability graph associated to the channel. 
	Similarly, the non-stochastic information theory framework initiated by Nair \cite{nair-2011-nonstoc-conf,nair-2013-nonstoc-jrnl} also provides multi-letter expressions in terms of non-stochastic information measures. 
	In our opinion, the availability of such formulas heavily relies on the unconstrainedness of the channel. 
	That is, viewed as an adversarial channel, the noise sequence $ \vs $ can take any value in $ \cS^n $. 
	Consequently, ``good codes tensorize'' in the sense that if $ \cC\subseteq\cX^n $ attains zero error then $ \cC\times\cC\subseteq\cX^{2n} $ also attains zero error\footnote{Here we think of the tensor product $ \cC\times\cC $ as the set of concatenated codewords of length-$2n$ with both length-$n$ components from $ \cC $.}. 
	Unfortunately, such a tensorization property is not true for channels with state constraints. 
	It can be easily seen that the adversary can allocate his power on the long codeword in a nonuniform manner so as to confuse the decoder. 
	Codes for the adversarial bitflip channel is a concrete counterexample.\footnote{
	Consider a bitflip channel which can arbitrarily flip $p$ fraction of bits in the transmitted sequence. 
	Let $ \cC\in\zon $ be a good code for this channel. 
	That is, the minimum distance of $ \cC $ is at least $ 2np $. 
	Then $ \cC\times\cC $ still has distance $2np$ while its length doubles. 
	This means that it can only correct a $p/2$ fraction of errors, no longer attaining zero error for the original channel with noise level $p$. 
	} 
	The possibility of obtaining tight $n$-letter expressions for the capacity of omniscient adversarial channels using our framework is left for future investigations.

	\item \label{itm:open-cooperation}
	Recall that our main theorem asserts that for the sake of capacity positivity, it suffices to only consider distributions corresponding to mixtures of i.i.d. random variables. 
	Achievability-wise, one can achieve positive rates, whenever possible, by sampling random codes using mixtures of product self-couplings, i.e., ``good'' distributions as per \Cref{def:good_distr}.
	Conversely, if one could not achieve positive rates using good distributions, then she/he cannot achieve them using \emph{any other distributions}. 
	In the above sense, the set of good distributions we introduced plays a fundamental role in understanding capacity thresholds.
	This brings a natural question of whether there exist scenarios where correlated distributions help enlarge the region of positive rates and are hence also fundamentally ``good''. 
	One feasible way of physically instantiating correlation between input distributions is to allow \emph{cooperation}. 
	There is a recent line of works on \emph{oblivious} adversarial MACs (i.e., the classical AVMAC model) with cooperation \cite{wiese-2011-compound-mac-coop,wiese-2012-avmac-coop,boche-2016-avmac-list-dec-coop,huleihel-steinberg-2017-coop}. 
	That is, two encoders are allowed to communicate through a rate-limited channel\footnote{Note that if the channel between the two encoders is rate-unbounded, then the MAC problem reduces to a point-to-point problem. }. 
	It is an interesting problem to examine the behaviour of MACs with cooperations under the \emph{omniscient} model.

	\item \label{itm:open-list-dec}
	It is an intriguing question to extend our results to list decoding with constant list sizes. 
	The list decoding problem for both (oblivious) AVCs \cite{hughes-1997-list-avc,sarwate-gastpar-2012-list-dec-avc-state-constr,boche-2018-list-dec-obli,hosseinigoki-kosut-2018-oblivious-gaussian-avc-ld,zhang-2020-tight} and AVMACs \cite{boche-2016-avmac-list-dec-coop,nitinawarat-2013-avmac,cai-2016-list-dec-obli-avmac,zhang-2020-listdec_obli_avmac} is well-studied. 
	There are also papers on combinatorial list decoding for special MACs \cite{d-2019-separable-list-dec-mac,shchukin-2016-list-dec-mac}, not mentioning a huge body of work on list decoding for bitflip channels. 
	However, zero-error list decoding for \emph{general omniscient} adversarial channels remains relatively uncharted until recently \cite{zhang-2020-generalized-listdec-itcs}. 
	One of the major technical challenges for MACs that is absent in the point-to-point case has to do with list configurations. 
	A list for MAC can be represented by a bipartite graph \cite{cai-2016-list-dec-obli-avmac,zhang-2020-listdec_obli_avmac}. 
	For a target list size $ L\in\bZ_{\ge2} $, the bipartite graph with $L$ edges corresponding to an $L$-list may have different ``shapes''. 
	Such complications call for delicate analysis.

	\item \label{itm:open-other-multiuser-channels}
	It is plausible that our framework, built upon the prior work \cite{wbbj-2019-omni-avc}, is eligible for tackling the capacity threshold problem of other adversarial multiuser channels, e.g., broadcast channels, interference channels, relay channels, etc. 
	We leave this for further exploration. 
	The non-adversarial/unconstrained version of these problems has been considered by Devroye \cite{devroye-2016-zero-positive}.

	\item \label{itm:open-myopic}
	Motivated by the situation where the fundamental limit of oblivious MAC is well-understood \cite{pereg-steinberg-2019-avmac} while that of the omniscient counterpart is out of reach of the current techniques, it is tempting to study an intermediate model which interpolates between the oblivious and the omniscient models. 
	One model of this kind known as the \emph{myopic} channels was initiated by Sarwate \cite{sarwate-itw2010} and was advanced in a sequence of followup work \cite{dey-sufficiently-2015,budkuley-2020-symm-myop,zhang-2018-myop-isit}. 
	Despite the progress, even the capacity threshold of general point-to-point myopic channels is unknown.
	In the case of MAC, one natural definition of the myopic variant could be that the adversary gets to observe a noisy version of the transmitted sequence pair through a \emph{stochastic} (non-adversarial) MAC. 
	Such a model, as far as we know, remains unexplored.

	\item \label{itm:open-boundary-case}
	Strictly speaking, both our achievability and converse proofs rely on a \emph{strict} separation between the set of good distributions and the confusability set. 
	Specifically, we have to assume that the good set minus the confusability set has nonempty interior in the achievability proof; we have to assume that the good set is a proper subset of the confusability set in the converse proof. 
	The case where the good set \emph{kisses} the confusability set remains unsolved. 
	Such boundary cases are solved for some special channels including the (point-to-point) bitflip channel (see, e.g., \cite[Theorem 4.4.1]{essential-cod-thy-book}). 
	Similar subtleties also arise in the oblivious AVC/AVMAC setting where the boundary cases are in general open but are solved when the optimal jamming strategy is deterministic (which is the case, in particular, if the channel is deterministic) \cite{csiszar-narayan-it1988-obliviousavc,pereg-steinberg-2019-avmac}. 
	In all above solved cases, the capacity is zero at the boundary. 
	Inspired by these results, we conjecture that the capacity of our omniscient adversarial MACs is also zero in the boundary case. 
	That is, our converse can be (conjecturally) strengthened.

	\item \label{itm:open-large-alphabet}
	Our proof heavily relies on the assumption of finite alphabets. 
	It is unclear how to extend our proof to the case where the alphabet sizes grow with $n$. 
	In fact, we believe that the behaviour of the capacity (region) is significantly different in the large alphabet regime. 
	Indeed, for bitflip channels, there are algebraic constructions (notably the Reed--Solomon codes) attaining the capacity upper bound (the Singleton bound). 
	In other words, unlike in the small alphabet case, the first-order asymptotics of bitflip channels are known as long as the alphabet sizes are sufficiently large (in particular at least $ n$ suffices). 
	It remains an intriguing question to explore the behaviour of omniscient adversarial MACs in the large alphabet regime. 
	
	\item \label{itm:open-post-plotkin}
	Our converse results (\Cref{thm:converse}) give upper bounds on the size of codes when the channel does not admit positive rates. 
	For instance, if the set of good distributions is {``$\eps$-contained''} (as per \Cref{eqn:kab-gab-separation}) in the confusability set, then our proof gives $ \max\curbrkt{\card{\cCa},\card{\cCb}}\le f(1/\eps) $ which is independent of $n$. 
	However, the function $f(\cdot)$ involves Ramsey number and is therefore enormous. 
	We do not expect this bound to have an optimal dependence on $ 1/\eps $. 
	This type of question regarding the size of codes above the Plotkin bound was studied previously only for special channels. 
	For instance, for the (point-to-point) bitflip channels with noise level $p$, the optimal dependence is known to be $ \Theta(1/\eps) $ \cite{levenshtein-1961-hadamard-post-plotkin} where $ \eps = p - 1/4 $ is the gap between the Plotkin point and the noise level. 
	Optimal bounds are also known for list decoding over bitflip channels with odd\footnote{In \cite{abp-2018-listdec-above-plotkin}, the list size was parameterized by $ L-1 $ and optimal bounds were only shown for \emph{even} $L$, i.e., \emph{odd} list sizes.} list sizes \cite{abp-2018-listdec-above-plotkin}. 
	We are not aware of any result on codes above the Plotkin bound for adversarial MACs. 



\end{enumerate}



\section{Acknowledgement}
We thank Amitalok J. Budkuley and Sidharth Jaggi for many helpful discussions at the early stage of this work. 
We also thank Nir Ailon, Qi Cao and Chandra Nair for discussions on a related problem regarding zero-error binary adder MACs. 
This project has received funding from the European Union’s Horizon 2020 research and innovation programme under grant agreement No 682203-ERC-[Inf-Speed-Tradeoff]. 

\appendices

\section{Table of notation}
\label{app:table-of-notation}

Frequently used notation is listed in the following table (\Cref{tab:notation}).
\begin{center}
\begin{longtable}{lll} 
\hline
Notation & Meaning & Definition \\ 
\hline
$ \asymma(\cdot),\asymmb(\cdot),\asymmab(\cdot),\asymm(\cdot) $ & Asymmetry of a joint distribution & \Cref{def:asymm} \\
$ (\cCa,\cCb)\subseteq\cXa^n\times\cXb^n $ & Code pair & \Cref{def:codes} \\
$ \cogooda,\cogoodb,\cogoodab $ & Sets of co-good tensors with marginals $ (\ipdistra,\ipdistrb) $ & \Cref{def:cogood-distr} \\
$ \dec\colon\cY^n\to[M_1]\times[M_2] $ & Decoder of the receiver & \Cref{def:codes} \\
$ \enca\colon[M_1]\to\cXa^n ,\encb\colon[M_2]\to\cXb^n $ & Encoders of the transmitters & \Cref{def:codes} \\
$ \gooda,\goodb,\goodab $ & Sets of good distributions with marginals $ (\ipdistra,\ipdistrb) $ & \Cref{def:good_distr} \\
$ \good $ & Set of simultaneously good distributions with marginals $ (\ipdistra,\ipdistrb) $ & \Cref{def:good_distr} \\
$ \cJa,\cJb,\cJab $ & Sets of self-couplings with marginals $ (\ipdistra,\ipdistrb) $ & \Cref{def:self-coupling} \\
$ \jam\colon\cXa^n\times\cXb^n\to\cS^n $ & Jamming function of the adversary & \Cref{def:max-error} \\
$ \cKa,\cKb,\cKab $ & Confusability sets with marginals $ (\ipdistra,\ipdistrb) $ & \Cref{def:conf-set} \\

$ \mactwofull $ & Omniscient adversarial MAC & \Cref{def:omni-adv-mac} \\
$ (\ma,\mb)\in[M_1]\times[M_2] $ & Messages of the transmitters & \Cref{def:omni-adv-mac} \\
$ M_1 = |\cCa|,M_2 = |\cCb| $ & Sizes of codebooks & \Cref{def:codes} \\
$ \sqrbrkt{P_{\bfx,\bfy}}_{\bfx}\in\Delta(\cX) $ & Marginal distribution of $ P_{\bfx,\bfy}\in\Delta(\cX\times\cY) $ on the variable $ \bfx $ & \Cref{sec:notation} \\
$ (R_1,R_2) $ & Rate pair & \Cref{def:codes} \\
$ \vs\in\cS^n $ & Jamming sequence of the adversary & \Cref{def:omni-adv-mac} \\
$ \cS $ & Alphabet of the adversary & \Cref{def:omni-adv-mac} \\
$ \cSa,\cSb,\cSab $ & Sets of symmetric distributions with marginals $ (\ipdistra,\ipdistrb) $ & \Cref{def:sym_distr} \\
$ \syma,\symb,\symab $ & Sets of symmetric tensors with marginals $ (\ipdistra,\ipdistrb) $ & \Cref{def:sym_tensor} \\
$ W_{\bfy|\bfxa,\bfxb,\bfs} $ & Channel transition law & \Cref{def:omni-adv-mac} \\
$ (\vxa,\vxb)\in\cXa^n\times\cXb^n $ & Input sequences from the transmitters & \Cref{def:omni-adv-mac} \\
$ \cXa,\cXb $ & Alphabets of the transmitters & \Cref{def:omni-adv-mac} \\
$ \vy\in\cY^n $ & Output sequence to the receiver & \Cref{def:omni-adv-mac} \\
$ \cY $ & Alphabet of the receiver & \Cref{def:omni-adv-mac} \\

$ (\ipconstra,\ipconstrb)\subseteq\Delta(\cXa)\times\Delta(\cXb) $ & Input constraints & \Cref{def:omni-adv-mac} \\
$ \Delta(\cX) $ & Probability simplex on $ \cX $ & \Cref{sec:notation} \\  
$ \Da,\Db,\Dab $ & Sets of generalized self-couplings with marginals $ (\ipdistra,\ipdistrb) $ & \Cref{def:gen-self-coupling} \\ 
$ \Delta^{(n)}(\cX) $ & Sets of types of $ \cX^n $-valued vectors & \Cref{def:type} \\
$ \stconstr\subseteq\Delta(\cS) $ & State constraints & \Cref{def:omni-adv-mac} \\
$ \nu(P_\bfx,n) $ & -- & \Cref{eqn:def-poly} \\
$ \tau_{\vx}\in\Delta^{(n)}(\cX) $ & Type of $ \vx\in\cX^n $ & \Cref{def:type} \\
\hline
\caption{Table of frequently used notation.} 
\label{tab:notation}
\end{longtable}
\end{center}




\section{Proof of Plotkin bound for binary noisy $\XOR$ MACs (\Cref{thm:plotkin_binary_noisy_xor_mac})}
\label[app]{app:pf_plotkin_binary_noisy_xor_mac}
\begin{proof}[Proof of \Cref{thm:plotkin_binary_noisy_xor_mac}]
Suppose $ p = 1/4 + \eps $ for some constant $ \eps>0 $. 
Let $ \paren{\cCa, \cCb} $ be a code pair which attains zero error on the binary noisy $ \XOR $ MAC. 
Let $ M_1 \coloneqq \card{\cCa}, M_2 \coloneqq \card{\cCb} $.
We will show that $ M_1M_2\le \nicefrac{1}{4\eps} + 1 $. 
To this end, inspired the classical Plotkin bound in coding theory, we estimate the following quantity
\begin{align}
\sum_{\paren{\vxa_1, \vxa_2, \vxb_1, \vxb_2}\in\cCa^2\times\cCb^2} \disth{\vxa_1\oplus\vxb_1}{\vxa_2\oplus\vxb_2}. 
\label[term]{eqn:dc_bin_noisy_add}
\end{align}

One the one hand, by the goodness of $ (\cCa,\cCb) $, as long as $ (\vxa_1, \vxb_1) \ne (\vxa_2, \vxb_2) $, we have $ \disth{\vxa_1\oplus\vxb_1}{\vxa_2\oplus\vxb_2}>2np $. 
For $ (\vxa_1, \vxb_1) = (\vxa_2, \vxb_2) $, the summand is apparently zero. 
Therefore, \Cref{eqn:dc_bin_noisy_add} is larger than $ (M_1^2M_2^2 - M_1M_2)\cdot2np $. 

On the other hand, we can expand \Cref{eqn:dc_bin_noisy_add} as follows. 
\begin{align}
& \sum_{\paren{\vxa_1, \vxa_2, \vxb_1, \vxb_2}\in\cCa^2\times\cCb^2} \disth{\vxa_1\oplus\vxb_1}{\vxa_2\oplus\vxb_2} \notag \\
=& \sum_{\paren{\vxa_1, \vxa_2, \vxb_1, \vxb_2}\in\cCa^2\times\cCb^2} \wth{\vxa_1\oplus\vxb_1\oplus\vxa_2\oplus\vxb_2} \notag \\
=& \sum_{\paren{\vxa_1, \vxa_2, \vxb_1, \vxb_2}\in\cCa^2\times\cCb^2} \sum_{(a_1, b_1, a_2, b_2)\in\cM} \sum_{j = 1}^n
	\indicator{\vxa_1(j) = a_1}\indicator{\vxb_1(j) = b_1}\indicator{\vxa_2(j) = a_2}\indicator{\vxb_2(j) = b_2} 
	\label{eqn:odd_parity} \\
=& \sum_{j = 1}^n \sum_{(a_1, b_1, a_2, b_2)\in\cM} \paren{\sum_{\vxa_1\in\cCa}\indicator{\vxa_1(j) = a_1}}
	\paren{\sum_{\vxb_1\in\cCb}\indicator{\vxb_1(j) = b_1}} \paren{\sum_{\vxa_2\in\cCa}\indicator{\vxa_2(j) = a_2}}
	\paren{\sum_{\vxb_2\in\cCb}\indicator{\vxb_2(j) = b_2}} \notag \\
=& \sum_{j = 1}^n \big( (M_1 - S_j)(M_2 - T_j)(M_1 - S_j) T_j
	+ (M_1 - S_j)(M_2 - T_j)S_j(M_2 - T_j) \notag \\
	& + (M_1 - S_j)T_j(M_1 - S_j)(M_2 - T_j) 
	+ S_j(M_2 - T_j)(M_1 - S_j)(M_2 - T_j) \notag \\
	& + S_j T_j S_j (M_2 - T_j)
	+ S_j T_j (M_1 - S_j) T_j
	+ S_j (M_2 - T_j) S_j T_j
	+ (M_1 - S_j) T_j S_j T_j \big) \label{eqn:def_sj_tj} \\
=& M_1^2M_2^2\sum_{j = 1}^n \paren{\ol\alpha_j\ol\beta_j\ol\alpha_j \beta_j
	+ \ol\alpha_j\ol\beta_j\alpha_j\ol\beta_j 
	+ \ol\alpha_j\beta_j\ol\alpha_j\ol\beta_j 
	+ \alpha_j\ol\beta_j\ol\alpha_j\ol\beta_j 
	+ \alpha_j \beta_j \alpha_j \ol\beta_j
	+ \alpha_j \beta_j \ol\alpha_j \beta_j
	+ \alpha_j \ol\beta_j \alpha_j \beta_j
	+ \ol\alpha_j \beta_j \alpha_j \beta_j} \label{def:def_alphaj_beta_j} 
\end{align}
In \Cref{eqn:odd_parity}, we use $ \cM \coloneqq \curbrkt{0001,0010,0100,1000,1110,1101,1011,0111} $ to denote the set of length-4 binary sequences with odd parity. 
In \Cref{eqn:def_sj_tj}, we define $ S_j \coloneqq \sum_{\vxa\in\cCa}\indicator{\vxa(j) = 1} $ and $ T_j \coloneqq \sum_{\vxb\in\cCb}\indicator{\vxb(j) = 1} $ to be the number of 1's in the $j$-th column of $ \cCa \in \zo^{M_1\times n} $ and $ \cCb \in \zo^{M_2\times n} $ respectively.
In \Cref{def:def_alphaj_beta_j}, we further define $ \alpha_j \coloneqq S_j/M_1 $ and $ \beta_j \coloneqq T_j/M_2 $ to be the density of 1's in the $j$-th column of $\cCa$ and $ \cCb $ respectively; we also use the notation $ \ol a\coloneqq 1-a $ for $ a\in[0,1] $. 

For any $j\in[n]$, since $ \alpha_j, \beta_j\in[0,1] $ the summand of \Cref{def:def_alphaj_beta_j} is at most $ 1/2 $.
This can be verified by solving the following simple constrained (degree-4) polynomial optimization problem:
\begin{align}
\max_{(\alpha,\beta)\in[0,1]^2} {\ol\alpha\ol\beta\ol\alpha \beta
	+ \ol\alpha\ol\beta\alpha\ol\beta 
	+ \ol\alpha\beta\ol\alpha\ol\beta 
	+ \alpha\ol\beta\ol\alpha\ol\beta 
	+ \alpha \beta \alpha \ol\beta
	+ \alpha \beta \ol\alpha \beta
	+ \alpha \ol\beta \alpha \beta
	+ \ol\alpha \beta \alpha \beta}. \notag 
\end{align}
The maximum $1/2$ is attained at $ \alpha = 1/4,\beta = 1/2 $. 
Therefore, \Cref{eqn:dc_bin_noisy_add} is at most $ M_1^2M_2^2n/2 $. 

Putting the lower and upper bounds on \Cref{eqn:dc_bin_noisy_add} together, we have 
\begin{equation}
\begin{array}{rrl}
& \paren{M_1^2M_2^2 - {M_1M_2}}\cdot2np <& \frac{M_1^2M_2^2n}{2} \\
\iff& \paren{1 - \frac{1}{M_1M_2}}2\paren{\frac{1}{4} + \eps} <& \frac{1}{2} \\
\iff& M_1M_2 <& \frac{1}{4\eps} + 1, 
\end{array}\notag
\end{equation}
which finishes the proof of \Cref{thm:plotkin_binary_noisy_xor_mac}. 
\end{proof}


\bibliographystyle{alpha}
\bibliography{IEEEabrv,ref} 

\end{document}